\documentclass[11pt]{article}

\usepackage[matrix,arrow]{xy}

\usepackage{epsfig}

\usepackage{amssymb,amsthm}
\usepackage{amsmath}
\usepackage{amscd}
\usepackage{latexsym}
\usepackage{graphics}
\usepackage{color}

\usepackage{graphicx}

\input{xy}
\xyoption{all}

\newtheorem{theorem}[equation]{Theorem}
\newtheorem{lemma}[equation]{Lemma}
\newtheorem{proposition}[equation]{Proposition}

\catcode`@=11
\renewcommand{\section}
{\@startsection{section}{1}{0pt}{\medskipamount}{\medskipamount}{\large\bf}}
\makeatletter\renewcommand{\subsection}
{\@startsection{subsection}{2}{\z@}{-3.25ex plus -1ex minus -.2ex}
{1.5ex plus .2ex}{\it }}
\makeatletter\renewcommand{\subsubsection}
{\@startsection{subsubsection}{3}{\z@}{-3.25ex plus -1ex minus -.2ex}
{1.5ex plus .2ex}{\noindent\underline}}

\numberwithin{equation}{section}
\catcode`@=12

\def\={\ =\ }
\def\dd{{\rm d}}

\newcommand{\Tr}[1]{\:{\rm Tr}\,#1}
\def\e{{\,\rm e}\,}
\newcommand{\mbf}[1]{{\boldsymbol {#1} }}

\newcommand{\Vect}{{\sf Vect}}

\newcommand{\coh}{{\sf coh}}

\DeclareMathOperator{\Hom}{Hom}

\DeclareMathOperator{\End}{End}

\DeclareMathOperator{\Ext}{Ext}

\DeclareMathOperator{\charac}{ch}

\DeclareMathOperator{\Hilb}{Hilb}

\DeclareMathOperator{\ch}{ch}

\def\ii{{\,{\rm i}\,}}

\newcommand{\bbr}{\mathbb{R}}

\newcommand{\bbc}{\mathbb{C}}

\newcommand{\bbP}{\mathbb{P}}
\newcommand{\bbQ}{\mathbb{Q}}

\newcommand{\cN}{\mathcal{N}}

\newcommand{\cG}{\mathcal{G}}

\newcommand{\call}{\mathcal{L}}
\newcommand{\cale}{\mathcal{E}}
\newcommand{\calf}{\mathcal{F}}
\newcommand{\cals}{\mathcal{S}}
\newcommand{\cali}{\mathcal{I}}
\newcommand{\calo}{\mathcal{O}}
\newcommand{\calm}{\mathfrak{M}}
\newcommand{\caln}{\mathfrak{N}}
\newcommand{\calp}{\mathcal{P}}
\newcommand{\calq}{\mathcal{Q}}

\newcommand{\calE}{\mathcal{E}}

\newcommand{\Ocal}{\mathcal{O}}

\def\alg{{\mathcal A}}

\def\hil{{\mathcal H}}

\def\cK{{\mathcal K}}

\newcommand{\id}{{1\!\!1}} %% identity operator
\newcommand{\IZ}{\mathbb{Z}}

\newcommand{\F}{\mathbb{F}}

\newcommand{\cO}{{\cal O}}

\newcommand{\cR}{{\cal R}}
\newcommand{\cZ}{{\cal Z}}

\def\Id{{\rm id}}

\def\ch{{\rm ch}}
\newcommand{\Td}{{\rm Td}}

\newcommand{\bbz}{{\mathbb Z}}

\def\Dirac{{D\!\!\!\!/\,}} % Dirac operator
\def\e{\epsilon}

\def\beq{\begin{equation}}
\def\eeq{\end{equation}}
\def\bea{\begin{eqnarray}}
\def\eea{\end{eqnarray}}

\renewcommand{\e}{\,\mathrm{e}\,}
\newcommand{\im}{\,\mathrm{i}\,}

\newcommand{\R}{{\mathbb{R}}}
\newcommand{\N}{{\mathbb{N}}}
\newcommand{\C}{{\mathbb{C}}}
\newcommand{\Z}{{\mathbb{Z}}}
\newcommand{\Q}{{\mathbb{Q}}}
\newcommand{\PP}{{\mathbb{P}}}

\newcommand{\Ecal}{{\mathfrak E}}

\newcommand{\Pcal}{{\cal P}}
\newcommand{\Acal}{{\mathfrak A}}

\newcommand{\yb}{{\bar{y}}}

\newcommand{\Lcal}{{\cal L}}

\def\Dirac{{D\!\!\!\!/\,}} % Dirac operator

\def\Hom{{\rm Hom}}
\def\Ext{{\rm Ext}}
\def\End{{\rm End}}

\def\>{\rangle}
\def\<{\langle}
\def\+{\dagger}
\def\={\ =\ }

\begin{document}

\begin{titlepage}
\setcounter{page}{0}
\begin{flushright}
HWM--09--12\\
EMPG--09--21\\
\end{flushright}

\vskip 1.8cm

\begin{center}

{\Large\bf Instantons, Topological Strings
   and Enumerative Geometry}\footnote{Invited article for the special issue ``Nonlinear and Noncommutative Mathematics: New Developments and Applications in Quantum Physics'' of {\sl Advances in Mathematical Physics}.}

\vspace{15mm}

{\large Richard J. Szabo}
\\[5mm]
\noindent {\em Department of Mathematics\\ Heriot-Watt University\\
Colin Maclaurin Building, Riccarton, Edinburgh EH14 4AS, U.K. \\ and Maxwell Institute
  for Mathematical Sciences, Edinburgh, U.K.}
\\[5mm]
{Email: {\tt R.J.Szabo@ma.hw.ac.uk}}

\vspace{35mm}

\begin{abstract}

\noindent
We review and elaborate on certain aspects of the connections between
instanton counting in maximally supersymmetric gauge theories and the
computation of enumerative invariants of smooth varieties. We study in
detail three instances of gauge theories in six, four and two
dimensions which naturally arise in the context of topological string
theory on certain non-compact threefolds. We describe how the
instanton counting in these gauge theories are related to the
computation of the entropy of supersymmetric black holes, and how
these results are related to wall-crossing properties of enumerative
invariants such as Donaldson--Thomas and Gromov--Witten
invariants. Some features of moduli spaces of torsion-free sheaves and
the computation of their Euler characteristics are also elucidated.

\end{abstract}

\end{center}
\end{titlepage}

\tableofcontents

\bigskip

\section{Introduction\label{Intro}}

\noindent
Topological theories in physics usually relate BPS quantities to geometrical
  invariants of the underlying manifolds on which the physical theory is defined.  For the purposes of the present article, we shall focus on two particular and well-known instances of this. The first is instanton counting in supersymmetric gauge theories in four dimensions, which gives the
  Seiberg--Witten and Donaldson--Witten invariants of four-manifolds. 
The second is topological string theory, which is related to the enumerative geometry of Calabi--Yau three-folds and computes, for example, Gromov--Witten invariants, Donaldson--Thomas invariants, Gopakumar--Vafa BPS invariants, and key aspects of Kontsevich's homological mirror symmetry conjecture.

From a physical perspective, these topological models are not simply
of academic interest, but they also serve as exactly solvable systems
which capture the physical content of certain sectors of more
elaborate systems with local propagating degrees of freedom. Such is
the case for the models we will consider in this paper, which are
obtained as topological twists of a given physical theory. The
topologically twisted theories describe the BPS sectors of physical
models, and compute non-perturbative effects therein. For example, for
certain supersymmetric charged black holes, the microscopic
Bekenstein--Hawking--Wald entropy is computed by the Witten index of
the relevant supersymmetric gauge theory. This is equivalent to the counting of stable BPS bound states
  of D-branes in the pertinent geometry, and is related to invariants of threefolds via the OSV conjecture~\cite{OSV}.

From a mathematical perspective, we are interested in counting
invariants associated to moduli spaces of coherent sheaves on a smooth
complex projective variety $X$. To define such invariants, we need
moduli spaces that are varieties rather than algebraic stacks. The
standard method is to choose a polarization on $X$ and restrict
attention to semistable sheaves. If $X$ is a K\"ahler manifold, then a
natural choice of polarization is provided by a fixed K\"ahler
two-form on $X$. Geometric invariant theory then constructs a
projective variety which is a coarse moduli space for semistable
sheaves of fixed Chern character. In this paper we will be interested
in the computation of suitably defined Euler characteristics of
certain moduli spaces, which are the basic enumerative invariants. We
shall also compute more sophisticated holomorphic curve counting
invariants of a Calabi--Yau threefold $X$, which can be defined using
virtual cycles of the pertinent moduli spaces and are invariant under
deformations of $X$. In some instances the two types of invariants coincide.

An alternative approach to constructing moduli varieties is through
framed sheaves. Then there is a projective $Quot$ scheme which is a
fine moduli space for sheaves with a given framing. A framed sheaf can
be regarded as a geometric realization of an instanton in a
noncommutative gauge theory on $X$~\cite{DN,KS,Szabo1} which
asymptotes to a fixed connection at infinity. The noncommutative gauge
theory in question arises as the worldvolume field theory on a
suitable arrangement of D-branes in the geometry. In Nekrasov's
approach~\cite{Nekrasov}, the set of observables that enter in the
instanton counting are captured by the infrared dynamics of the topologically twisted gauge theory, and they compute the intersection theory of the (compactified) moduli spaces. The purpose of this paper is to overview the enumeration of such noncommutative instantons and its relation to the standard counting invariants of $X$.

In the following we will describe the computation of BPS indices of stable
  D-brane states via instanton counting in certain noncommutative and $q$-deformations of gauge theories on branes in various dimensions. We will pay particular attention to three non-compact examples which each arise in the context of Type~IIA string theory:
\begin{itemize}
\item[(1)] D6--D2--D0 bound states in D6-brane gauge theory \ --- \ These compute
  Donaldson--Thomas invariants and describe atomic configurations in a melting
  crystal model~\cite{ORV}. This also provides a solid example of a (topological) gauge theory/string theory duality. The counting of noncommutative instantons in the pertinent topological gauge theory is described in detail in ref.~\cite{CSS}.
\item[(2)] D4--D2--D0 bound states in D4-brane gauge theory \ --- \
  These count black hole microstates and allow us to probe the OSV
  conjecture. Their generating functions also appear to be intimately
  related to two-dimensional rational conformal field theory.
\item[(3)] D2--D0 bound states in D2-brane gauge theory \ --- \ These compute Gromov--Witten invariants of local curves. Instanton counting in the two-dimensional gauge theory on the base of the fibration is intimately related to instanton counting in the four-dimensional gauge theory obtained by wrapping supersymmetric D4-branes around certain non-compact four-cycles $C$, and also to the enumeration of flat connections in Chern--Simons theory on the boundary of~$C$. These interrelationships are explored in detail in refs.~\cite{CCGPSS1,CCGPSS2,CCGPSS3,GSST}.
\end{itemize}
These counting problems provide a beautiful hierarchy of relationships between topological string theory/gauge theory in six dimensions, four-dimensional supersymmetric gauge theories, Chern--Simons theory in three dimensions, and a certain $q$-deformation of two-dimensional Yang--Mills theory. They are also intimately related to two-dimensional conformal field theory.

\bigskip

\section{Topological string theory  \label{topstrings}}

\noindent
The basic setting in which to describe all gauge theories that we will
analyse in this paper within a unified framework is through
topological string theory, although many aspects of these models are
independent of their connection to topological strings. In this
section we briefly discuss some physical and mathematical aspects of
topological string theory, and how they naturally relate to the gauge
theories that we are ultimately interested in. Further details about
topological string theory can be found in e.g. ref.~\cite{Marino}, or
in ref.~\cite{Vonk} which includes a more general
introduction. Introductory and advanced aspects of
toric geometry are treated in the classic text~\cite{fulton} and in
the reviews~\cite{coxrev}. The standard reference for the sheaf theory
that we use is the book~\cite{HLbook}, while a more physicist-geared
introduction with applications to string theory can be found in the
review~\cite{Aspinwall1}.

\subsection{Topological strings and Gromov--Witten theory
  \label{GWtheory}}

Topological string theory may be regarded as a theory whose state space is a
``subspace'' of that of the full physical Type~II string theory. It is designed so that it can resolve the mathematical problem of
counting maps
\beq
f\,:\, \Sigma_g ~\longrightarrow~ X
\label{fSigmagX}\eeq
from a closed oriented Riemann surface
$\Sigma_g$ of genus $g$ into some target space $X$. In the physical
Type~II theory, any harmonic map $f$, with respect to chosen metrics
on $\Sigma_g$ and $X$, are allowed. They are described by solutions to
second order partial differential equations, the Euler--Lagrange
equations obtained from a variational principle based on a
sigma-model. The simplification supplied by topological strings is
that one replaces this sigma-model on the worldsheet $\Sigma_g$ by a
two-dimensional topological field theory, which can be realized as a
topological twist of the original $\cN=2$ superconformal field
theory. In this reduction, the state space descends to its BRST
cohomology defined with respect to the $\cN=2$ supercharges, which naturally carries a Frobenius algebra structure. This defines a consistent quantum theory if and only if the target space $X$ is a Calabi--Yau threefold, i.e. a complex K\"ahler manifold of dimension $\dim_\C(X)=3$ with  trivial canonical holomorphic line bundle $K_X$, or equivalently trivial first Chern class $c_1(X):=c_1(K_X)=0$. We fix a closed nondegenerate K\"ahler $(1,1)$-form $\omega$ on $X$.

The corresponding topological string amplitudes $F_g$ have interpretations in
compactifications of Type~II string theory on the product of the target space $X$ with four-dimensional Minkowski space $\R^{3,1}$. For instance, at genus zero the amplitude $F_0$ is the prepotential for vector multiplets of $\mathcal{N}\=2$ supergravity in four dimensions. The higher genus contributions $F_g$, $g\geq1$ correspond to higher derivative corrections of the schematic form $R^2\,T^{2g-2}$, where $R$ is the curvature and $T$ is the graviphoton field strength. We will now explain how to compute the amplitudes $F_g$. There are two types of topological string theories that we consider in turn.

\subsubsection*{A-model}

The A-model topological string theory isolates symplectic structure aspects of the Calabi--Yau threefold $X$. It is built on holomorphically embedded curves (\ref{fSigmagX}). The holomorphic string maps $f$ in this case are called \emph{worldsheet instantons}. They are classified topologically by their homology classes
$$
\beta\=f_*[\Sigma_g]~\in~H_2(X,\bbz) \ .
$$
With respect to a basis $S_i$ of two-cycles on $X$, one can write
\beq
\beta\=\sum_{i=1}^{b_2(X)}\,
n_i\,[S_i]
\label{betasum}\eeq
where the Betti number $b_2(X)$ is the rank of the second homology
group $H_2(X,\bbz)$, and $n_i\in\bbz$. Due to the topological nature
of the sigma-model, the string theory functional integral localizes
equivariantly (with respect to the BRST cohomology) onto contributions
from worldsheet instantons.

The sum over all maps can be encoded in a generating function $F_{\rm top}^X(g_s,\mbf Q)$ which depends on the string coupling $g_s$ and a vector of variables $\mbf Q=(Q_1,\dots,Q_{b_2(X)})$ defined as follows. Let
$$
t_i\=\int_{S_i}\,\omega
$$
be the complex K\"ahler parameters of $X$ with respect to the basis $S_i$. They appear in the values of the sigma-model action evaluated on a worldsheet instanton. For an instanton in curve class (\ref{betasum}), the corresponding Boltzmann weight is
$$
\mbf Q^\beta~:= ~ \prod_{i=1}^{b_2(X)}\,\big(Q_i\big)^{n_i} \qquad \mbox{with} \quad Q_i~:=~\e^{-t_i} \ .
$$
Then the quantum string theory is described by a genus expansion of the free energy
$$
F_{\rm top}^X(g_s,\mbf Q)\=\sum_{g=0}^\infty\,g_s^{2g-2}\,
F_g(\mbf Q) 
$$
weighted by the Euler characteristic $\chi(\Sigma_g)=2-2g$ of $\Sigma_g$, where the genus $g$ contribution to the statistical sum is given by
$$
F_g(\mbf Q)\=\sum_{\beta\in H_2(X,\bbz)}\,N_{g,\beta}(X) ~ \mbf Q^\beta \ ,
$$
and in this formula the classes $\beta\neq0$ correspond to worldsheets of genus $g$. The numbers $N_{g,\beta}(X)$ are called the \emph{Gromov--Witten invariants} of $X$ and they ``count'', in an appropriate sense, the number of worldsheet instantons (holomorphic curves) of genus $g$ in the two-homology class $\beta$. They can be defined as follows.

A worldsheet instanton (\ref{fSigmagX}) is said to be \emph{stable} if
the automorphism group ${\rm Aut}(\Sigma_g,f)$ is finite. Let
$\calm_g(X,\beta)$ be the (compactified) moduli space of isomorphism classes of
stable holomorphic maps (\ref{fSigmagX}) from connected genus $g$
curves to $X$ representing $\beta$. This is the \emph{instanton moduli
  space} onto which the path integral of topological string theory
localizes. It is a proper Deligne--Mumford stack over $\C$ which generalizes the moduli space $\calm_g$ of ``stable'' curves of genus $g$. While the dimension of $\calm_g$ is $3g-3$, the moduli space $\calm_g(X,\beta)$ is in general reducible and of impure dimension, as all possible stable maps occur. However, there is a perfect obstruction theory~\cite{GraPand} which generically has virtual dimension
$$
(1-g)\,\big(\dim_\C(X)-3\big)+\int_\beta\,c_1(X) \ .
$$
When $X$ is a Calabi--Yau threefold, this integer vanishes and there is a virtual fundamental class~\cite{GraPand}
$$
\big[\calm_g(X,\beta)\big]^{\rm vir}~\in~CH_0\big(\calm_g(X,\beta)\big)
$$
in the degree zero Chow group. In this case we define
$$
N_{g,\beta}(X)~:=~ \int_{[\calm_g(X,\beta)]^{\rm vir}}\,1 \ ,
$$
and so the Gromov--Witten invariants give the ``virtual numbers'' of
worldsheet instantons.
One generically has $N_{g,\beta}(X)\in\Q$ due to the orbifold nature of
the moduli stack $\calm_g(X,\beta)$. One can also define invariants by integrating the Euler class
of an \emph{obstruction bundle} over $\calm_g(X,\beta)$. There are precise recipes for computing the Gromov--Witten invariants $N_{g,\beta}(X)$ for \emph{toric} varieties $X$.

\subsubsection*{B-model}

The B-model topological string theory isolates complex structure aspects of the Calabi--Yau threefold $X$. It enumerates the \emph{constant} string maps which send the entire surface $\Sigma_g$ to a fixed point in $X$, and hence have trivial curve class $\beta=0$. The Gromov--Witten invariants in this case are completely understood. There is a natural isomorphism
$$
\calm_g(X,0)~\cong~ \calm_g\times X \ ,
$$
and the degree zero Gromov--Witten invariants $N_{g,0}(X)$ involve only the classical cohomology ring $H^\bullet(X,\bbQ)$ and ``Hodge integrals'' over the moduli spaces of Riemann surfaces $\calm_g$ defined as follows.

There is a canonical stack line bundle $\mathfrak{L}\to\calm_g$ with fibre $T^*_{\Sigma_g}$ over the moduli point $[\Sigma_g]$, the cotangent space of $\Sigma_g$ at some fixed point. We define the tautological class $\psi$ to be the first Chern class of $\mathfrak{L}$, $\psi:=c_1(\mathfrak{L})\in H^2(\calm_g,\bbQ)$. The \emph{Hodge bundle} $\mathfrak{E}\to\calm_g$ is the complex vector bundle of rank $g$ whose fibre over a point $\Sigma_g$ is the space $H^0(\Sigma_g,K_{\Sigma_g})$ of holomorphic sections of the canonical line bundle $K_{\Sigma_g}\to\Sigma_g$. Let $\lambda_j:=c_j(\mathfrak{E})\in H^{2j}(\calm_g,\bbQ)$. A \emph{Hodge integral} over $\calm_g$ is an integral of products of the classes $\psi$ and $\lambda_j$.

Explicit expressions for $N_{g,0}(X)$ for \emph{generic} threefolds $X$ are then obtained as follows. Let $\{\gamma_a\}_{a\in A}$ be a basis for $H^\bullet(X,\bbz)$ (modulo torsion), and let $D_2\subset A$ index the generators of degree two. Then one has
\bea
N_{0,0}(X)&=&\sum_{a_i\in A}\,\frac1{3!}~\int_X\,(\gamma_{a_1}\,\smile\,\gamma_{a_2}\,\smile\, \gamma_{a_3}) \ , \nonumber\\[4pt]
N_{1,0}(X)&=& - \sum_{a\in D_2}\,\frac1{24}~\int_X\,\gamma_a\,\smile\, c_2(X) \ , \nonumber\\[4pt]
N_{g\geq2,0}(X) &=& \frac{(-1)^g}2\,\int_X\,\big(c_3(X)-c_1(X)\,\smile\, c_2(X)\big)~ \int_{\calm_g}\,\lambda_{g-1}^3 \ , \nonumber
\eea
where the Hodge integral can be expressed in terms of Bernoulli numbers as
$$
\int_{\calm_g}\,\lambda_{g-1}^3 \= \frac{|B_{2g}|}{2g}\,\frac{|B_{2g-2}|}{2g-2}\, \frac1{(2g-2)!} \ .
$$
Note that $c_1(X)=0$ above when $X$ is Calabi--Yau.

Thus we know how to compute everything in the B-model, and it is completely under control. Our main interest is thus in extending these computations to the A-model. In analogy with the above considerations, one can note that there is a natural forgetful map
$$
\pi\,:\, \calm_g(X,\beta)~\longrightarrow~ \calm_g \ , \qquad (f,\Sigma_g)~\longmapsto~\Sigma_g \ ,
$$
and then reduce any integral over $\calm_g(X,\beta)$ to $\calm_g$ using the corresponding Gysin push-forward map $\pi_!$. However, this is difficult to do explicitly in most cases. The Gromov--Witten theory of $X$ is the study of tautological intersections in the moduli spaces $\calm_g(X,\beta)$. There is a string duality between the A-model and the B-model which is related to homological mirror symmetry.

\subsection{Open topological strings}

An open topological string in $X$ is described by a holomorphic
embedding $f:\Sigma_{g,h}\to X$ of a curve $\Sigma_{g,h}$ of genus $g$
with $h$ holes. A D-brane in $X$ is a choice of 
Dirichlet boundary condition on these string maps, which ensures that
the Cauchy problem for the Euler--Lagrange equations on $\Sigma_{g,h}$
locally has a unique solution. They correspond to \emph{Lagrangian} submanifolds $L$ of the Calabi--Yau threefold $X$, i.e. $\omega|_L=0$. If $\partial\Sigma_{g,h}=\sigma_1\cup\cdots\cup\sigma_h$, then we consider holomorphic maps such that
$$
f(\sigma_i)~\subset~L \ .
$$
This defines open string instantons, which are labelled by their \emph{relative} homology classes
$$
f_*[\Sigma_{g,h}]\= \beta~\in~ H_2(X,L) \ .
$$
If we assume that $b_1(L)=1$, so that $H_1(L,\Z)$ is generated
by a single non-trivial one-cycle $\gamma$, then
$$
f_*[\sigma_i]\=w_i\,\gamma \ ,
$$
where $w_i\in\Z$, $i=1,\dots,h$ are the winding numbers of the boundary maps $f|_{\sigma_i}$.

The free energy of the A-model open topological string theory at genus $g$ is given by
$$
F_{\mbf w,g}(\mbf Q)\=\sum_\beta\,N_{\mbf w,g,\beta}(X)~\mbf Q^\beta \ ,
$$
where $\mbf w=(w_1,\dots,w_h)$ and the numbers $N_{\mbf w,g,\beta}(X)$ are called \emph{relative Gromov--Witten invariants}. To incorporate all topological sectors, in addition to the string coupling $g_s$ weighting the Euler characteristics $\chi(\Sigma_{g,h})=2-2g-h$, we introduce an $N\times N$ Hermitean matrix $V$ to weight the different winding numbers. This matrix is associated to the holonomy of a gauge connection (Wilson line) on the D-brane. Then, taking into account that the holes are indistinguishable, the complete genus expansion of the generating function is
$$
F_{\rm top}^X(g_s,\mbf Q;V)\=\sum_{g=0}^\infty~\sum_{h=1}^\infty~\sum_{\mbf w\in\Z^h}\,\frac1{h!}\,g_s^{2g-2+h}\,F_{\mbf w,g}(\mbf Q)~ \prod_{i=1}^h\,\Tr\big(V^{w_i}\big) \ .
$$
The traces are computed by formally taking the limit $N\to\infty$ and expanding in irreducible representations $R$ of the D-brane gauge group $U(\infty)$.

\subsection{Black hole microstates and D-brane gauge theory}

When $X$ is a Calabi--Yau threefold, certain BPS black holes on
$X\times\bbr^{3,1}$ can be constructed by D-brane
engineering. D-branes in $X$ correspond to submanifolds of $X$
equipped with vector bundles with connection, the Chan--Paton gauge
bundles, and they carry \emph{charges} associated with the Chern
characters of these bundles. This data defines a class in the
differential K-theory of $X$, which provides a topological
classification of D-branes in $X$.

The microscopic black hole entropy can be computed by counting stable bound states of D0--D2--D4--D6 branes wrapping
holomorphic cycles of $X$ with the following configurations:
\begin{itemize}
\item D6-brane charge $\calq_6$;
\item D4-branes wrapping an ample divisor
\beq
[C]\= \displaystyle{\sum_{i=1}^{b_2(X)}\,\calq_4^i\, 
[C_i]~\in~H_4(X,\bbz)}
\label{amplediv}\eeq
with respect to a basis of four-cycles $C_i$,
$i=1,\dots,b_4(X)=b_2(X)$, of $X$;
\item D2-branes wrapping a two-cycle 
$$[S]\=\displaystyle{\sum_{i=1}^{b_2(X)}\,\calq_2^i\,
[S_i]~\in~H_2(X,\bbz)} \ ; $$
and
\item D0-brane charge $\calq_0$.
\end{itemize}
These D-brane charges give the black hole its charge quantum numbers. If we consider large enough numbers of D-branes in this system, then they form bound states which become large black holes with smooth event horizons, that can be counted and therefore account for the microscopic black hole entropy. In this scenario $p^I=(\calq_6,\calq_4^i)$ are interpreted as magnetic charges, and $q_I=(\calq_0,\calq_2^i)$ as electric charges.
The thermal partition function defined via a canonical ensemble for the
D0 and D2 branes with chemical potentials $\mu^I=(\phi^0,\phi_i^2)$, and a microcanonical
ensemble for the D4 and D6 branes, is given by
\beq
{Z_{\rm BH}(\calq_6,\mbf\calq_4,\mbf\phi^2,\phi^0)
\=\sum_{\calq_0,\calq_2^i}\,\Omega(\calq_0,\mbf\calq_2,
\mbf\calq_4,\calq_6)~\e^{-\calq_0\,\phi^0-\calq_2^i\,\phi_i^2}} \ ,
\label{BHpartfn}\eeq
where $\Omega$ is the degeneracy of BPS
states with the given D-brane charges. 

As we mentioned in Section~\ref{GWtheory}, the closed topological string amplitudes $F_g$ are related to supergravity quantities on Minkowski spacetime $\R^{3,1}$. The fact that the genus zero free energy $F_0$ for topological strings on $X$ is a prepotential for BPS black hole charges in $\cN=2$ supergravity determines the entropy $S_{\rm BH}(p,q)$ of an extremal black hole as a Legendre transformation of $F_0$, provided that one fixes the charge moduli by the attractor mechanism. The genus zero topological string amplitude $F_0$ is a homogeneous function of degree two in the $\cN=2$ vector multiplet fields $X^I$. The black hole entropy in the supergravity approximation is then
$$
S_{\rm BH}(p,q)\=\mu^I\,q_I-{\rm Im}\,F_0\big(X^I=p^I+\ii\mu^I\big) \ ,
$$
where the chemical potentials $\mu^I$ are determined by the charges
$p^I$ and $q_I$ by solving the equation
$$
q_I\= \frac{\partial\,{\rm Im}\,F_0}{\partial\mu^I} \ .
$$

Further analyses of the entropy of $\mathcal{N}=2$ BPS black holes on $\bbr^{3,1}$ have been extended to higher genus and suggest the relationship
\beq
{Z_{\rm BH}(\calq_6,\mbf\calq_4,\mbf\phi^2,\phi^0)
\=\big|Z_{\rm top}^X(g_s,\mbf Q)\big|^2
}
\label{OSV}\eeq
between the black hole partition function (\ref{BHpartfn}) and the topological string partition function $$Z_{\rm top}^X(g_s,\mbf Q)\=\exp F_{\rm top}^X(g_s,\mbf Q) \ , $$ where the 
moduli on both sides of this equation are related by their fixing at the attractor point
$$
g_s\=\frac{4\pi\ii}{\frac\ii\pi\,\phi^0+\calq_6} \ , \qquad
t_i\=-\frac{2\phi^2_i+\ii\calq_4^i}{\frac\ii\pi\,\phi^0+\calq_6} \ .
$$
The remarkable relationship (\ref{OSV}) is called the OSV conjecture~\cite{OSV}. It provides a means of using the perturbation expansion of topological strings and Gromov--Witten theory to compute black hole entropy to all orders. Alternatively, although the evidence for this proposal is derived for large black hole charge, the left-hand side of the expression (\ref{OSV}) makes sense for finite charges and in some cases is explicitly computable in closed form. It can thus be used to define {\it non-perturbative}
topological string amplitudes, and hence a non-perturbative completion of a string theory. 

In the following we will focus on the computation of the black hole partition function (\ref{BHpartfn}). The fact that this partition function is computable in a D-brane gauge theory will then give a physical interpretation of the enumerative invariants of $X$ in terms of black hole entropy. Suppose that we have a collection of D$p$-branes wrapping a submanifold $M_{p+1}\subset X$, with
$\dim_\R(M_{p+1})=p+1$ and Chan--Paton gauge field strength
$F$. D-branes are charged with respect to supergravity differential
form fields, the Ramond-Ramond fields, which are also classified
topologically by differential K-theory. Recall that such an array couples to all
$n$-form Ramond-Ramond fields $C_{(n)}$ through anomalous Chern--Simons couplings
$$
\int_{M_{p+1}}~\sum_{n\geq0}\,C_{(n)}\wedge\Tr\exp\big(
2\pi\,\alpha'\,F\big) \ ,
$$
where $\sqrt{\alpha'}$ is the string length. In particular, these couplings contain all terms
  $$\int_{M_{p+1}}\,C_{(p+1-2m)}\wedge\Tr\big(F^m\big) \ , $$ 
and so the topological charge $\ch_m(\cale)$ of a Chan--Paton gauge bundle $\cale\to M_{p+1}$ on a D$p$-brane is equivalent to D$(p-2m)$-brane
charge. A prominent example of this, which will be considered in detail later on, is the coupling $\int_{M_{p+1}}\,C_{(p-3)}\wedge\Tr\big(F\wedge 
  F\big)$. For $p=3$, this shows that the counting of D4--D0 brane bound states is
  equivalent to the enumeration of instantons on the four-dimensional part of the D4-brane in
  $X$. The remaining sections of this paper look at these relationships from the point of view of various BPS configurations of these D-branes. We will study the enumeration problems from the point of view of gauge theories on the D-branes in order of decreasing dimensionality, stressing the analogies between each description.

\bigskip

\section{D6-brane gauge theory and Donaldson--Thomas invariants}

\noindent
In this section we will look at a single D6-brane ($\calq_6=1$) and turn off
  all D4-brane charges ($\calq_4^i=0$). We will discuss various
  physical theories which are modelled by the D6-brane gauge theory in
  this case, but otherwise have no \emph{a priori} relation to string theory. These will include a tractable model for quantum gravity and the statistical mechanics of certain atomic crystal configurations. From the perspective of enumerative geometry, these partition functions will compute the Donaldson--Thomas theory of $X$.

\subsection{K\"ahler quantum gravity\label{KahlerQG}}

We will construct a model of quantum gravity on any K\"ahler threefold $X$, which will motivate the sorts of counting problems that we consider in this section. The partition function is defined by
\beq
Z\=\sum_{\stackrel{\scriptstyle\rm quantized}{\scriptstyle\omega}}\,
\e^{-S}
\label{Kahlerpartfn}\eeq
where
\beq
S\=\frac1{g_s^2}\,\int_X\,\frac1{3!}\,\omega\wedge\omega\wedge
\omega \ .
\label{Kahleraction}\eeq
The sum is over ``quantized'' K\"ahler two-forms on $X$, in the following sense. We decompose the ``macroscopic'' form $\omega$ into a fixed ``background'' K\"ahler two-form $\omega_0$ on $X$ and
the curvature $F$ of a holomorphic line bundle $\Lcal$ over $X$ as
\beq
\omega\=\omega_0+g_s\,F \ .
\label{omegadecomp}\eeq
To satisfy the requirement that there are no D4-branes in $X$, we impose the fluctuation condition
\beq
\int_\beta\,F\=0 
\label{fluctcond}\eeq
for all two-cycles $\beta\in H_2(X,\bbz)$.

Substituting (\ref{omegadecomp}) together with (\ref{fluctcond}) into (\ref{Kahleraction})
gives the action
\beq
S\=\frac1{g_s^2}\,\frac1{3!}\,\int_X\,\omega_0^3+\frac12\,\int_X\,
F\wedge F\wedge\omega_0+g_s\,\int_X\,\frac1{3!}\,F\wedge
F\wedge F \ .
\label{Kahleractionexp}\eeq
The statistical sum (\ref{Kahlerpartfn}) thus becomes (dropping an
irrelevant constant term)
\beq
Z\=\sum_{\stackrel{\scriptstyle\rm line}{\scriptstyle{\rm bundles} \ \Lcal}}\,
q^{\ch_3(\Lcal)}~\prod_{i=1}^{b_2(X)}\,\big(Q_i\big)^{\int_{C_i}\,\ch_2(\Lcal)} \ ,
\label{statsumline}\eeq
where $q=-\e^{-g_s}$, $Q_i=\e^{-t_i}$, and $\ch_m(\Lcal)$ denotes the $m$-th Chern character of the given line bundle $\Lcal\to X$. Note the formal similarity with the A-model topological string partition function constructed in Section~\ref{GWtheory} However, there is a problem with the way in which we have thus far set up this model. The fluctuation condition (\ref{fluctcond}) on $F$ implies
$\ch_2(\Lcal)=\ch_3(\Lcal)=0$. Hence only trivial line bundles can contribute to the sum (\ref{statsumline}), and the partition function is trivial.

The resolution to this problem is to enlarge the range of summation in (\ref{statsumline}) to include singular gauge fields and ideal sheaves. Namely, we take $F$ to correspond to a {\it singular} $U(1)$ gauge field $A$ on
$X$. This can be realized in two (related) possible ways:
\begin{itemize}
\item[(1)] We can make a singular gauge field $A$ non-singular on the blow-up
$$
\widehat{X}~\longrightarrow~X
$$
of the target space, obtained by blowing up the singular points of $A$ on $X$ into copies of the complex projective plane $\PP^2$. This means that the quantum gravitational path integral induces a topology change of the target space $X$. This is refered to as ``quantum foam'' in refs.~\cite{BN,INOV}; or
\item[(2)] We can relax the notion of line bundle to ideal sheaf. Ideal sheaves lift to line
  bundles on $\widehat{X}$. However, there are ``more'' sheaves on $X$ than blow-ups
  $\widehat{X}$ of $X$.
\end{itemize}

In this paper we will take the second point of view. Recall that torsion free sheaves $\mathcal{E}$ on $X$ can be defined by the property that they sit inside short exact sequences of
sheaves of the form
\beq
\xymatrix{
0~\ar[r] &~\cale~\ar[r] & ~\calf~\ar[r] & ~
\cals_Z~\ar[r] & ~0 \ ,
}
\label{torsionfreeseq}\eeq
where $\calf$ is a holomorphic vector bundle on $X$, and $\cals_Z$ is
a coherent sheaf supported at the singular
points $Z\subset X$ of a gauge connection $A$ of $\calf$. Applying the Chern character to (\ref{torsionfreeseq}) and using its additivity on exact sequences gives
$$
\ch_m(\cale)\=\ch_m(\calf)-\ch_m(\cals_Z)
$$
for each $m$. Thus torsion free sheaves $\cale$ fail to be vector
bundles at singular points of gauge fields, and including the singular locus can reinstate the non-trivial topological quantum numbers that we desired above.

As we will discuss in detail in this section, this construction is realized explicitly by considering a noncommutative gauge theory on the target space $X=\bbc^3$. We shall see that the instanton solutions of gauge theory on a noncommutative
  deformation $\bbc_\theta^3$ are described in terms of {\it ideals}
  $\mathcal{I}$ in the polynomial ring $\bbc[z^1,z^2,z^3]$. For generic $X$, the global object that corresponds locally to an ideal is an {\it ideal
  sheaf}, which in each coordinate patch $U_\alpha\subset X$ is described as an ideal
  $\cali_{U_\alpha}$ in the ring $\calo_{U_\alpha}$ of holomorphic
  functions on $U_\alpha$. More abstractly, an ideal sheaf is a rank one torsion free sheaf $\cale$ with $c_1(\cale)=0$. This is a purely commutative description, since the holomorphic functions on
  $\bbc^3$ form a commutative subalgebra of $\bbc^3_\theta$ for the Moyal deformation that we will consider. Thus the desired singular gauge field configurations will be realized explicitly in terms of noncommutative instantons~\cite{BN,INOV}.

\subsection{Crystal melting and random plane partitions}

As we will see, the counting of ideal sheaves is in fact equivalent to
a combinatorial problem, which provides an intriguing connection
between the K\"ahler quantum gravity model of Section~\ref{KahlerQG} and a particular
statistical mechanics model~\cite{ORV}. Consider a cubic crystal 
$$
\mbox{\includegraphics[width=5cm]{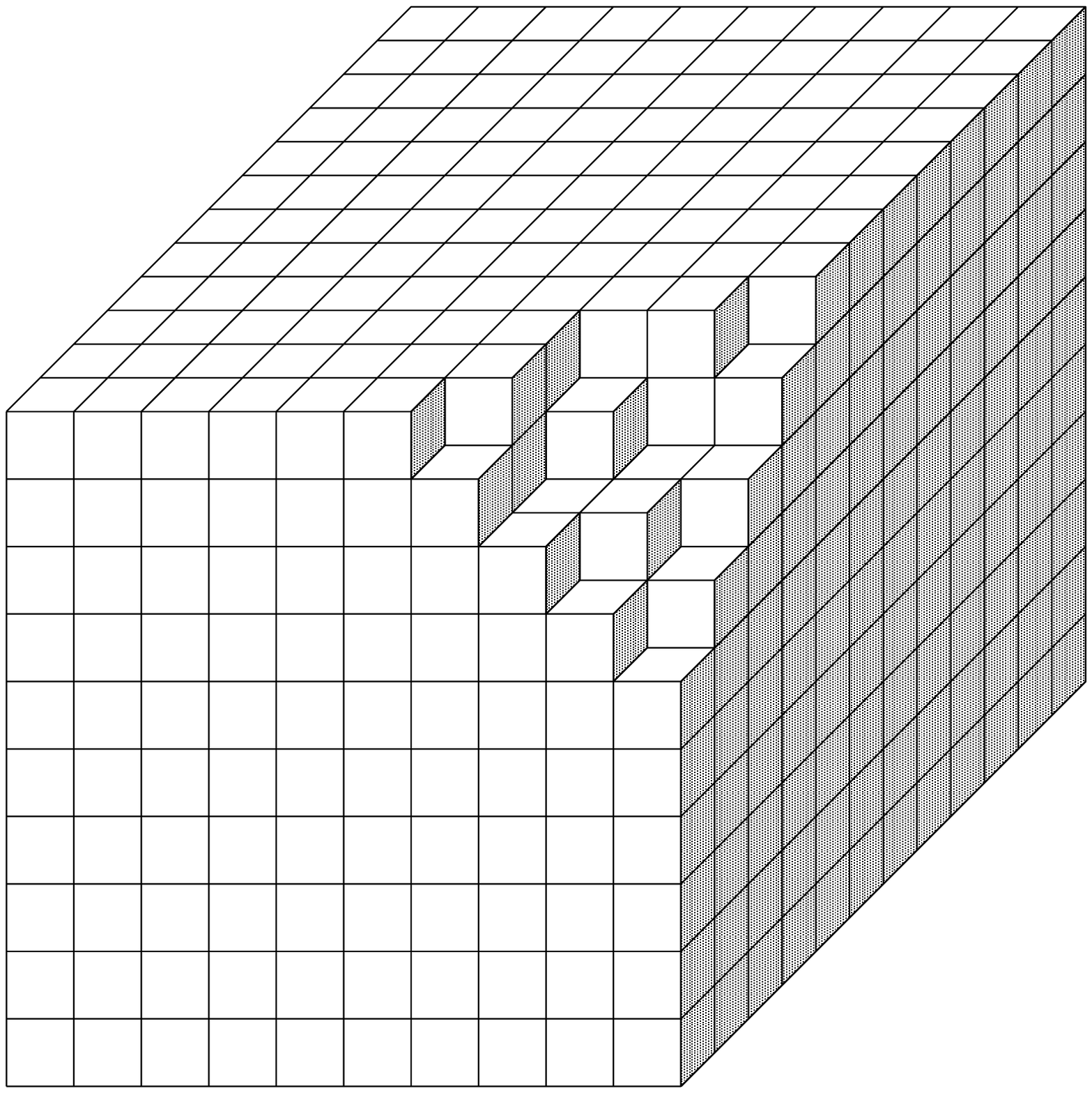}}
$$
located on the lattice $\bbz_{\geq0}^3\subset\bbr^3$. Suppose that we start heating the crystal at its outermost right
corner. As the crystal melts, we remove atoms, depicted symbolically
here by boxes, and arrange them into stacks of boxes in the positive
octant. Owing to the rules for arranging the boxes according to the
order in which they melt, this configuration defines a plane partition
or a three-dimensional Young diagram
$$
\mbox{\includegraphics[width=4cm]{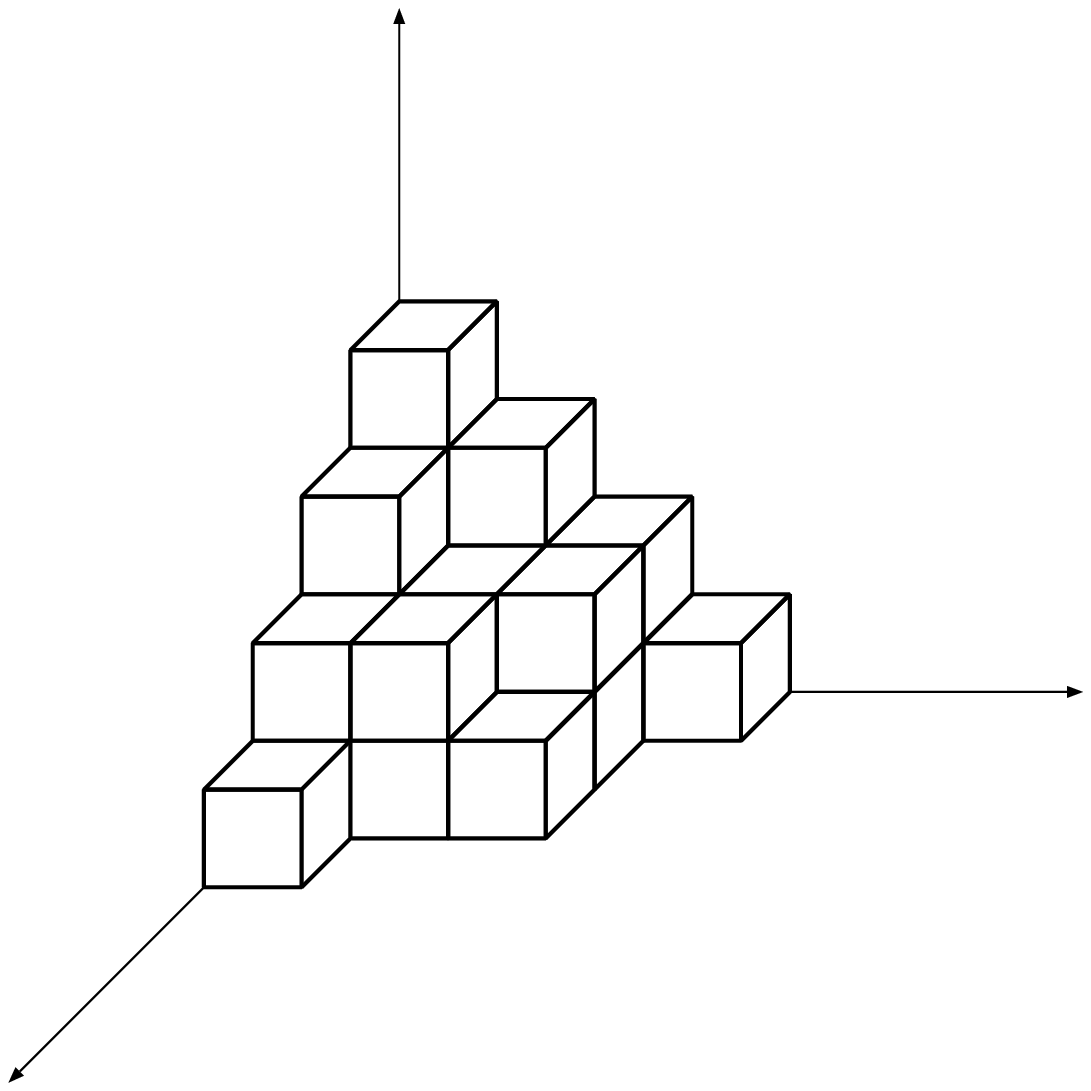}}
$$
Removing each atom from the corner of the crystal contributes a factor $q=\e^{-\mu/T}$ to the Boltzmann weight, where $\mu$ is the chemical potential and $T$ is the temperature.

Let us define more precisely the combinatorial object that we have constructed, which generalizes the usual notion of partition and Young tableau. 
A {\it plane partition} is a semi-infinite 
  rectangular array of non-negative integers
$$
\pi\=\begin{matrix} \pi_{1,1}&\pi_{1,2}&\pi_{1,3}& \cdots \\
\pi_{2,1}&\pi_{2,2}&\pi_{2,3}& \cdots \\
\pi_{3,1}&\pi_{3,2}&\pi_{3,3}& \cdots \\
\vdots&\vdots&\vdots& \end{matrix}
$$
such that $\pi_{i,j}\geq\pi_{i+1,j}$, $\pi_{i,j}\geq\pi_{i,j+1}$ for all
$i,j\geq1$. We may regard a partition $\pi$ as a three-dimensional Young diagram, in which we pile $\pi_{i,j}$ cubes vertically at the
$(i,j)$-th position in the plane as depicted above. The volume of a plane partition
$$
|\pi|\=\displaystyle{\sum_{i,j\geq1}\,\pi_{i,j}}
$$ 
is the total number of cubes. The diagonal slices of $\pi$ are the
partitions $(\pi_{i,i+m})_{i\geq1}$, $m\geq0$ 
obtained by cutting the three-dimensional Young diagram with planes,
and they represent a sequence of ordinary partitions (Young
  tableaux) $\lambda=(\lambda_1,\lambda_2,\dots)$, with
  $\lambda_i\geq\lambda_{i+1}$ for all $i\geq1$. Here $\lambda_i\geq0$
  is the length of the $i$-th row of the Young diagram, viewed as a
  collection of unit squares, and only finitely many $\lambda_i$ are non-zero.

The counting problem for random plane partitions can be solved explicitly in closed form. For this, we consider the statistical mechanics in a canonical ensemble in which each plane partition $\pi$ has
energy proportional to its volume $|\pi|$. The corresponding partition function then gives the generating function for plane partitions
\begin{eqnarray}
Z&:=&\sum_\pi\,q^{|\pi|} \nonumber \\[4pt]
&=&\sum_{N=0}^\infty\,pp(N)\,q^N \nonumber \\[4pt]
&=&\prod_{n=1}^\infty\,\frac1{\big(1-q^n\big)^n}~=:~M(q) \ ,
\label{MacMahon}\end{eqnarray}
where $pp(N)$ is the number of plane partitions $\pi$ with $|\pi|=N$ boxes. The function $M(q)$ is called the \emph{MacMahon function}.

\subsection{Six-dimensional cohomological gauge theory}

We will now describe a $U(1)$ gauge theory formulation of the above statistical models~\cite{INOV,CSS}.
If we gauge fix the residual symmetry of the quantized K\"ahler gravity action (\ref{Kahleractionexp}), we obtain the action
\begin{eqnarray}
S&=&\frac12\,\int_X\,\left(\dd_A\Phi\wedge \star\,\dd_A\overline{\Phi}+
\big|F^{2,0}\big|^2+\big|F^{1,1}\big|^2\right) \nonumber \\ && +\,
\frac12\,\int_X\,\left(F\wedge F\wedge\omega_0+
\frac{g_s}3\,F\wedge F\wedge F\right) \ ,
\label{6dgaugeaction}\end{eqnarray}
where 
$\dd_A=\dd-\ii[A,-]$ is the gauge covariant derivative acting on the
complex scalar field $\Phi$, $\star$ denotes the Hodge operator with
respect to the K\"ahler metric of $X$, and $F=\dd A$ is the curvature two-form which has the
K\"ahler decomposition $F=F^{2,0}+F^{1,1}+F^{0,2}$. The field theory defined by this action arises in three (related) instances as:
\begin{itemize}
\item[(1)] A topological twist of maximally supersymmetric Yang--Mills theory in six dimensions;
\item[(2)] The dimensional reduction of supersymmetric Yang--Mills theory in ten dimensions on $X$; and
\item[(3)] The low-energy effective field theory on a D6-brane wrapping
  $X$ in Type~IIA string theory, with D2 and D0 brane sources.
\end{itemize}

The gauge theory has a BRST symmetry~\cite{BKS,HP} and its partition function localizes at the BRST fixed points described by the equations
\begin{eqnarray}
F^{2,0}&=&0~=~F^{0,2} \ , \label{FA20}\\[4pt]
F^{1,1}\wedge\omega_0\wedge\omega_0&=&0 \ , \label{FA11}\\[4pt]
\dd_A\Phi&=&0 \ . \label{dAPhi}
\end{eqnarray}
These equations also describe three (related) quantities:
\begin{itemize}
\item The {\it Donaldson--Uhlenbeck--Yau (DUY) equations} expressing
  Mumford--Takemoto slope stability of holomorphic vector bundles over $X$ with finite
  characteristic classes;
\item BPS solutions in the gauge theory which correspond to (generalized) instantons; and
\item Bound states of D0--D2 branes in a single D6-brane wrapping $X$.
\end{itemize}
Recall that the equations (\ref{FA20}) and (\ref{FA11}) are a special
instance of the Hermitean Yang--Mills equations in which a
constant $\lambda$ is added to the right-hand side of
(\ref{FA11}). These equations arise in compactifications of heterotic
string theory. The condition that the compactification preserves at
least one unbroken supersymmetry requires $\lambda=0$. These are the
natural BPS conditions on a K\"ahler manifold $(X,\omega_0)$ which
generalize the usual self-duality equations in four dimensions.

The localization of the gauge theory partition function $Z$ onto the corresponding instanton moduli space $\calm_X$ can be written symbolically as~\cite{HP,INOV}
\beq
Z\=\int_{\calm_X}\,e(\caln_X) \ ,
\label{ZeNX}\eeq
where $e(\caln_X)$ is the Euler class of the obstruction bundle
$\caln_X$ whose fibres are spanned by the zero modes of the antighost
fields. The zero modes of the fermion fields in the full
supersymmetric extension of (\ref{6dgaugeaction})~\cite{BKS,HP} are in
correspondence with elements in the cohomology groups of the twisted Dolbeault
complex
$$
\xymatrix{
\Omega^{0,0}(X,{\rm
  ad}\,\Pcal)~\ar[r]^{\overline{\partial}_A} & ~\Omega^{0,1}(X,{\rm
  ad}\,\Pcal)~\ar[r]^{\overline{\partial}_A} & ~\Omega^{0,2}(X,{\rm
  ad}\,\Pcal)~\ar[r]^{\overline{\partial}_A} & ~\Omega^{0,3}(X,{\rm
  ad}\,\Pcal)
}
$$
with ${\rm ad}\,\Pcal$ the adjoint gauge bundle over $X$. By
incorporating the gauge fields one can rewrite this complex in the
form~\cite{INOV}
$$
\xymatrix{
\Omega^{0,0}(X,{\rm
  ad}\,\Pcal)~\ar[r] & ~ {\begin{matrix} \Omega^{0,1}(X,{\rm
  ad}\,\Pcal) \\ \oplus \\ \Omega^{0,3}(X,{\rm
  ad}\,\Pcal) \end{matrix}} ~\ar[r] & ~ \Omega^{0,2}(X,{\rm
  ad}\,\Pcal) \ ,
}
$$
which describes solutions of the DUY equations up to linearized
complex gauge transformations. The morphism $\Omega^{0,3}(X,{\rm
  ad}\,\Pcal)\to \Omega^{0,2}(X,{\rm
  ad}\,\Pcal)$ here is responsible for the appearence of the
obstruction bundle in (\ref{ZeNX})~\cite{HP,INOV}.

In order for the integral (\ref{ZeNX}) to be well-defined, we need to choose a compactification of $\calm_X$. In light of our earlier discussion, we will take this to be the Gieseker compactification, i.e. the moduli space of ideal sheaves on
$X$. The corresponding variety $\calm_X$ stratifies into components
${\rm Hilb}_{n,\beta}(X)$ given by the \emph{Hilbert scheme} of points
and curves in $X$, parametrizing isomorphism classes of ideal sheaves $\cale$ with
$\ch_1(\cale)=c_1(\cale)=0$, $\ch_2(\cale)=-\beta$, and
$\ch_3(\cale)=-n$. The partition function (\ref{ZeNX}) is the
generating function for the number of D0--D2 brane bound states in the
D6-brane wrapping $X$. Mathematically, these are the
\emph{Donaldson--Thomas invariants} of $X$. We will define this moduli
space integration, and hence these invariants, more precisely in Section~\ref{Wallcrossing}

\subsection{Localization in toric geometry}

Toric varieties provide a large class of algebraic varieties in which
difficult problems in algebraic geometry can be reduced to combinatorics.
Much of this paper will be concerned with these geometries as they possess symmetries which facilitate computations, particularly those
involving moduli space integrations. Let us start by recalling some
basic notions from toric geometry. Below we give the pertinent
definitions specifically in the case of varieties of complex dimension
three, the case of immediate interest to us, but they extend to
arbitrary dimensions in the obvious ways.

A smooth complex threefold $X$ is called a {\it toric
    manifold} if it densely contains a (complex algebraic) torus $T^3$ and the
  natural action of $T^3$ on itself (by translations) extends to the
  whole of $X$. Basic examples are the torus $T^3$ itself, the affine space
  $\bbc^3$, and the complex projective space $\bbP^3$. If in addition
  $X$ is Calabi--Yau, then $X$ is necessarily non-compact.

One of the great virtues of working with toric varieties $X$ is that
their geometry can be completely described by combinatorial data
encoded in a {\it toric
    diagram}. The toric diagram is a graph consisting of the following ingredients:
\begin{itemize}
\item A set of {\it vertices} $f$ which are the fixed points of the
  $T^3$-action on $X$, such that $X$ can be covered by
  $T^3$-invariant open charts homeomorphic to $\bbc^3$;
\item A set of {\it edges} $e$ which are $T^3$-invariant projective
  lines $\bbP^1\subset X$
  joining particular pairs of fixed points $f_1$, $f_2$; and
\item A set of {\it ``gluing rules''} for assembling the $\C^3$
  patches together to reconstruct the variety $X$. In a neighbourhood of each edge $e$,
  $X$ looks like the normal bundle over the corresponding
  $\bbP^1$. Since this normal bundle is a holomorphic bundle of rank two and every
  bundle over $\PP^1$ is a sum of line bundles (by the
  Grothendieck--Birkhoff theorem), it is of the form 
  $$\calo_{\bbP^1}(-m_1)\oplus\calo_{\bbP^1}(-m_2)$$ for some integers
  $m_1,m_2$. The normal bundle in this way
  determines the local geometry of $X$ near the edge $e$ via the
  transition function
  $$(w_1,w_2,w_3)~\longmapsto~
  \big(w_1^{-1}\,,\,w_2\,w_1^{-m_1}\,,\,w_3\,w_1^{-m_2}\big)$$
between the corresponding affine patches (going from the north pole to
the south pole of the associated $\PP^1$). In the Calabi--Yau case, the
Chern numbers $c_1(X)=0$ and $c_1(\PP^1)=2$ imply the condition
$m_1+m_2=2$.
\end{itemize}

For an open toric manifold $X$, we can exploit the toric symmetries to regularize
the infrared singularities on the instanton moduli space $\calm_X$ by
``undoing'' the $T^3$-rotations by gauge transformations~\cite{Nekrasov}. In this way
we will compute our moduli space integrals by using techniques from
equivariant localization, which in the present context will be refered
to as \emph{toric localization}. Recall that the bosonic sector of the
topologically twisted theory comprises a gauge connection $A_i$ and a
complex Higgs field $\Phi$. In particular, the supercharges contain a
scalar $Q$ and a vector $Q_i$. Generically, only $Q$ is conserved and
can be used to define the topological twist of the gauge theory. If
the threefold $X$ has symmetries then one can also use $Q_i$. In the
generic formulation of the theory, one only considers the scalar
topological charge $Q$ and restricts attention to gauge-invariant
observables. But in the present case one can also use the linear
combination
$$
Q_{\Omega} \= Q+\epsilon_a\,\Omega^a_{ij}\,
x^i\,Q^j \ ,
$$
where $\epsilon^a$ are the parameters of the isometric action of
$T^3\subset U(3)$ on the K\"ahler space $\C^3$, and
$\Omega^a=\Omega^a_{ij}\,x^j\,\frac\partial{\partial x_i}$ are vector
fields which generate
$SO(6)$ rotations of $\bbc^3\cong\bbr^6$. In this case we also 
consider observables which are only gauge-invariant up to a rotation. This
means that the new observables are equivariant differential
forms and the BRST charge $Q_{\Omega} $ can be interpreted as an equivariant
differential $\mathrm{d} + \imath_{\Omega} $ on the
space of field configurations, where $\imath_\Omega$ acts by
contraction with the vector field $\Omega$. 

This procedure modifies
the action and the equations of motion by mixing gauge invariance
with rotations. This set of modifications can sometimes be obtained by
defining the gauge theory on an appropriate supergravity background
called the ``$\Omega$-background''. In particular, the fixed point
equation (\ref{dAPhi}) is modified to
\beq
\dd_A\Phi\=\imath_\Omega F \ .
\label{dAPhiOmega}\eeq
The set of equations (\ref{FA20}), (\ref{FA11}) and (\ref{dAPhiOmega}) 
minimizes the action of the cohomological gauge theory in the
$\Omega$-background, and describes $T^3$-invariant instantons
(or, as we shall see, ideal sheaves). In particular, there is a
natural lift of the toric action to the instanton moduli space
$\calm_X$. We will henceforth study the gauge theory equivariantly, and interpret
the truncation of the partition function (\ref{ZeNX}) as an equivariant integral
over $\calm_X$. This will always mean that we work
solely in the Coulomb branch of the gauge theory. Due to the equivariant deformation of the BRST charge,
these moduli space integrals can be computed via equivariant
localization.

\subsection{Equivariant integration over moduli spaces\label{Eqint}}

We now explain the localization formulas that
will be used to compute partition functions throughout this paper. Let
$\calm$ be a smooth algebraic variety. Then we can define the
$\tilde T$-equivariant cohomology $H_{\tilde T}^\bullet(\calm,\Q)$ as
the ordinary cohomology $H^\bullet(\calm_{\tilde T},\Q)$ of the
Borel--Moore homotopy quotient $\calm_{\tilde T}:=(\calm\times
E\tilde T)/{\tilde T}$, where
$E\tilde T=(\C^\infty\setminus\{0\})^{N+k}$ is a contractible space 
on which $\tilde T= U(1)^N\times T^k$ acts freely. In the present
example of interest, $N=1$ and $k=3$. Given a
$\tilde T$-equivariant vector bundle $\Ecal\to\calm$, the quotient
$\Ecal_{\tilde T}=(\Ecal\times E\tilde T)/{\tilde T}$ is a vector
bundle over $\calm_{\tilde T}=(\calm\times E\tilde T)/{\tilde T}$. The
$\tilde T$-equivariant Euler class of $\Ecal$ is
the invertible element defined by
$$
e_{\tilde T}(\Ecal)~:=~e(\Ecal_{\tilde T})~\in~H_{\tilde T}^\bullet(
\calm,\Q) \ ,
$$
where $e$ is the ordinary Euler class for vector
bundles (the top Chern class). 

Let $B\tilde T:=E\tilde T/\tilde T=(\PP^\infty)^{k+N}$. Then
$E\tilde T\to B\tilde T$ is a universal principal $\tilde T$-bundle,
and there is a fibration $\calm_{\tilde T}\to B\tilde T$ with fibre
$\calm$. Integration in equivariant cohomology is defined as the
pushforward $\oint_\calm$ of the collapsing map $\calm\to{\rm pt}$, which coincides
with integration over the fibres $\calm$ of the bundle
$\calm_{\tilde T}\to B\tilde T$ in ordinary cohomology. Let
$p_i:B\tilde T\to\PP^\infty$ for $i=1,\dots,k$ and
$q_l:B\tilde T\to\PP^\infty$ for $l=1,\dots,N$ be the canonical
projections onto the $i$-th and $l$-th factors. Introduce equivariant
parameters $\epsilon_i=(c_1)_{\tilde T}(p_i^*\Ocal_{\PP^\infty}(1))$,
with
$t_i=\e^{\epsilon_i}=(\ch_{\tilde T})_1(p_i^*\Ocal_{\PP^\infty}(1))$,
and $a_l=(c_1)_{\tilde T}(q_l^*\Ocal_{\PP^\infty}(1))$, with
$e_l=\e^{a_l}=(\ch_{\tilde T})_1(q_l^*\Ocal_{\PP^\infty}(1))$.

The Atiyah--Bott localization formula in equivariant cohomology states that
\beq
\oint_\calm\,\alpha\=\oint_{\calm^{\tilde T}}\,
\frac{\alpha\big|_{\calm^{\tilde T}}}{e_{\tilde T}(\caln)}
\label{ABlocformula}\eeq
for any equivariant differential form $\alpha\in H_{\tilde T}^\bullet(\calm,\Q)$, where the complex
vector bundle $\caln\to\calm^{\tilde T}$ is the normal bundle over the
(compact) fixed point submanifold in $\calm$. When
$\calm^{\tilde T}$ consists of finitely many isolated points $f$, this formula simplifies
to 
\beq
\oint_\calm\,\alpha\=\sum_{f\in\calm^{\tilde T}}\,\frac{\alpha(f)}{e_{\tilde T}(T_f\calm)} \ .
\label{ABloc}\eeq
Each term in this sum takes values in the polynomial ring
$$
H_{\tilde T}^\bullet(f,\Q)
\=H^\bullet\big(B\tilde T\,,\,\Q\big)~\cong~
\Q[\epsilon_1,\dots,\epsilon_k,a_1,\dots,a_N]
$$
in the generators of
$\tilde T= U(1)^N\times T^k$. When the manifold $\calm$ is
non-compact, integration along the fibre is not a well-defined
$\Q$-linear map. Nevertheless, when $\calm^{\tilde T}$ is compact, we
can formally \emph{define} the equivariant integral $\oint_\calm\,\alpha$ by the right-hand side
of the formula (\ref{ABlocformula}).

Going back to our example, when $X=\C^3$, one has $\ch_2(\cale)=0$ and
the partition function $Z$ is saturated by
contributions from isolated, pointlike instantons (D0-branes) by a
formal application of the 
{localization formula} (\ref{ABloc}). However, these
expressions are all rather symbolic, as we are not guaranteed that the algebraic
scheme $\calm_X$ is a smooth variety, i.e. that the instanton
moduli space has a well-defined stable tangent bundle with tangent
spaces all of the same dimension. However, 
the variety $\calm_X$ is \emph{generically} smooth and 
there is a well-defined \emph{virtual} tangent bundle. The moduli
space integration (\ref{ABloc}) can then be formally
defined by virtual localization in equivariant Chow
theory. As discussed in ref.~\cite{MNOP}, the (stratified components
of the) instanton moduli space $\calm_X$
carries a canonical perfect obstruction theory in the sense
of ref.~\cite{GraPand}. In obstruction theory, the virtual tangent space at
a point $[\cale]\in\calm_X$ is given by 
$$
T_{[\cale]}^{\rm vir}\calm_X\=\Ext_{\Ocal_X}^1(\cale,\cale)\ominus\Ext_{\Ocal_X}^2(\cale,\cale) \ ,
$$
where $\Ext^1_{\Ocal_X}(\cale,\cale)$ is the Zariski tangent space and $\Ext^2_{\Ocal_X}(\cale,\cale)$ the
obstruction space of $\calm_X$ at $[\cale]$.
Its dimension is
given by the difference of Euler characteristics
$\chi(\Ocal_X\otimes\Ocal_X^\vee)-\chi(\cale\otimes \cale^\vee\,)$. The kernel of the trace map
$$
\Ext^2_{\Ocal_X}(\cale,\cale) ~ \longrightarrow ~ H^2(X,\Ocal_X)
$$
is the obstruction to smoothness at a point $[\cale]$ of the moduli space.

The bundles $\Ecal_i:=\Ext^i_{\Ocal_X}(\cale,\cale)$,
$i=1,2$ for $[\cale]\in\calm_X$ define a canonical
$T^k$-equivariant perfect obstruction theory
$\Ecal_\bullet=(\,\Ecal_1\to
\Ecal_2)$ (see ref.~\cite[Sec.~1]{GraPand}) on the
instanton moduli space $\calm=\calm_X$. In this case, one may
construct a virtual fundamental class $[\calm]^{\rm vir}$ and apply a
virtual localization formula. The general theory is developed
in ref.~\cite{GraPand} and requires a $T^k$-equivariant
embedding of $\calm$ in a smooth variety $\mathfrak{Y}$. The existence of
such an embedding in the present case follows from the stratification of $\calm_X$ into Hilbert schemes of points
and curves. Then one can deduce the localization formula over
$\calm$ from the known ambient localization formula over the smooth
variety $\mathfrak{Y}$, as above. In this paper we shall only need a
special case of this general framework, the virtual Bott residue
formula.

We can decompose $\Ecal_i$ into $T^k$-eigenbundles. The scheme
theoretic fixed point locus $\calm^{T^k}$ is the maximal
$T^k$-fixed closed subscheme of $\calm$. It carries a canonical
perfect obstruction theory, defined by the $T^k$-fixed part of
the restriction of the complex $\Ecal_\bullet$ to $\calm^{T^k}$,
which may be used to define a virtual structure on
$\calm^{T^k}$. The sum of the non-zero $T^k$-weight spaces of
$\Ecal_\bullet\big|_{\calm^{T^k}}$ defines the virtual normal
bundle $\caln^{\rm vir}$ to $\calm^{T^k}$. Define the Euler class
of a virtual bundle $\Acal=\Acal_1\ominus\Acal_2$
using formal multiplicativity, i.e. as the ratio of the Euler classes
of the two bundles, $e(\Acal)=e(\Acal_1)/e(\Acal_2)$. Then the
virtual Bott localization formula for the Euler class of a bundle
$\Acal$ of rank equal to the virtual dimension of $\calm$ reads~\cite{GraPand} 
\beq
\oint_{[\calm]^{\rm vir}}\,e(\Acal)\=\oint_{[\calm^{T^k}\,]^{\rm
    vir}}\, 
\frac{e_{T^k}\big(\Acal\big|_{\calm^{T^k}}\big)}
{e_{T^k}\big(\caln^{\rm vir}\big)} \ ,
\label{Bottloc}\eeq
where the integration is again defined via pushforward maps. The
equivariant Euler classes on the right-hand side of this formula are
invertible in the localized equivariant Chow ring of the scheme
$\calm$ given by
$CH_{T^k}^\bullet(\calm)\otimes_{\Q[\epsilon_1,\dots,\epsilon_k]}
\Q[\epsilon_1,\dots,\epsilon_k]_{\mathfrak{m}}$, where
$\Q[\epsilon_1,\dots,\epsilon_k]_{\mathfrak{m}}$ is the localization
of the ring $\Q[\epsilon_1,\dots,\epsilon_k]$ at the maximal ideal $\mathfrak{m}$
generated by $\epsilon_1,\dots,\epsilon_k$.

If $\calm$ is smooth, then $\calm^{T^k}$ is the non-singular set
theoretic fixed point locus, consisting here of finitely many points $[\cale]$. However, in general the
formula (\ref{Bottloc}) must be understood scheme theoretically, here as a
sum over $T^k$-fixed closed subschemes of $\calm$ supported at
the points $[\cale]\in\calm^{T^k}$ (with $k=3$). With
$\rho_\cale^i:T^k\to\End_\C\big(\Ext_{\Ocal_X}^i(\cale,\cale)\big)$, $i=1,2$ denoting the
induced torus actions on the tangent and obstruction bundles on
$\calm$, one generically has decompositions
\beq
\Ext^i_{\Ocal_X}(\cale,\cale)\=\Ext_1^i(\cale,\cale)\oplus
\ker\big(\rho_\cale^i(T^k)\big)
\label{Ext1EEinvt}\eeq
where $\Ext_1^i(\cale,\cale)$ is a $T^k$-invariant subspace of
$\Ext^i_{\Ocal_X}(\cale,\cale)$. As demonstrated 
in ref.~\cite[Sec.~4.5]{MNOP}, the
kernel module in (\ref{Ext1EEinvt}) vanishes. Hence each subscheme here is just
the reduced point $[\cale]$ and the $T^k$-fixed obstruction theory
at $[\cale]$ is trivial. Under these conditions, the virtual localization
formula (\ref{Bottloc}) may be written as 
$$
\oint_{[\calm]^{\rm vir}}\,e(\Acal)\=\sum_{[\cale]\in\calm^{T^k}}\,
\frac{e_{T^k}\big(\Acal([\calE])\big)}
{e_{T^k}\big(T_{[\cale]}^{\rm vir}\calm\big)} \ .
$$
The right-hand side of this formula again takes values in the
polynomial ring $\Q[\epsilon_1,\dots,\epsilon_k]$. When $\Ext_{\Ocal_X}^0(\cale,\cale)=\Ext_{\Ocal_X}^2(\cale,\cale)=0$ for
all $[\cale]\in\calm_X$, the moduli space
$\calm_X$ is a smooth algebraic variety with the trivial
perfect obstruction theory and this equation reduces immediately to the standard localization formula
in equivariant cohomology given above. In this paper we
will make the natural choice for the bundle $\Acal$, the virtual
tangent bundle $T^{\rm vir}\calm$ itself.

\subsection{Noncommutative gauge theory}

To compute the instanton contributions (\ref{ZeNX}) to the partition
function of the
cohomological gauge theory, we have to resolve the small instanton
ultraviolet singularities of $\calm_X$. This can be achieved by
replacing the space 
$X=\bbc^3\cong\bbr^6$ with its noncommutative deformation
$\bbr^6_\theta$ defined by letting the coordinate generators $x^i$,
$i=1,\dots,6$ satisfy
the commutation relations of the Weyl algebra
$$
\big[x^i\,,\,x^j\big]\=\ii\theta^{ij} \ ,
$$
where 
$$
\big(\theta^{ij}\big)\= 
\begin{pmatrix} 0&\theta_1& & & & \\{}
-\theta_1&0& & & & \\{} & &0&\theta_2& & \\{}
 & &-\theta_2&0& & \\{} & & & &0&\theta_3\\{}
 & & & &-\theta_3&0 \end{pmatrix}
$$
is a constant $6\times6$ skew-symmetric matrix which we take in Jordan
canonical form without loss of generality (by a suitable linear transformation of $\R^6$ if
necessary). We will assume that $\theta_1,\theta_2,\theta_3>0$ for
simplicity. The noncommutative polynomial algebra
$$
\alg\=\bbc\big[x^1\,,\,x^2\,,\,x^3\big]\,\big/\,
\big\langle[x^i,x^j]-\ii\theta^{ij}\big\rangle
$$ 
is regarded as the ``algebra of
functions'' on the noncommutative space $\bbr^6_\theta$.

We can represent the algebra $\alg$ on the standard Fock module
\beq
\hil\=\bbc\big[\alpha^\dag_1\,,\,\alpha^\dag_2\,,\,
\alpha^\dag_3\big]|0,0,0\rangle\=
\bigoplus_{i,j,k=0}^\infty\,\bbc|i,j,k\rangle \ ,
\label{Fockspace}\eeq
where the orthonormal basis states $|i,j,k\rangle$ are connected by
the action of creation and
annihilation operators $\alpha_a^\dag$ and $\alpha_a$, $a=1,2,3$. They
obey $\alpha_a|0,0,0\rangle=0$ and
$$
\big[\alpha_a^\dag\,,\,\alpha_b\big]\=\delta_{ab} \ , \qquad
\big[\alpha_a\,,\,\alpha_b\big]\=0\=\big[\alpha_a^\dag\,,\,
\alpha_b^\dag\big] \ .
$$
In the Weyl operator realization with the complex combinations of
operators $$z^a\=x^{2a-1}-\ii
x^{2a}\=\sqrt{2\theta_a}~\alpha_a \ , \qquad \bar z^{\bar a}\=x^{2a-1}+\ii
x^{2a}\=\sqrt{2\theta_a}~\alpha_a^\dag$$ for $a=1,2,3$, derivatives of
fields are replaced by the inner automorphisms
$$
\partial_{z^a}f~\longrightarrow~\frac1{2\theta_a}\,
\delta_{a\bar b}\,\big[\bar z^{\bar b}\,,\,f
\big] \ , 
$$
while spacetime averages are replaced by traces over $\hil$ according
to
$$
\int_{\bbr^6}\,\dd^6x~f(x)~\longrightarrow~
(2\pi)^3\,\theta_1\,\theta_2\,\theta_3~\Tr_\hil(f) \ .
$$

In the noncommutative gauge theory, we introduce the covariant coordinates
$$
X^i\=x^i+\ii\theta^{ij}\,A_j
$$
and their complex combinations
$$
Z^a\=\frac1{\sqrt{2\theta_a}}\,\big(X^{2a-1}+\ii X^{2a}\big)
$$
for $a=1,2,3$. Then the $(1,1)$ and $(2,0)$ components of the
curvature two-form can be respectively
expressed as
$$
F_{a\bar b}\=[Z_a,Z_{\bar b}]+\frac1{2\theta_a}\,\delta_{a\bar
  b} \ , \qquad F_{ab}\=[Z_{a},Z_{b}] \ ,
$$
while the covariant derivatives of the Higgs field becomes 
$$
(\partial_A)_a\Phi\=[Z_a,\Phi] \ .
$$
The instanton equations (\ref{FA20}), (\ref{FA11}) and
(\ref{dAPhiOmega}) then become algebraic equations
\bea
\big[Z^a\,,\,Z^b\big]&=&0 \ , \label{ZaZb}\\[4pt] \big[Z^a\,,\,Z_a^\dag\big]&=&3 \ ,
\label{ZaZadag}\\[4pt] \big[Z_a\,,\,\Phi\big]&=&\epsilon_a\,Z_a \ .
\label{ZaPhi}\eea
These equations describe BPS bound states of the D0--D6 system in a
$B$-field background, which is necessary for reinstating
supersymmetry~\cite{Mihailescu,WittenB}. In addition, $T^3$-invariance of the (unique) holomorphic
three-form on $X$ imposes the Calabi--Yau condition
\beq
\epsilon_1+\epsilon_2+\epsilon_3\=0 \ .
\label{CYepsiloncond}\eeq

\subsection{Instanton moduli space\label{InstmodspDT}}

A major technical advantage of introducing the noncommutative
deformation is that the instanton moduli space can be constructed
explicitly, by solving the noncommutative instanton equations
(\ref{ZaZb})--(\ref{ZaPhi}). First we construct the vacuum solution of
the noncommutative gauge theory, with $F=0$. It is obtained by
setting $A=0$ and is given explicitly by harmonic oscillator algebra
$$
Z^a\=\alpha_a \ , \qquad
\Phi\=\sum_{a=1}^3\,\epsilon_a\,\alpha_a^\dag\,\alpha_a \ .
$$

Other solutions are found via the solution generating technique
described in e.g. refs.~\cite{KrausSh,NekrasovLectures}. For the general solution, fix an integer $n\geq1$ and let $U_n$ be a partial isometry 
on the Hilbert space $\hil$ which projects out all states
$|i,j,k\rangle$ with $i+j+k<n$. Such an operator satisfies the equations
$$
U_n^\dag\,U_n\=1-\Pi_n \ , \qquad U_n\,U_n^\dag\=1
$$
where
$$
\Pi_n\=\sum_{i+j+k<n}\,|i,j,k\rangle\langle i,j,k|
$$
is a Hermitean projection operator onto a finite-dimensional subspace
of $\hil$. Then we make the ansatz
\beq
Z^a\=U_n\,\alpha_a\,f(N)\,U_n^\dag \ , \qquad
\Phi\=U_n\,\Big(\,\sum_{a=1}^3\,\epsilon_a\,\alpha_a^\dag\,\alpha_a\,
\Big)\,U_n^\dag \ ,
\label{ZaPhiansatz}\eeq
where $f$ is a real function of the number operator
$$N\=\displaystyle{\sum_{a=1}^3\,\alpha_a^\dag\,\alpha_a} \ . $$

Using standard harmonic oscillator algebra, we can write the DUY
equations (\ref{ZaZb})--(\ref{ZaPhi}) as
$$
U_n\,\big(N\,f^2(N-1)-(N+3)\,f^2(N)+3\big)\,U_n^\dag\=0 \ .
$$
This recursion relation has a unique solution with the initial conditions
$f(i)=0$, $i=0,1,\dots,n-1$, and the finite energy condition
$f(r)\rightarrow1$ as $r\rightarrow\infty$. It is given by~\cite{PS}
$$
f(N)\=\sqrt{1-\frac{n\,(n+1)\,(n+2)}{(N+1)\,(N+2)\,(N+3)}}~
(1-\Pi_n) \ .
$$
The topological charge of the corresponding noncommutative instanton is
\beq
\ell(n)\=\ch_3(\cale) \= - \mbox{$\frac\ii6$}\,\theta_1\,\theta_2\,\theta_3~ 
\Tr_\hil(F\wedge F\wedge F)\=\mbox{$\frac16$}\,
n\,(n+1)\,(n+2) \ .
\label{ellnch3}\eeq
Thus the instanton number is the number of states in $\hil$ with
$N<n$, i.e. the number of vectors removed by $U_n$, or equivalently
the rank of the projector $\Pi_n$.

The partial isometry $U_n$ identifies the full Fock space $\hil=
\bbc[\alpha^\dag_1,\alpha^\dag_2,
\alpha^\dag_3]|0,0,0\rangle$ with the subspace $$\hil_\cali\=
\displaystyle{\bigoplus_{f\in\cali}\,
  f\big(\alpha^\dag_1\,,\,\alpha^\dag_2\,,\, 
\alpha^\dag_3\big)|0,0,0\rangle} \ , $$ where
$$
\cali\=\bbc\big\langle w_1^i\,w_2^j\,w_3^k~\big|~i+j+k\geq
n\big\rangle
$$
is a monomial ideal of codimension $\ell=\ell(n)$ in the polynomial ring
$\bbc[w_1,w_2,w_3]$. The instanton moduli space can thus be identified
as the Hilbert scheme $$\calm_X\=\Hilb_{\ell,0}(X)\=X^{[\ell]}$$ of
  $\ell$ points in $X=\bbc^3$. The Hilbert--Chow morphism
$$
X^{[\ell]}~\longrightarrow~ {\rm Sym}^\ell(X)\= X^\ell\,\big/\,S_\ell
$$
identifies the Hilbert scheme of points as a crepant resolution of the
(coincident point) singularities of the
$\ell$-th symmetric product orbifold of $X$. The ideal $\cali$ defines a plane
partition $\pi$ with $|\pi|\=\ell$ boxes given by
$$
\pi\=\big\{(i,j,k)~\big|~i,j,k\geq1 ~ , ~
w_1^{i-1}\,w_2^{j-1}\,w_3^{k-1}\notin\cali\big\} \ .
$$
Heuristically, this configuration represents instantons which sit on
top of each other at the origin of $\bbc^3$, and
  along its coordinate axes where they asymptote to four-dimensional noncommutative
  instantons at infinity described by ordinary Young tableaux $\lambda$.

\subsection{Donaldson--Thomas theory\label{DTtheory}}

We can finally compute the instanton contributions to the partition
function of the cohomological gauge theory on any toric Calabi--Yau
threefold $X$~\cite{INOV,CSS}. Let us start with the case $X=\C^3$. Using
(\ref{ellnch3}), the contribution of an instanton corresponding to a
plane partition $\pi$ contributes a factor
$$
\displaystyle{\exp\Big(\,-\frac{\ii g_s}{48\pi^3}\,
\Tr_{\hil_\cali}\big(F^3\big)\Big)\=\e^{-g_s\,|\pi|}}
$$
to the Boltzmann weight appearing in the functional integral. There is
also a measure factor which comes from integrating out the bosonic and
fermionic fields in the supersymmetric gauge theory. This yields a
ratio of fluctuation determinants
\begin{eqnarray*}
Z_\pi&=&\frac{\det({\rm ad}\,\Phi)~\prod\limits_{i<j}\,\det(
{\rm ad}\,\Phi+\epsilon_i+\epsilon_j)}
{\det({\rm ad}\,\Phi+\epsilon_1+\epsilon_2+\epsilon_3)~
\prod\limits_{i=1}^3\,\det({\rm ad}\,\Phi+\epsilon_i)} \\[4pt]
&=&\exp\Big(-\int_0^\infty\,\frac{\dd t}t~\frac{\charac_\cali(t)\,
\charac_\cali(-t)}{\big(1-\e^{t\,\epsilon_1}\big)\,
\big(1-\e^{t\,\epsilon_2}\big)\,\big(1-\e^{t\,\epsilon_3}\big)}\Big)
\end{eqnarray*}
with the normalized character
\begin{eqnarray*}
\charac_\cali(t)&=& \prod\limits_{i=1}^3\,\big(1-\e^{t\,\epsilon_i}\big)~
\Tr_{\hil_\cali}\big(\e^{t\,\Phi}\big) \\[4pt]
&=& 1-\prod\limits_{i=1}^3\,\big(1-\e^{t\,\epsilon_i}\big)~
\sum\limits_{(i,j,k)\in\pi}\,\e^{t\,(\epsilon_1\,(i-1)+
\epsilon_2\,(j-1)+\epsilon_3\,(k-1))} \ ,
\end{eqnarray*}
where we have used the solution for $\Phi$ in (\ref{ZaPhiansatz}). Using 
the Calabi--Yau condition $\epsilon_1+\epsilon_2+\epsilon_3=0$, it is easy to see that these
determinants cancel up to a sign.

After some computation, one can explicitly determine this sign to get
$$
Z_\pi\=Z_{\pi=\emptyset}\cdot (-1)^{|\pi|} \ .
$$
The contribution $Z_\emptyset$ from the empty partition is the
one-loop perturbative contribution to the functional integral, and
hence will be dropped. Then the instanton sum for the partition
function is given by
$$
Z_{\rm DT}^{\bbc^3}(q)\=\sum_\pi\,\big(-\e^{-g_s}\big)^{|\pi|}\=
\sum_\pi\,q^{|\pi|} \ ,
$$
which is just the MacMahon function $M(q)$ with $q=-\e^{-g_s}$. This is the
known formula for the Donaldson--Thomas partition function on
$\bbc^3$.

For later use, let us note a convenient resummation formula for this partition
function~\cite{ORV}. Using interlacing relations, the sum over plane partitions $\pi$ can be converted into a
triple sum over the Young tableaux obtained from the main diagonal
slice $\lambda=(\pi_{i,i})_{i\geq1}$, together with a sum over pairs
of semi-standard tableaux of shape $\lambda$ obtained by putting $m+1$ of them
in boxes of the skew diagram associated to the $m$-th diagonal slice
for each $m\geq0$. The partial sum over each semi-standard tableaux
coincides with the combinatorial definition of the \emph{Schur
  functions} at a particular value, which can be expressed through the
hook formula
$$
\displaystyle{s_\lambda(q^\rho)\= q^{n(\lambda)+|\lambda|/2}\,
\prod_{(i,j)\in\lambda}\,\frac1{1-q^{h(i,j)}}}
$$
where $n(\lambda)\=\sum_i\,(i-1)\,\lambda_i$, and $h(i,j)$ is the hook length of
the box located at position $(i,j)$ in the Young tableau $\lambda\subset\Z_{\geq0}^2$. Then the partition function
can be rewritten as a sum over \emph{ordinary} partitions
$$
Z_{\rm DT}^{\bbc^3}(q)\=\sum_\lambda\,s_\lambda(q^\rho)^2 \ .
$$

This construction can be generalized to \emph{arbitrary} toric
Calabi--Yau threefolds $X$ by using the gluing rules of toric
geometry. The two simplest such varieties are described by the toric diagrams
$$
\mbox{\includegraphics[width=10cm]{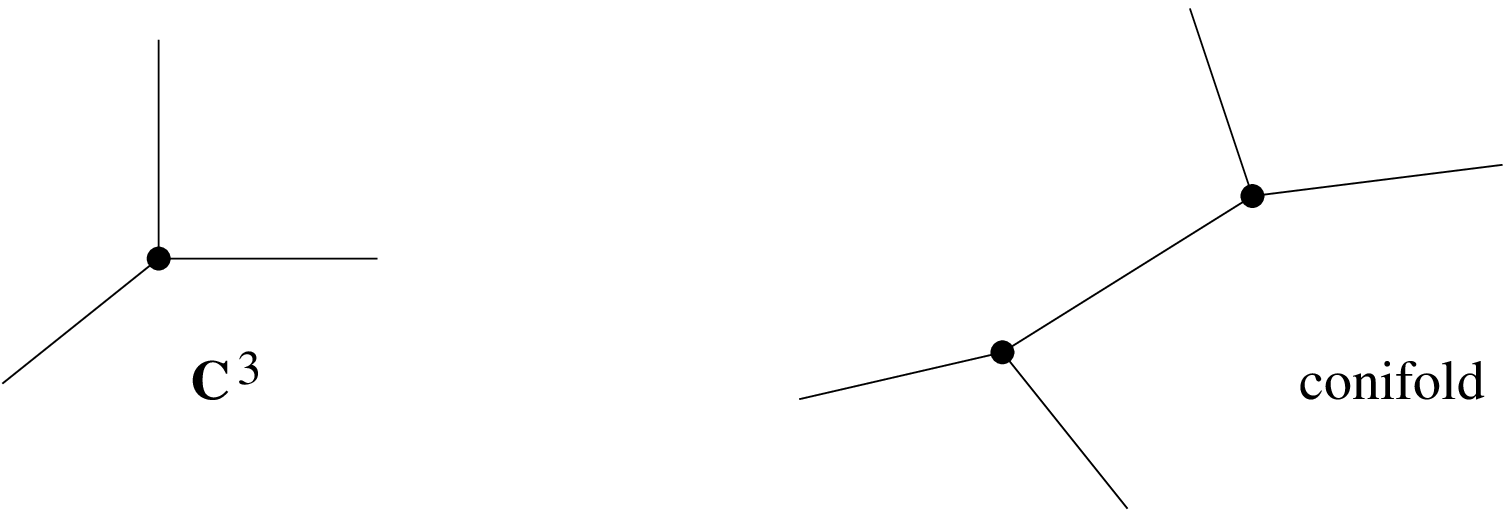}}
$$
with a single vertex representing $X=\C^3$, whose partition function
was computed above and is the basic building block for the generating
functions on more complicated geometries, and a single line joining
two vertices representing the resolved conifold
$X=\calo_{\bbP^1}(-1)\oplus\calo_{\bbP^1}(-1)$, where the $\PP^1$
contribute to the $F\wedge F\wedge\omega_0$ term of the gauge theory
action (\ref{6dgaugeaction}). The $T^3$-invariant
noncommutative $U(1)$ instantons on $X$ correspond to ideal sheaves $\cale$ and are
described by the following combinatorial data:
\begin{enumerate}
\item[(1)] Three-dimensional partitions $\pi_f$ at each vertex $f$ of
  the toric diagram, corresponding to monomial ideals
  $\cali_f\subset\bbc[w_1,w_2,w_3]$; and
\item[(2)] Two-dimensional partitions $\lambda_e$ at each edge $e$ of
  the toric diagram, representing the four-dimensional instanton
  asymptotics of $\pi_f$.
\end{enumerate}
This description requires generalizing the calculation on $X=\C^3$
above to compute the \emph{perpendicular partition function}
   $P_{\lambda,\mu,\nu}(q)$~\cite{ORV}, which is defined to be the 
   generating function for three-dimensional partitions with fixed asymptotics $\lambda$,
   $\mu$, and $\nu$ in the three coordinate directions. Such partitions correspond to instantons on $\bbc_\theta^3$ with non-trivial
   boundary conditions at infinity along each of the coordinate axes. It can be
   expressed in terms of \emph{skew Schur functions}, with $Z_{\rm
     DT}^{\C^3}(q)=P_{\emptyset,\emptyset,\emptyset}(q)$. 

For the example of the resolved conifold, using the gluing rules one easily computes
\begin{eqnarray}
{Z}_{\rm DT}^{\rm conifold}(q,Q) &=& \sum_{\pi_f}\,
q^{|\pi_f|+\sum_{(i,j)\in\lambda} \,(i+j+1)}~ (-1)^{|\lambda|}
~Q^{|\lambda|} \nonumber \\[4pt] &=&
\prod_{n=1}^\infty\,\frac{\big(1-(-1)^{n}\,q^n\,Q\big)^{n}}
{\big(1-(-1)^{n}\,q^n\big)^{2n}}\= M(-q)^2~\prod_{n=1}^\infty\,\big(1-(-q)^n\,Q\big)^{n} \ .
\label{conpart}\end{eqnarray}
More generally, with these rules one finds that the instanton partition function is the generating function
$$
Z_{\rm DT}^X(q,\mbf Q)\=\sum_{n\in\Z}~\sum_{\beta\in H_2(X,\Z)}\,I_{n,\beta}(X)~q^n\,\mbf Q^\beta
$$
for the Donaldson--Thomas invariants $I_{n,\beta}(X)\in\Z$, which are
defined as follows. The moduli variety $\Hilb_{n,\beta}(X)$ of ideal
sheaves on $X$ is a projective scheme with a perfect obstruction
theory. For general threefolds $X$, it has virtual dimension~\cite[Lem.~1]{MNOPII}
$$
\int_\beta\,c_1(X)
$$
which coincides with that of $\calm_g(X,\beta)$ from
Section~\ref{GWtheory} In the Calabi--Yau case, the virtual dimension
is zero, and the corresponding virtual cycle is
$$
\big[\Hilb_{n,\beta}(X)\big]^{\rm vir}~\in~CH_0\big(\Hilb_{n,\beta}(X)\big) \ .
$$
Then the Donaldson--Thomas invariants
$$
I_{n,\beta}(X)~:=~\int_{[\Hilb_{n,\beta}(X)]^{\rm vir}}\,1
$$
count the virtual numbers of ideal sheaves on $X$ with the given Chern
character. The right-hand side is defined via equivariant integration,
as explained in Section~\ref{Eqint} The torus action on $X$ lifts to
the moduli scheme $\Hilb_{n,\beta}(X)$. The $\tilde T$-fixed locus
$\Hilb_{n,\beta}(X)^{\tilde T}$ has a $\tilde T$-equivariant virtual
theory with cycle $[\Hilb_{n,\beta}(X)^{\tilde T}]^{\rm vir}\in
CH_0(\Hilb_{n,\beta}(X)^{\tilde T})$ and virtual normal bundle
$\caln_X^{\rm vir}$ in the equivariant K-theory $K_{\tilde
  T}^0(\Hilb_{n,\beta}(X)^{\tilde T})$. This construction gives
precise meaning to the moduli space integral (\ref{ZeNX}) via
application of the virtual localization formula in equivariant Chow
theory, described in Section~\ref{Eqint}

\subsection{Wall-crossing formulas\label{Wallcrossing}}

We will now make contact with Section~\ref{topstrings}. For the present
class of threefolds $X$, there is a gauge theory/string theory
duality~\cite{MNOP}. This follows from the fact that the perpendicular partition
function $P_{\lambda,\mu,\nu}(q)$ is related to the Calabi--Yau crystal formulation of the \emph{topological vertex}
 $$
 C_{\lambda,\mu,\nu}(q)\=M(q)^{-1}~q^{\frac12\,(\|\lambda\|^2+
 \|\mu\|^2+\|\nu\|^2)}~P_{\lambda,\mu,\nu}(q)
 $$
 with $\|\lambda\|^2=\sum_i\,\lambda_i^2$, which are the building
 blocks for the computation of the generating function for
 Gromov--Witten invariants using rules analogous to those described in
 Section~\ref{DTtheory}~\cite{Marino} Using these relations one can show that the
 six-dimensional cohomological gauge theory is S-dual to the A-model topological string theory. The respective partition functions are related by
\beq
{Z_{\rm top}^X(g_s,\mbf Q)\=M(q)^{-\chi(X)}~
 Z_{\rm DT}^X\big(q=-\e^{-g_s}\,,\,\mbf Q\big)} \ ,
\label{DTGWcorr}\eeq
where the Euler characteristic $\chi(X)$ of $X$ is the number of vertices
in its toric diagram. For the conifold example, the gluing rules for
the topological vertex yield~\cite{Marino}
$$
Z_{\rm top}^{\rm
  conifold}(g_s,Q)\=\sum_{\lambda}\,C_{\emptyset,\emptyset,\lambda}(q)
\,
C_{\emptyset,\emptyset,\lambda}(q)~Q^{|\lambda|}\=
\exp\Big(\,\sum_{n=1}^\infty\,\frac{Q^n}n\,
\frac1{\big(q^{n/2}-q^{-n/2}\big)^2} \,\Big) \ ,
$$
which should be compared with (\ref{conpart}).
This Gromov--Witten/Donaldson--Thomas
correspondence is known to hold for arbitrary toric threefolds~\cite{MOOP}. The
relationship (\ref{DTGWcorr}) can be thought of as a wall-crossing
formula, as we now explain.

The relationship (\ref{DTGWcorr}) is in apparent contradiction with
the OSV conjecture (\ref{OSV}), if we wish to interpret the right-hand
side as the generating function $Z_{\rm BH}(1,0,\mbf\phi^2,\phi^0)$ of
a suitable index for black hole microstates. However, the conjectural
relations (\ref{OSV}) and (\ref{DTGWcorr}) hold in different regimes
of validity. The number of BPS particles in four dimensions formed by
wrapping supersymmetric bound states of D-branes around holomorphic
cycles of $X$ depends on the choice of a stability condition, and the
BPS countings for different stability conditions are related by
wall-crossing formulas. For example, stability of black holes requires
that their chemical potentials $\mu^I$ lie in the ranges $\calq_0\,\phi^0>0$ and $\calq_2^i\,\phi_i^2>0$. 

On the other hand, the validity of (\ref{DTGWcorr}) is related to the
existence of BPS invariants $B_{g,\beta}(X)\in\Z$ such that the
topological string amplitudes have an expansion given by~\cite{GopVafa}
$$
Z_{\rm top}^X(g_s,\mbf Q)\=\sum_{g=0}^\infty~\sum_{\beta\in
  H_2(X,\Z)}\,\mbf Q^\beta~ \sum_{\stackrel{\scriptstyle \gamma\in
    H_2(X,\Z)\setminus\{0\}}{\scriptstyle\beta=k\,\gamma}}\,B_{g,\beta}(X)~ \frac1k\, \left(2\sinh\Big(\frac{k\,g_s}2\Big)\right)^{2g-2} \ ,  
$$
of which the conifold partition function (\ref{conpart}) is an
explicit case. These are partition functions of D6--D2--D0 brane bound
states only for certain K\"ahler moduli. Analyses of Calabi--Yau
compactifications of Type~II string theory show that the Hilbert
spaces of BPS states jump discontinuously across real codimension one
walls in the moduli space of vacua, known as walls of marginal
stability. The noncommutative instantons do not account for walls of
marginal stability extending to infinity. One should instead apply
some sort of stability condition (such as $\Pi$-stability) to elements
of the bounded derived category of coherent sheaves
${\sf D}^b(\coh(X))$ of the given
charge, which gives a topological classification of A-model D-branes on $X$. These issues are discussed in more detail in refs.~\cite{JM,CJ,AOVY}.

From a mathematical perspective, we can study this phenomenon by
looking at framed moduli spaces, which consist of instantons that are
trivial ``at infinity''. More precisely, we can consider a toric
compactification of $X$ obtained by adding a compactification divisor
$D_\infty$, and consider sheaves $\calf$ with a fixed trivialization on
$D_\infty$. The K\"ahler polarization defined by
$\omega_0$ allows us to define the moduli space
$\calm_X^{\omega_0}=\calm_X$ of stable sheaves. Then the symbolic
definition of the gauge theory partition function (\ref{ZeNX}) as a
particular Euler characteristic can be made precise in the more local
definition of Donaldson--Thomas invariants given by
ref.~\cite{Behrend}.

As a scheme with a perfect obstruction theory, the instanton moduli
space $\calm_X$ can be viewed locally as the scheme theoretic critical
locus of a holomorphic function, the superpotential $W$, on a compact
manifold $\cal X$ with the action of a gauge group $\cal G$. $\calm_X$
has virtual dimension zero, and at non-singular points the obstruction
sheaf $\caln_X$ on $\calm_X$ coincides with the cotangent
bundle. Hence if $\calm_X$ were everywhere non-singular, then the
partition function (\ref{ZeNX}) would just compute the signed Euler
characteristic $(-1)^{\dim_\C(\calm_X)}\,\chi(\calm_X)$. At singular
  points, however, the invariants differ from these characteristics.

There is a constructible function
$\nu:\calm_X\to\Z$ which can be used to define the \emph{weighted} Euler
characteristic
$$
\chi(\calm_X,\nu)~:=~ \sum_{n\in\Z}\,n\cdot\chi\big(\nu^{-1}(n)\big) \
.
$$
For sheaves of fixed Chern character, this coincides with the
curve-counting invariants $I_{n,\beta}(X)$. At non-singular points,
$\nu=(-1)^{\dim_\C(\calm_X)}$, while at singular points it is the more
complicated function given by
$$
\nu(\cale)\=(-1)^{\dim_\C({\cal X}/{\cal G})}\,\big(1-{\rm
  MF}_W(\cale) \big) \ ,
$$
where ${\rm MF}_W(\cale)$ is the Milnor fibre of the superpotential
$W$ at the point corresponding to $\cale$. The weighted Euler
characteristic is a deformation invariant of $X$.

In this approach, one
can use topological Euler characteristics to define $I_{n,\beta}(X)$
as invariants associated to moduli varieties of framed
sheaves. Fixing $\beta\in H_2(X,\Z)$ and $n\in\Z$, the variety
$\Hilb_{n,\beta}(X)$ equivalently parametrizes isomorphism classes of the following objects:
\begin{itemize}
\item[(a)] Surjections (framings) 
$$
\xymatrix{
\calo_X~\ar[r] & ~ \calf ~\ar[r] & ~0
}
$$
with $\ch(\calf)=(1,0,\beta,n)$;
\item[(b)] Stable sheaves $\cale$ with $\ch(\cale)=(1,0,-\beta,-n)$ and
  trivial determinant; and
\item[(c)] Subschemes $S\subset X$ of dimension $\leq1$ with curve
  class $[S]=\beta$ and holomorphic Euler characteristic
  $\chi(\calo_S)=n$.
\end{itemize}
The equivalences between these three descriptions is described
explicitly for $X=\C^3$ in ref.~\cite{CSS}. 

As we vary the polarization $\omega_0$, the moduli spaces
$\calm_X^{\omega_0}(X)$ change and so do the associated counting
invariants, leading to a wall-and-chamber structure. The wall-crossing behaviour of the enumerative invariants
is studied in~refs.~\cite{Joyce,KontSoib}. The analog of varying $\omega_0$
for framed sheaves is to consider quotients of the structure sheaf
$\calo_X$ in different abelian subcategories of the bounded derived
category ${\sf D}^b(\coh(X))$ of coherent sheaves on $X$. The analog of wall-crossing
gives the Pandharipande--Thomas theory of stable pairs~\cite{PT} and
the BPS invariants above. For this, the quotients of $\Ocal_X$ are the stable pairs $(\cale,\alpha)$, where $\cale$ is a coherent $\Ocal_X$-module of pure dimension one with $\ch_2(\cale)=-\beta$ and $\chi(\cale)=-n$, and $\alpha:\Ocal_X\to\cale$ is a non-zero sheaf map such that ${\rm coker}(\alpha)$ is of pure dimension zero, together with le~Poitier's $\delta$-stability condition for coherent systems. In this case the change of Donaldson--Thomas
invariants is described by the Kontsevich--Soibelman wall-crossing formula~\cite{KontSoib}.

To cast these constructions into the language of
noncommutative instantons, a proper definition of noncommutative toric
manifolds is desired, beyond the heuristic approach presented above
whereby only open $\C^3$ patches are deformed. Isospectral type
deformations of toric geometry, and instantons therein, are
investigated in ref.~\cite{CLS}. It may also aid in the classification of $U(N)$ noncommutative instantons on $\C^3$ for rank
$N>1$, along the lines of what was done in
Section~\ref{InstmodspDT} (See ref.~\cite{CSS} for some explicit examples.) This appears to be related to the problem of defining
a non-abelian version of Donaldson--Thomas theory which counts
higher-rank torsion-free sheaves, for which no general, appropriate notion of
stability is yet known.

\bigskip

\section{D4-brane gauge theory and Euler characteristics \label{D4GT}}

\noindent
In this section we will take $\calq_6=0$ (no D6-branes), and consider $N$ D4-branes
  wrapping a four-cycle $C\subset X$. In this case the worldvolume gauge theory
  on the D4-branes is the $\mathcal{N}=4$
  Vafa--Witten topologically twisted $U(N)$ Yang--Mills theory on $C$,
  where the topological twist is generically required in order to
  realize covariantly constant spinors on a curved geometry. When the gauge theory is formulated on an
  arbitrary toric singularity $C$ in four dimensions, we may regard
  $C$ as a four-cycle inside the Calabi--Yau threefold $X=K_C$, and we will obtain
  an explicit description of the instanton moduli spaces and their
  Euler characteristics. The precise forms of the partition functions will be
  amenable to checks of the OSV conjecture (\ref{OSV}), and
  hence a description of wall-crossing phenomena.

\subsection{$\mathcal{N}=4$ supersymmetric Yang--Mills theory on K\"ahler
  surfaces}

Vafa and Witten~\cite{VW} introduced a topologically
twisted version of $\cN=4$ supersymmetric Yang--Mills theory in four dimensions. The twisting
procedure modifies the quantum numbers of the fields in the
physical theory in such a way that a particular linear combination
of the supercharges becomes a scalar. This scalar supercharge is
used to define the cohomological field theory and its observables on an
arbitrary four-manifold $C$. In the following we will only
consider the case where $C$ is a connected smooth K\"ahler manifold with K\"ahler
two-form $k_0$. When certain conditions are met the partition function
of the twisted gauge theory computes the Euler characteristic of the
instanton moduli space.

Let $g_{ij}$ be the K\"ahler metric of $(C,k_0)$. Then 
the twisted gauge theory corresponds to the moduli problem associated
with the equations
\begin{eqnarray}
\sigma_{ij} &:=& F_{ij}^+ + \mbox{$\frac{1}{4}$}\, \big[ B^+_{ik}
\, , \, B_{jl}^+ \big] \, g^{kl} + \mbox{$\frac{1}{2}$}\,
\big[ \Phi \, , \, B^+_{ij} \big] \= 0 \ , \nonumber \\[4pt] \kappa_i &:=& \dd_A^{j}
B^+_{ij} + (\dd_A)_i \Phi \= 0 \ ,
\label{VWeqs}\end{eqnarray}
where $F^+=\frac12\,(F-\star F)$ is the self-dual part of the curvature two-form with respect to the K\"ahler metric. The field space $\mathfrak{W}$ is spanned by a connection $A_i$ on a
principal $G$-bundle $\calp \rightarrow C$, a scalar field $\Phi$, and
a self-dual two-form $B^+_{ij}$, so that
\begin{equation} \label{N4fieldspace}
\mathfrak{W} \= \mathfrak{A}_\calp \times \Omega^0( C , \mathrm{ad}\, \calp) \times
\Omega^{2,+}( C , \mathrm{ad}\, \calp)
\end{equation}
where $\mathfrak{A}_\calp$ denotes the space of connections on $\calp$ and
$\mathrm{ad}\,\calp$ is the adjoint bundle of $\calp$. Their superpartners
$\psi_i$, $\zeta$, and $\tilde{\psi}^+_{ij}$ live in the
tangent space to $\mathfrak{W}$. Associated with the equations of motion are two multiplets $( \chi_{ij}^+, H^+_{ij})$ and $( \tilde\chi_i, \tilde{H}_i)$
which are sections of the bundle $$\mathfrak{F} \=  \Omega^{2,+}(C ,
\mathrm{ad}\,\calp) \oplus  \Omega^1(C, \mathrm{ad}\,\calp) \ . $$ Schematically, the action of
the topological gauge theory is of the form
$$
S \= \{ Q , \Psi \} + \int_C\, \Tr (F \wedge F) + \int_C\, \Tr (F \wedge
 k_0)
$$
where $Q$ is the scalar supercharge singled out by the twisting
procedure. The gauge fermion $\Psi$ is a suitable functional of the
fields which
contains the term
$$
\int_C\, \sqrt{g}~\Tr \Big(\chi_{ij}^+\, \left( H^{+ \,ij} +
\sigma^{ij} \right) + \tilde{\chi}_i\,\big( \tilde{H}^i +
\kappa^i\big) \Big) \ ,
$$
that makes the gauge theory localize onto the solutions of the equations
(\ref{VWeqs}).

Geometrically, the partition function can be interpreted as a
Mathai--Quillen representative of the Thom class of the bundle
$\mathfrak{V} = \mathfrak{W} \times_{\mathcal{G}} \mathfrak{F}$, where
$\mathcal{G}={\rm Aut}(\calp)$ is the group of gauge transformations. Its
pullback via the sections in (\ref{VWeqs}) gives the Euler class of
$\mathfrak{V}$. Under
favourable circumstances, appropriate vanishing theorems hold~\cite{VW} which
ensure that each solution of the system (\ref{VWeqs}) has
$\Phi=B^+=0$ and corresponds to an instanton, i.e. a solution to the self-duality equations $F^+=0$. In this case the gauge
theory localizes onto the instanton moduli space $\mathfrak{M}_C$ and
the Boltzmann weight gives a representative of the
Euler class of the tangent bundle $T \mathfrak{M}_C$. Therefore the
partition function computes moduli space integrals of the form
$$
\int_{\mathfrak{M}_C} \, e( T \mathfrak{M}_C) \= \chi(
\mathfrak{M}_C) \ ,
$$
which gives the Euler characteristic of the instanton moduli space.
Since the instanton moduli space is not generally a smooth variety, most of the quantities introduced above can only be
defined formally. We will discuss how to define these integrations
more precisely later on. In particular, we will allow for non-trivial vacuum
expectation values for the Higgs field $\Phi$, in order to define the
partition function in the $\Omega$-background as before. We will assume that the
vanishing theorems can be extended to this case as well, by replacing the instanton moduli space with
its compactification obtained by adding torsion-free sheaves on $C$ as
before.

The Euler characteristic of instanton moduli space can be computed through the index of the deformation
complex associated with $\mathcal{N}=4$ topological Yang--Mills theory
via the equations (\ref{VWeqs}). It has the form~\cite{LL}
\begin{equation} \label{VWcomplex}
\xymatrix@1{
  \Omega^0(C , \mathrm{ad}\,\calp)
   \quad\ar[r]^{\hspace{-0.5cm} \textsf{D}} &\quad
   {\begin{matrix} \Omega^1(C, \mathrm{ad}\,\calp)
   \\ \oplus \\
   \Omega^0(C, \mathrm{ad}\,\calp) \\ \oplus \\ \Omega^{2,+}
   (C, \mathrm{ad}\,\calp) \end{matrix}}\quad \ar[r]^{ \textsf{s}} & \quad
   {\begin{matrix} \Omega^{2,+}(C, \mathrm{ad}\,\calp) \\ \oplus \\
   \Omega^1(C, \mathrm{ad}\,\calp)
   \end{matrix}}
}
\end{equation}
where the first morphism is an infinitesimal gauge transformation
$$
\textsf{D} (\phi) \= \left( \begin{matrix} \mathrm{d}_A \phi
\\ \left[ \Phi , \phi \right] \\ \left[ B^+ , \phi \right]
\end{matrix} \right) \ ,
$$
while the second morphism corresponds to the linearization of the sections
$(\sigma_{ij} , \kappa_i)$ given by
\begin{eqnarray*}
\textsf{s}\big( \psi \,,\, \zeta \,,\, \tilde{\psi}^+ \big) &=& 
p^+ \mathrm{d}_A \psi - \big[ \tilde{\psi}^+ \,,\, B^+ \big] +
\big[ \tilde{\psi}^+ \,,\, \Phi \big] + \big[ B^+ \,,\, \zeta \big]
\\ && \qquad +\, \mathrm{d}_A \zeta + \left[ \psi , \Phi \right] +
p^+ \mathrm{d}_A^* \tilde{\psi}^+ + \left[ \psi \,,\, B^+ \right]
 \ ,
\end{eqnarray*}
with $p^+$ giving the projection of a two-form onto its self-dual part. Under the assumption that all
solutions of the original system of equations (\ref{VWeqs})
have $\Phi = B^+ = 0$, the complex (\ref{VWcomplex}) splits into the
Atiyah--Hitchin--Singer instanton deformation complex
\beq\label{AHScomplex}
\xymatrix{
  \Omega^0(C, \mathrm{ad}\,\calp)
   \quad \ar[r]^{ \hspace{-0.3cm} \mathrm{d}_A} & \quad
  \Omega^1(C, \mathrm{ad}\,\calp) \quad
  \ar[r]^{p^+\circ \mathrm{d}_A} & \quad
   \Omega^{2,+}(C, \mathrm{ad}\,\calp)
}
\eeq
plus
$$
  \Omega^0(C, \mathrm{ad}\,\calp) ~\oplus~ \Omega^{2,+}(C, \mathrm{ad}\,\calp)
   \quad \xrightarrow{ (\mathrm{d}_A \,,\, p^+\circ \mathrm{d}_A^*)} \quad
  \Omega^1(C, \mathrm{ad}\,\calp)
$$
which is again the instanton deformation complex. One can compute
the index of the original complex (\ref{VWcomplex}) (assuming the Vafa--Witten
vanishing theorems) by computing the index of the two complexes
above. However, these are equal and contribute with opposite signs. This means that the Euler characteristic of instanton moduli space receives contributions only from isolated points and simply counts the number of such points. On a toric surface $C$, this is anticipated from the toric localization formula (\ref{ABloc}) and will be made explicit below.

In the applications to black hole microstate counting, we will consider gauge group $G=U(N)$. The chemical potential $\int_C\,C_{(2)}\wedge\Tr(F)$ for the D2-branes requires taking the Ramond-Ramond field $C_{(2)}$ proportional to the two-forms $k^i$ on $C$ which are dual to the basis two-cycles $S_i$, in order to get the correct charges. In this case the D0-brane charges
$$
\calq_0 \= \frac1{8\pi^2}\,\int_C\,\Tr(F\wedge F)
$$
correspond to the instanton numbers of the gauge bundle $\calp$, while the D2-brane charges
$$
 \calq_2^i \= \frac1{2\pi}\,\int_{S_i}\,\Tr(F)
$$
correspond to non-trivial magnetic fluxes $c_1(\calp)\neq0$ through $S_i$. To compute the macroscopic black hole entropy from the counting of the corresponding BPS states in the gauge theory, we introduce observables associated to these sources and compute their gauge theory expectation values using the localization arguments above to get
\begin{eqnarray*}
Z_{\rm
  BH}(N,\mbf\phi^2,\phi^0) &=&\Big\langle\exp\Big(-\frac{\phi^0}{8\pi^2}\,\int_C\,\Tr(F\wedge
F) -\frac{\phi_i^2}{2\pi}\,\int_C\,k^i\wedge\Tr(F)\Big)\Big
\rangle_{\rm SYM} \\[4pt] &=&
\sum_{\calq_0,\calq_2^i}\,\Omega(\calq_0,\mbf\calq_2;N)~
\e^{-\calq_0\,\phi^0-\calq_2^i\,\phi_i^2} \ ,
\end{eqnarray*}
where $\Omega(\calq_0,\mbf\calq_2;N)$ is the Witten index which computes the Euler
characteristic of the moduli space $\calm_{N,\mbf\calq_2,\calq_0}(C)$ of $U(N)$ instantons on $C$ with Chern
invariants $c_1(\calp)=\mbf\calq_2\in H^2(X,\Z)$ and $-\ch_2(\calp)=\calq_0\,\nu\in H^4(X,\Z)$. Here $\nu$ is the generator of $H^4(X,\Z)\cong\Z$ which is Poincar\'e dual to a point in $X$.

\subsection{Toric localization and the instanton moduli space\label{ADHMC2}}

Instantons on $C=\C^2$ can be described as follows. Consider the
quiver
$$
\xymatrix{
v \ \bullet \ \ar@(ur,ul)|{\, b_1 \,} \ar@(dl,dr)|{\, b_2 \,} \ar@/^/[rr]|{\, j \,} &&
\ \bullet \ w \ar@/^/[ll]|{\, i \,}
}
$$
with the single relation $r$ specified by the linear combination of paths
$$
r\= [b_1,b_2]+i\,j \ .
$$
This is called the ADHM quiver ${\sf Q}_{\rm ADHM}$. The \emph{ADHM
  construction} establishes a one-to-one correspondence between stable
framed representations of the quiver ${\sf Q}_{\rm ADHM}$ in the
category $\Vect_\C$ of finite-dimensional complex vector spaces and framed torsion free sheaves
on the projective plane $\PP^2$. In the rank one case, these are
equivalent to ideal sheaves on $\C^2$, and the correspondence gives
an isomorphism with the Hilbert schemes of points on $\C^2$.

Let $V$ and $W$ be inner product spaces of complex dimensions $k=\calq_0$ and $N$,
respectively. The
instanton moduli space $\calm_{N,k}(\C^2)$ can be realized as a hyperK\"ahler
quotient by the natural action of $U(k)$ on the variety consisting of linear operators
$$
 B_{1},B_{2} ~\in~ \Hom_\C(V,V) \ , \qquad I ~\in~ \Hom_\C(W,V) \ ,
\qquad J ~\in~ \Hom_\C(V,W) 
$$
constrained by the ADHM equations
\bea 
\mu_c &=& [B_1 , B_2] + I \, J \= 0 \ , \nonumber\\[4pt] \mu_r
&=& \big[B_1 \,,\, B_1^{\dagger}\,\big] + \big[B_2 \,,\,
B_2^{\dagger}\, \big] + I\,I^{\dagger} -
J^{\dagger}\, J \= 0 \ . 
\label{adhm}\eea
On $\C^2$ one can obtain a better compactification of this
moduli space by deforming the gauge theory to a noncommutative field
theory, as before~\cite{NS,KS,DN}. This is equivalent to a modification of the hyperK\"ahler
quotient that defines the instanton moduli space, obtained by changing the images of the moment maps of (\ref{adhm})
to 
$$
\mu_c \= 0 \ , \qquad \mu_r \= \zeta~ \Id_{V} \ ,
$$
where $\zeta=\theta_1+\theta_2$. This quotient gives a compactification
of the instanton moduli space ${\mathfrak{M}}_{N,k}(\C^2)$ obtained by blowing up its
singularities.

The classification of toric fixed points is given in ref.~\cite{Nakajima}, by identifying the
instanton moduli space $\calm_{1,k}(\C^2)$ with the Hilbert scheme
of points $(\mathbb{C}^2)^{[k]}$. The fixed points
are point-like instantons which are in one-to-one correspondence
with Young tableaux $\lambda$ having $|\lambda|=k$ boxes. In the more general case
of a $U(N)$ gauge theory in the Coulomb branch, one takes $N$ copies
of the $U(1)$ theory and the fixed points are
classified in terms of $N$-tuples of Young diagrams $\vec{\lambda} =
(\lambda_1 , \dots , \lambda_N)$, called $N$-coloured Young diagrams. One can show that the fixed points
are isolated.

One can describe the local structure of the instanton moduli space by
using the ADHM construction at the fixed points of the $\tilde
T=U(1)^N\times T^2$ action. Following ref.~\cite{Nakajima} one introduces a two-dimensional $T^2$-module $\underline{Q}$ to keep track of the toric
action. The operators $(B_1 ,
B_2 , I , J )$ corresponding to a fixed point configuration are
elements of the $\tilde T$-modules
$$
 (B_{1},B_{2}) ~\in~ \Hom_\C(V,V)\otimes \,\underline{Q} \ , \qquad I ~\in~ \Hom_\C(W,V) \ ,
\qquad J ~\in~ \Hom_\C(V, W)\otimes \mbox{$\bigwedge^2$}\,\underline{Q}
 \ .
$$
Then the local structure of the instanton moduli space
is described by the complex
\begin{equation} \label{adhmdefcomplex}
\xymatrix@1{
  \Hom_\C(V , V)
   \quad\ar[r]^{\!\!\!\!\!\!\!\!\!\!\!\!\!\!\!\!\!\!\sigma} &\quad
   {\begin{matrix} \Hom_\C(V ,V)\otimes \,\underline{Q} 
   \\ \oplus \\
   \Hom_\C(W , V) \\ \oplus  \\ \Hom_\C(V ,
   W) \otimes \bigwedge^2\,
   \underline{Q}  \end{matrix}}\quad \ar[r]^{~~\tau} & \quad
   {\begin{matrix} \Hom_\C(V , V) \otimes \bigwedge^2\,
       \underline{Q}
   \end{matrix}}
}
\end{equation}
which is just a finite-dimensional version of the
Atiyah--Hitchin--Singer instanton deformation complex (\ref{AHScomplex}). The 
map $\sigma$ corresponds to infinitesimal (complex) gauge transformations
while $\tau$ is the linearization of the ADHM constraint
$\mu_c = 0$. In general, the complex (\ref{adhmdefcomplex}) has three
nonvanishing cohomology groups. In our case we can safely assume that
$H^0$ and $H^2$ vanish. The only nonvanishing cohomology $H^1$ 
describes field configurations that obey the linearized ADHM
constraint $\mu_c = 0$ but are not gauge variations. It is thus a
local model for the tangent space to the instanton moduli space at each $\tilde T$-fixed point. Later on we will
compute weights of the toric action on the tangent space
modelled on (\ref{adhmdefcomplex}).

The partition function of the $U(1)$ topologically twisted gauge theory on
$X=\C^2$ is easily computed. The only non-trivial
topological charge is the instanton number $k=-\int_X\, F \wedge F $ and therefore the partition function has the form
$$
Z_{U(1)}^{\C^2}(q) \= \sum_{k=0}^\infty\, q^k~
\chi\big(\calm_{1,k}(\C^2)\big) \ .
$$
The expansion parameter can be identified in terms of gauge theory
variables $q:=\e^{2\pi\ii\tau}$ with
$$
\tau \=\frac{4\pi\ii}{g_{\rm YM}^2}+\frac\vartheta{2\pi}
$$
the complexified gauge coupling, which is related to topological
string variables $g_s=g_{\rm YM}^2/2$ at the attractor point.
At a toric fixed point, $k$ is identified as the number of boxes in a
partition $\lambda$. The Euler classes exactly cancel in the localization
formula (\ref{ABloc}), and one is left with the sum over fixed
points~\cite{BFMT}
$$
\chi\big(\calm_{1,k}(\C^2)\big) \=
\sum_{\lambda \,:\, |\lambda|=k}\, 1 \ .
$$
By Euler's formula, one has
$$
{\widehat{\eta}(q)^{-1} ~:=~ \prod_{n=1}^\infty\,\frac1{1-q^n}\=
    \sum_{N=0}^\infty\,p(N)\,q^N}
$$
where $p(N)$ is the number of partitions
$\lambda=(\lambda_1,\lambda_2,\dots)$ (ordinary, two-dimensional Young
tableaux) of degree $|\lambda|=\sum_i\,\lambda_i=N$. The function
$\widehat{\eta}(q)$ is related to the Dedekind function. It follows that the $U(1)$ partition function
$$
Z_{U(1)}^{\C^2}(q)\= \widehat\eta(q)^{-1}
$$
is the generating function for two-dimensional Young diagrams.

This construction can be easily generalized to the nonabelian
case. The fixed points are now $N$-coloured Young tableaux $\vec\lambda = ( \lambda_1 , \dots , \lambda_N) $ corresponding to a
partition of the instanton number $k = (k_1 , \dots , k_N)$. The instanton action is again equal to $k$, where the additional
factor of $N$ arising from the sum over colours $l=1,\dots,N$ cancels with the
normalization of the $F \wedge F$ term which carries a factor
$\frac{1}{N}$ (the inverse of the dual Coxeter number of the gauge group $G=U(N)$). The Euler characteristic of instanton moduli space is now
\begin{equation*}
\chi\big(\calm_{N,k}(\C^2)\big) \=
\sum_{\vec\lambda \, : \, |\vec\lambda|=k}\, 1
\end{equation*}
with $|\vec\lambda|:=\sum_l\,|\lambda_l|$. The $U(N)$ partition function reduces to $N$ copies of the $U(1)$ partition function,
$$
Z_{U(N)}^{\C^2}(q)\=\big(Z_{U(1)}^{\C^2}(q)\big)^N \ .
$$
This factorization follows from the fact that after toric
localization, the gauge symmetry $U(N)\to U(1)^N$ is broken to the
maximal torus. The Coulomb phase corresponds to well separated
D4-branes, but the topological nature of the gauge theory ensures that
the partition function is independent of the Higgs moduli representing
the lengths of open strings stretching between D-branes. In the rest of this section we will extend these constructions to generic toric
surfaces $X$.

\subsection{Hirzebruch--Jung spaces\label{HJspaces}}

Our main example will be the most general toric singularity in four dimensions, which defines a class of toric Calabi--Yau (hence open) surfaces known as Hirzebruch--Jung spaces $C=C(p,n)$. They are determined by two relatively prime positive integers $p$ and $n$ with $p>n$. Consider the quotient singularity $\bbc^2/\Gamma_{(p,n)}$, with the generator of the cyclic group $\Gamma_{(p,n)}\cong\bbz_p$ acting on $(z,w)\in\C^2$ as
$$
\big(z\,,\,w\big)
~\longmapsto~
\big(\e^{2\pi \ii n/p}\,z\,,\,\e^{2\pi \ii /p}\,w\big) \ .
$$
This orbifold has an $A_{p,n}$ singularity at the origin of $\C^2$. Then $C(p,n)$ is defined to be the minimal resolution of the $A_{p,n}$ singularity by a chain of $\ell$ exceptional divisors
  $S_i\cong\bbP^1$ whose intersection numbers are summarized in the intersection matrix
$$
\mbf C\=\begin{pmatrix} - e_1 & 1 & 0 & \cdots &0\\
1 & -e_2  & 1& \cdots &0\\
0 & 1 & - e_3 &\cdots&0\\
\vdots &\vdots & \vdots &\ddots&\vdots\\0&0&0&\cdots&-e_\ell
\end{pmatrix} \ ,
$$
which is called a \emph{generalized Cartan matrix}. The divisors thus
only intersect transversally with their nearest neighbours in the
chain. The self-intersection numbers $e_i\geq2$ of the spheres $\PP^1$ of the blow-up are determined from the continued fraction expansion
$$
\frac pn\=
e_1-{1\over\displaystyle e_{2}- {\strut
1\over \displaystyle e_{3}- {\strut 1\over\displaystyle\ddots {}~
e_{\ell-1}-{\strut 1\over e_\ell}}}} \ .
$$
Let us consider two particular well-known instances of these spaces.

\subsubsection*{Local $\PP^1$}

Setting $n=1$, the space $C(p,1)$ can be identified with the total space of the holomorphic line bundle $\calo_{\bbP^1}(-p)$ over $\bbP^1$ of degree $-p$, with
  $\ell\=1$ and $e_1\=p$. In this case $S=\bbP^1$ is the zero section divisor. In the context of topological string theory, such four-cycles appear in the  
``local'' Calabi--Yau threefolds
$X$ which are regarded as neighbourhoods of a holomorphically embedded
rational curve in a compact Calabi--Yau threefold, i.e. as the normal
bundle $\mathcal{N}\rightarrow\bbP^1$. Since $\cN$ is a holomorphic
vector bundle of rank two over $\PP^1$ and the Calabi--Yau condition implies
$c_1(\mathcal{N})=-\chi(\PP^1)=-2$, it follows that $X$ is the total space of a bundle of the form
$\calo_{\bbP^1}(-p)\oplus \calo_{\bbP^1}(p-2)\rightarrow\bbP^1$.

\subsubsection*{$A_{p-1}$ ALE space}

The complex surface $C=C(p,p-1)$ is an example of an asymptotically
locally Euclidean (ALE) space. This means that $C$ carries a scalar
flat K\"ahler metric $g$ such that $(C,g)$ is complete, and there
exists a compact set $K$ such that $C\setminus
K\cong(\R^4\setminus\overline{B_R}\,)/\Z_p$. Here $\Z_p\subset O(4)$
acts freely on $\R^4\setminus\overline{B_R}$ and the metric $g$
approximates the flat Euclidean metric on $\R^4$. Such a coordinate
system is called a coordinate system at infinity. We regard
$\Z_p\subset U(2)$ acting on $(z,w)\in\C^2\cong\R^4$ as described
above (with $n=p-1$), and the complex structure $I$ on $C$
approximates that on $\C^2=\R^4$. As the resolution of the Klein
singularity $\C^2/\Z_p$, $C(p,p-1)$
contains a chain of $\ell=p-1$ projective lines $\bbP^1$, each with
self-intersection number $e_i=2$. In this case, the intersection
matrix $\mbf C$ coincides with the 
Cartan matrix of the $A_{p-1}$ Dynkin diagram.

\subsection{Instantons on ALE spaces}

We begin by describing in some detail the instanton
moduli space in the case of the $A_{p-1}$ ALE spaces, for which a
rigorous construction is known. 
$U(N)$ instantons on ALE spaces are given by the ADHM
construction. Since the topological gauge theory is invariant under
blow-ups of the surface (using blow-up formulas), one can do the
instanton computation on the orbifold $\C^2/\Gamma$ where
$\Gamma=\Gamma_{(p,p-1)}\cong\Z_p$. This result is at the heart of the
McKay correspondence which provides a one-to-one correspondence
between irreducible representations of the orbifold group $\Gamma$ and
tautological bundles over the exceptional divisors of the minimal
resolution $C=C(p,p-1)$.

Since $C$ is non-compact, the instanton moduli space $\calm_C$ must be
defined with respect to connections which have appropriate asymptotic
decay at infinity. We will describe this in more generality later on
in terms of framed moduli spaces of torsion free sheaves. These connections correspond to
instantons of finite energy, and are asymptotic to flat connections
with $F=0$. In particular, there are solutions
which have fractional first Chern class and are related to instantons
that asymptote to flat connections with non-trivial holonomy at the
boundary of $C$, which is topologically the Lens space $L(p,p-1)= S^3/\Gamma$. The flat connections are classified by homomorphisms
  $\rho:\pi_1(C)\rightarrow U(N)$, where
  $\pi_1(C)=\Gamma\cong\bbz_p$. The asymptotic connection at infinity
  is thus labelled by irreducible representations $(k_0,k_1,\dots,k_{p-1})$ of
  the orbifold group $\bbz_p$, with $\sum_i\,k_i=N$, and are given
  explicitly by $$\rho_k\big(\e^{2\pi\ii
    /p}\big)\= \e^{2\pi\ii k/p}$$
where $k=0,1,\dots,p-1$.

Starting from the ADHM construction on $\C^2$ outlined
in Section~\ref{ADHMC2}, one constructs its $\Gamma$-invariant
decomposition. Consider the universal scheme $\cZ \subset C \times
\C^2$ given by the correspondence diagram
$$
\xymatrix@=15mm{
  & \cZ \ar[ld]_{q_1}\ar[rd]^{q_2}& \\
  C & & \C^2 \ .
}
$$
The tautological bundle on $C$ is defined by
$$
\cR \ := \ q_{1 *} \cO_{\cZ} \ .
$$
Under the action of $\Gamma$ on $\cZ$, $\cR$ transforms in the
regular representation and can thus be decomposed into
irreducible representations
$$
\cR \= \bigoplus_{k=0}^{p-1}\, \cR_k \otimes \rho_k \ , \qquad \cR_k\=
\Hom_\Gamma(\cR,\rho_k) \ .
$$
By the McKay correspondence, the bundles
$\cR_0=\cO_C,\cR_1,\dots,\cR_{p-1}$ form the canonical integral basis
of the K-theory group $K^0(C)$ constructed in ref.~\cite{G-SV}.

In this case we take $\underline{Q}\cong\C^2$ to be a module on which
the regular representation of $\Gamma$ acts. We also take $\Gamma
\subset SU(2) $ so that the determinant
representation is trivial as a $\Gamma$-module, i.e. $\bigwedge^2\,
\underline{Q} \otimes_\Gamma \cR \cong \cR$. The two vector spaces $V$
and $W$ which feature in the ADHM construction have a natural grading
under the action of the orbifold group $\Gamma$ given by
$$
V\= \bigoplus_{k=0}^{p-1}\,V_k\otimes\rho_k \ , \qquad
W\= \bigoplus_{k=0}^{p-1}\,W_k\otimes\rho_k \ .
$$
The modification of (\ref{adhmdefcomplex}) is given
by~\cite{KN}
\begin{equation} \label{compcomplex}
\xymatrix@1{
  \Hom_{\Gamma} (\mathcal{R}^* , V)
   \quad\ar[r]^{\!\!\!\!\!\!\!\!\!\!\!\!\!\!\sigma} &\quad
   {\begin{matrix} \Hom_{\Gamma} (\mathcal{R}^* , \,\underline{Q} \, \otimes_\Gamma V)
   \\ \oplus \\
   \Hom_{\Gamma} (\mathcal{R}^*, W) \end{matrix}} \quad \ar[r]^{~~~~~\tau} & \quad
   {\begin{matrix} \Hom_{\Gamma} (\mathcal{R}^* , V ) \ ,
   \end{matrix}}
}
\end{equation}
and the condition that the sequence (\ref{compcomplex}) is a complex is equivalent
to the (generalized) ADHM equations. After imposing a certain
stability condition, this
construction realizes the instanton moduli space as a quiver
variety $\calm(V,W)$.

This construction identifies two distinct types of instanton
contributions to the ALE partition function, which we consider in turn. As before, after
  toric localization the gauge symmetry breaks as $U(N)\rightarrow
  U(1)^N$ and the $U(N)$ partition function factorizes as
$$
Z^{\rm ALE}_{U(N)}(q,\mbf Q)\=\big(Z^{\rm ALE}_{U(1)}(q,\mbf Q)\big)^N
\ .
$$
It therefore suffices to focus on the $U(1)$ case in the following.

\subsubsection*{Regular instantons}

{\it Regular instantons} on $A_{p-1}$ live in the regular representation
  $k_0=k_1=\dots=k_{p-1}=k$ of the orbifold group
  $\Gamma=\bbz_p$. They correspond to D0-branes moving freely on $C$ with $p$ orbifold
  images away from the orbifold point. For gauge group $U(1)$, the
  moduli space is given by specifying $K=k\,p$ points on 
  $C$ up to permutations. Hence the moduli space $
\calm_{\rm reg}^{U(1)}(C)$ of regular $U(1)$ instantons on $C$ is
isomorphic to the Hilbert scheme $C^{[K]}$. The generating function for the Euler numbers of the
instanton moduli spaces can then be computed explicitly by applying
G\"ottsche's formula to get
$$
{ Z_{\rm reg}^{U(1)}(q)\=\sum_{K=0}^\infty\,q^K\,\chi\big(C^{[K]}\big)\=
\widehat{\eta}(q)^{-p} } \ .
$$
The $U(N)$ partition function is the $N$-th power of this
quantity. {Heuristically, we may think of this formula as originating
  by covering $C$ with $p=\chi(C)$ open charts to get
  $p$ copies of $U(N)$ instantons on $\C^2$, each contributing
  $\widehat{\eta}(q)^{-N}$. This can be demonstrated rigorously on any toric surface $C$ by a localization computation~\cite{ES,CK-PS}.

\subsubsection*{Fractional instantons}

To each irreducible representation $(k_0,k_1,\dots,k_{p-1})$ of the
orbifold group $\bbz_p$ there corresponds
  a {\it fractional instanton} which is stuck at the orbifold
  points. It has {\it no} moduli (or orbifold images), and can be
  regarded as a state in which open strings ending on the same
  D0-brane are projected out by the action of the orbifold group. They
  carry magnetic fluxes through the $\bbP^1$'s of the minimal
  resolution, and correspond to self-dual
  $U(1)$ gauge connections with curvatures
$$
F\=-2\pi\ii u_i\,c_1(\cR_i) \ ,
$$
where $u_i\in\Z$ and $\cR_i=\Ocal_{\bbP^1}(e_i)$ are the tautological
line bundles. The Chern classes $c_1(\cR_i)$, $i=1,\dots,p-1$ form a
basis of $H^2(C,\bbz)$. Fractional instantons can thus be thought of as
Dirac monopoles on the two-spheres of the orbifold resolution.

The corresponding intersection numbers are given by
\beq
\int_C\,c_1(\cR_i)\wedge c_1(\cR_j)\=-\big(\mbf C^{-1}\big)_{ij} \
, \qquad \int_{S_i}\,c_1(\cR_j)\=\delta_{ij} \ .
\label{tautintnum}\eeq
Since $C$ is non-compact, the intersection matrix $\mbf C$ is not
necessarily unimodular, and the corresponding instanton charges
can be fractional. The contribution of fractional instantons to the
supersymmetric Yang--Mills action with observables is thus given by
\begin{eqnarray*}
S_{\rm frac} &=& -\frac{\ii\tau}{4\pi}\,\int_C\,F\wedge F-
\frac{\ii\phi_j^2}{2\pi}\,\int_C\,F\wedge c_1(\cR_j) \nonumber\\[4pt]
&=& -\pi\ii\tau\,\big(\mbf C^{-1}\big)^{ij}\,u_i\,u_j+z_i\,u_i 
\end{eqnarray*}
with $z_i=(\mbf C^{-1})_{ij}\,\phi_j^2$. Setting $\mbf
u:=(u_1,\dots,u_{p-1})$ and identifying $Q_i=\e^{-z_i}$ using the
attractor mechanism, we find that the contribution of $U(1)$
fractional instantons to the full partition function is given by a theta-function 
$$
Z_{\rm frac}^{U(1)}(q,\mbf Q)\=\sum_{\mbf u\in\bbz^{p-1}}\,
q^{\frac12\,\mbf u\cdot\mbf C^{-1}\mbf u}~\mbf Q^{\mbf u}
$$
on a Riemann surface of genus $g=p-1$ and period matrix $\tau\,\mbf C^{-1}$.

\subsection{Instantons on local $\PP^1$}

Let us now discuss what is known beyond the ALE case, in the instance that
$C$ is the total space of the holomorphic line bundle
$\Ocal_{\PP^1}(-p)$~\cite{Nakajima3}. In this case an ADHM construction is not
available. Nevertheless, much of the construction in the ALE case
carries through, and by introducing weighted Sobolev norms, one shows
that $\calm_C$ is a smooth K\"ahler manifold with torsion-free
homology groups which vanish in odd degrees. Let us consider some
explicit examples.

The zero section of $\Ocal_{\PP^1}(-p)$, considered as a divisor
$S=\PP^1$ of $C$, produces a line bundle $\Lcal\to C$ such that
$c_1(\Lcal)$ is the generator of $H^2(C,\Z)=\Z$. It has a
unique self-dual connection asymptotic to the trivial connection. Let
$\underline{\C}=C\times\C$ be the trivial line bundle over $C$. Set
$\cale=\underline{\C}\oplus\Lcal$, and let $\calm(\cale)$ be the
moduli space of self-dual connections on $\cale$ which are asymptotic
to the trivial connection at infinity. Then
$\dim_\R(\calm(\cale))=2p$. Since $H_1(C,\R)=0$, using Morse theory
one can show that the only non-vanishing homology groups of the
instanton moduli space are
$H_0(\calm(\cale),\R)=H_2(\calm(\cale),\R)=\R$.

Alternatively, set
$\cale=\Lcal\oplus\Lcal^\vee$, and let $\calm_{(k)}(\cale)$,
$k=0,1,\dots,p-1$ be the moduli space of self-dual connections
asymptotic to $\rho_k\oplus\rho_k^*$. Then for any $p>2$,
$\dim_\R(\calm_{(k)}(\cale))=2$, while for $p=2$ (whereby
$\Ocal_{\PP^1}(-2)$ coincides with the $A_1$ ALE space) one has
$\dim_\R(\calm_{(k)}(\cale))=4$. In particular, for $p=2$ there is a
diffeomorphism $\calm_{(k)}(\cale)\cong T^*\PP^1$ with
$H_0(\calm_{(k)}(\cale),\R)=H_2(\calm_{(k)}(\cale),\R)=\R$, while for
$p=4$ one has $\calm_{(k)}(\cale)\cong B^2$. Instantons on local $\PP^1$ will be studied in more generality later on in terms of moduli spaces of framed torsion-free sheaves.

\subsection{Wall-crossing formulas\label{D4wallcrossing}}

In ref.~\cite{FMP} it was suggested that the structure of
the instanton partition function on ALE spaces can be extrapolated to
give a general result valid for all Hirzebruch--Jung surfaces
$C=C(p,n)$. Thus one postulates the form of the $U(N)$ partition
function
\begin{eqnarray}
Z_{U(N)}^C(q,\mbf Q)&=&\big(Z_{\rm reg}^{U(1)}(q)\,Z_{\rm
  frac}^{U(1)}(q,\mbf Q)\big)^N \nonumber \\[4pt]
&=& \frac1{\widehat{\eta}(\tau)^{N\,\chi(C)}}~
\sum_{\vec{\mbf  u}\in\bbz^{N\,b_2(C)}}\,
q^{\frac12\,\vec{\mbf u}\cdot\mbf C^{-1}\vec{\mbf u}}~
\mbf Q^{\mbf u}
\label{HJpartfn}\end{eqnarray}
where
\beq
\mbf u \ := \ \sum_{l=1}^N\,\mbf u_l \ . 
\label{uidef}\eeq
In the ALE case, one has $\chi(A_{p-1})=p$ and $b_2(A_{p-1})=p-1$. The
evidence for this formula is supported by calculations in the reduction to $q$-deformed
two-dimensional Yang--Mills theory~\cite{GSST,BT}, which captures the
contributions from fractional instantons on the exceptional divisors. It was even
conjectured to hold in the case when $C$ is a \emph{compact} toric surface, at least in
the region of moduli space where the instanton charges are
large. Though the regular and fractional instantons are always readily
constructed exactly as in the ALE case, the issue is whether or not this
formula takes into account \emph{all} of the instanton
contributions, and in which regions of moduli space the factorization
into $U(1)$ partition functions hold. Later on we will give more precise meanings to these
``asymptotic charge regions'' in terms of moduli
spaces of torsion free sheaves.

Curiously, the formula (\ref{HJpartfn}) coincides with the chiral torus partition function of a conformal field theory in two dimensions with central charge
$c=N\,\chi(C)$, i.e. of $N\,\chi(C)$ free bosons, with $N\,b_2(C)$ of
them compact. The compact degrees of freedom live in a torus
determined by the lattice $H^2(C,\Z)$ with the bilinear form $\mbf
C$. The appearence of this two-dimensional field theory can be understood
in M-theory, wherein the D4--D2--D0 brane quantum mechanics lifts to a
$(4,0)$ two-dimensional superconformal field theory on an M5-brane
worldvolume~\cite{MSW}. This is reminescent of the recent conjectural relations
between four-dimensional superconformal gauge theories and
two-dimensional Liouville conformal field
theories~\cite{AGT,Wyllard}. When $C$ is an $A_{p-1}$ ALE space and
the gauge group is $SU(N)$, there is yet another conformal field
theoretic interpretation~\cite{Nakajima3,VW,DHSV,DS}. In this case, the level $N$ affine
$\widehat{\mathfrak{su}(p)}$ Lie algebra acts on the cohomology ring of the
instanton moduli space, and the partition function (with appropriate
local curvature and signature corrections inserted) coincides with the
character of the affine Kac--Moody algebra given by
$$
Z_{SU(N)}^{\rm ALE}(q,\mbf Q) \= \sum_{n=0}^\infty~\sum_{\mbf
  u\in\Z^{p-1}}\, \Omega(n,\mbf u)~ q^{n-c/24}~\mbf Q^{\mbf u} \=
\Tr_\hil\big(q^{L_0-c/24}~\mbf Q^{\mbf J_0}\big) \ ,
$$
where $\mbf u=c_1(\calp)$, $n=c_2(\calp)=k+\frac12\,\mbf u\cdot\mbf
C\mbf u$, and $\Omega(n,\mbf u)$ are the degeneracies of
BPS states with the specified quantum numbers. Here $\hil$ is the
Hilbert space on which the chiral algebra acts, which can be
represented in terms of free fermion or boson conformal field theories
with extended symmetry generators $\mbf J_0$~\cite{DHSV}. The $U(1)$
partition function is also expressed as a
$\widehat{\mathfrak{u}(1)_1}$ character in refs.~\cite{DHSV,DS}.

Let us now compare the instanton partition function (\ref{HJpartfn}) with the
black hole partition function. Microscopic black hole entropy formulas for BPS bound
states of D0--D2--D4 branes in Type~IIA
supergravity are readily available on \emph{compact} Calabi--Yau
threefolds $X$~\cite{MSW}, in the large volume limit and when
contributions from worldsheet instantons are negligible. In this limit
we can expand the cycle $[C]$ as in (\ref{amplediv}). Let $\alpha_i$,
$i=1,\dots,b_2(X)$ be an integral basis of two-cocycles for $H^2(X,\Z)$
dual to the four-cycles $[C_i]$ with the intersection numbers
$$
D_{ijk}\=\frac16\,\int_X\,\alpha_i\wedge \alpha_j\wedge \alpha_k \ , 
\qquad c_{2,i}\=\int_X\,\alpha_i\wedge c_2(X) \ ,
$$
and let $D^{ij}$ be the matrix inverse of
$D_{ij}=D_{ijk}\,\calq_4^k$. The genus zero topological string
amplitude $F_0$ can be expressed as
$$
F_0\= D_{ijk}\,\frac{X^i\,X^j\,X^k}{X^0} \ .
$$
With $\calq_0=\frac1{8\pi^2}\,\int_C\,F\wedge F$ and
  $\calq_i=\calq_2^i=\big(\mbf C^{-1}\big)^{ij}\,u_j$, the black hole
  entropy is given by~\cite{MSW}
\beq
S_{\rm BH}(\calq_0,\mbf\calq,\mbf \calq_4) \= 2\pi\,\sqrt{\mbox{$\big(D_{ijk}\,\calq_4^i\,\calq_4^j\, \calq_4^k+
\frac16\,c_{2,i}\,\calq_4^i\big)\,\big(\calq_0+\frac1{12}\,D^{ij}\,
\calq_i\,\calq_j\big)$}} \ .
\label{SBHMSW}\eeq
We would now like to interpret the gauge theory partition function
(\ref{HJpartfn}) as the corresponding black hole partition function $Z_{\rm
  BH}(\mbf\calq_4,\mbf\phi^2,\phi^0)$. For this, we expand it as
$$
Z_{U(N)}^C(q,\mbf Q)\= \sum_{\calq_0,\calq_i}\,\Omega(\calq_0,\mbf\calq,\mbf \calq_4)~
\e^{-\calq_0\,\phi^0-\mbf\calq\cdot\mbf\phi^2} \ .
$$
Then Cardy's formula gives the black hole entropy as
$$
S_{\rm BH}(\calq_0,\mbf\calq,\mbf \calq_4) \= \log\Omega(\calq_0,\mbf\calq,\mbf \calq_4) \ .
$$
This expression agrees~\cite{FMP} with the macroscopic
supergravity result (\ref{SBHMSW}) for the
  Bekenstein--Hawking--Wald entropy in the large $\calq_0$ limit. Wall-crossing issues in a similar context are discussed in more detail in refs.~\cite{DenefMoore,DM,dBDE-SMVdeB,AM,GMN}.

In parallel to what we did in Section~\ref{Wallcrossing}, in order to
explore the pertinent wall-crossing formulas in this instance we will now
consider moduli spaces which parametrize isomorphism classes of the
following objects:
\begin{itemize}
\item[(a)] Surjections (framings) $\Ocal_C\to\calf\to0$ of torsion
  free sheaves with $\ch(\calf)=(N,d,-k)$;
\item[(b)] Stable torsion free sheaves $\cale$ on $C$ with
  $\ch(\cale)=(N,d,k)$; and
\item[(c)] Closed subschemes $S\subset C$ of dimension $\leq1$ with
  dual curve class $[S]^\vee=d$ and holomorphic Euler characteristic
  $\chi(\Ocal_S)=k$.
\end{itemize}
In contrast to the six-dimensional situation, in four dimensions the
connections between these three classes of objects is somewhat more
subtle. We shall examine each of them in turn and how they compare
with the gauge theory results we have thus far obtained.

\subsection{Moduli spaces of framed instantons\label{modspframed}}

We begin with Point~(a) at the end of Section~\ref{D4wallcrossing} Let $C$ be a smooth, quasi-projective, open,
toric surface. We will assume that $C$ admits a projective
compactification $\,\overline{C}\,$, i.e. $\,\overline{C}\,$ is a smooth, compact, projective,
toric surface with a smooth divisor $\ell_\infty\subset
\,\overline{C}\,$ (called the ``line at infinity'') which is a $T^2$-invariant
$\PP^1$ in $\,\overline{C}\,$, and such that $C=\,\overline{C}\,\setminus\ell_\infty$. We will
also require that $\ell_\infty\cdot\ell_\infty>0$, in addition to
$\ell_\infty\cong\PP^1$.
The difference between the counting of framed instantons on the
compact toric surface $\,\overline{C}\,$ (with boundary condition at ``infinity'')
and of unframed instantons on the open toric surface $C$ is a
universal perturbative contribution, which will be dropped here. Since
$\,\overline{C}\,$ is compact, these calculations will only capture the
contributions from instantons with integer charges on $C$. We will
mention later on how to incorporate the contributions from fractional
instantons on~$C$.

Let us fix two numbers $N\in\N$ and $k\in\Q$, and an integer
cohomology class $d\in H^2(\,\overline{C}\,,\Z)$. Let $\calm_{N,d,k}(C)$ be the
framed moduli space consisting of isomorphism classes $[\cale]$ of
torsion-free sheaves $\cale$ on $\,\overline{C}\,$ such that:
\begin{itemize}
\item[(1)] $\cale$ has the following topological Chern invariants:
\beq
N\=\ch_0(\cale)\={\rm rank}(\cale) \ , \qquad d\=\ch_1(\cale)\=c_1(\cale)  \qquad
\mbox{and} \qquad k\=-\int_{\,\overline{C}\,}\,\ch_2(\cale) \ ;
\label{topinvtsE}\eeq
and
\item[(2)] $\cale$ is locally-free in a neighbourhood of 
  $\ell_\infty$, and there is an isomorphism
  $\cale\big|_{\ell_\infty}\cong\Ocal_{\ell_\infty}^{\oplus N}$ called the
  ``framing at infinity''.
\end{itemize}
These topological conditions imply that the moduli space
$\calm_{N,d,k}(C)$ is non-empty only when
\beq
d\,\big|_{\ell_\infty}\=N\,c_1\big(\Ocal_{\ell_\infty}\big)\=0 \ ,
\label{dellinfty}\eeq
so that $d$ defines a class $d\in H_{\rm cpt}^2(C,\Z)$ in the
compactly supported cohomology of $C$. Furthermore, from
$\ch_2(\cale)=-c_2(\cale)+\frac12\,c_1(\cale)\wedge c_1(\cale)$ the last relation in
(\ref{topinvtsE}) can be written as
\beq
k\=\int_{\,\overline{C}\,}\,\big(c_2(\cale)-\mbox{$\frac{1}2$}\,d\wedge d\big) \ .
\label{kc2E}\eeq
Framed sheaves yield stable pairs~\cite{HL}, analogous to those described in Section~\ref{Wallcrossing}, after suitable choices of polarization on $\overline{C}$ and stability parameter.

We write $\cale(-\ell_\infty):=\cale\otimes\Ocal_{\,\overline{C}\,}(-\ell_\infty)$. At a
given point $[\cale]\in\calm_{N,d,k}(C)$, the space of reducible
connections is $\Ext_{\Ocal_{\,\overline{C}\,}}^0(\cale,\cale(-\ell_\infty))=\Hom_{\Ocal_{\,\overline{C}\,}}(\cale,\cale(-\ell_\infty))$,
the Zariski tangent space is $\Ext_{\Ocal_{\,\overline{C}\,}}^1(\cale,\cale(-\ell_\infty))$, and the
obstruction space is $\Ext_{\Ocal_{\,\overline{C}\,}}^2(\cale,\cale(-\ell_\infty))$. The cohomology of
the instanton deformation complex is greatly simplified by the fact
that in this case~\cite{GL}
\beq
\Ext_{\Ocal_{\,\overline{C}\,}}^0\big(\cale\,,\,\cale(-\ell_\infty)\big)\=
\Ext_{\Ocal_{\,\overline{C}\,}}^2\big(\cale\,,\,\cale(-\ell_\infty)\big)\=0 \ .
\label{Ext01van}\eeq
Using (\ref{Ext01van}) and the Riemann--Roch theorem, one
shows~\cite{GL} that the moduli space $\calm_{N,d,k}(C)$
is a smooth quasi-projective variety of (complex) dimension
$2N\,k+d^2$, where $d^2:=\int_{C}\,d\wedge d$, whose tangent space at
a point $[\cale]$ is isomorphic to the vector space
$\Ext_{\Ocal_{\,\overline{C}\,}}^1(\cale,\cale(-\ell_\infty))$.

We will now describe a natural torus invariant subspace of 
the instanton moduli space. The $T^2$-invariance of $\ell_\infty$
implies that the pullback of the $T^2$-action on $\,\overline{C}\,$ defines an
action on $\calm_{N,d,k}(C)$. There is also an action of the diagonal
maximal torus $T^N$ of $GL(N,\C)$ on the framing. Altogether we get
an action of the complex algebraic torus $\tilde T= T^2\times T^N$ on
$\calm_{N,d,k}(C)$~\cite{GL}. We are interested in the
fixed point set $\calm_{N,d,k}(C)^{\tilde T}$ of this torus action.

An isomorphism class $[\cale]\in\calm_{N,d,k}(C)$ is fixed by the $T^N$-action if and
only if it decomposes as 
\beq
\cale\=\cale_1\oplus\cdots\oplus \cale_N \ , \qquad \cale_l~\in~\calm_{1,d,k}(C)
\label{Edecomp}\eeq
such that $\cale_l\big|_{\ell_\infty}$ is mapped
to the $l$-th factor $\Ocal_{\ell_\infty}$ of
$\Ocal_{\ell_\infty}^{\oplus N}$ under the framing isomorphism. Since
the double dual $\cale_l^{\vee\vee}$ is a line bundle which is trivial on
$\ell_\infty$, it is equal to $\Ocal_{\,\overline{C}\,}(D_l)$ for some divisor
$D_l\subset \,\overline{C}\,$ disjoint from $\ell_\infty$. Via the natural
injection $\cale_l\subset \cale_l^{\vee\vee}=\Ocal_{\,\overline{C}\,}(D_l)$, the sheaf
$\cale_l$ is thus equal to $\cali_l(D_l)=\cali_l\otimes\Ocal_{\,\overline{C}\,}(D_l)$ for some
ideal sheaf $\cali_l$ of a zero-dimensional subscheme $Z_l\subset \,\overline{C}\,$,
also disjoint from $\ell_\infty$. If $\cale$ is also fixed by the
$ T^2$-action, then so are $D_l$, $\cali_l$ and~$Z_l$. 

The supports of the $ T^2$-invariant subschemes $Z_l$ and $D_l$ can be 
represented explicitly in terms of the toric geometry of $C$. Recall
that this is described in terms of an underlying toric graph
$\Delta(C)$ whose vertices are in bijective correspondence with the
$ T^2$-fixed points in $C$, and two vertices are joined by an edge $e$
if and only if the corresponding fixed points are connected by a
$ T^2$-invariant $\PP^1$. Let $V(C)$ and $E(C)$ respectively denote
the set of vertices and edges of $\Delta(C)$. Then since $\Delta(C)$
is a chain, one has $|E(C)|=|V(C)|-1$. The number of vertices
$n=|V(C)|$ is also the number of two-cones in the toric fan of
$C$ and it is related to the Euler characteristic
$\chi(C)$ of the surface as
$$
n\=\chi(C)\=\chi(\,\overline{C}\,)-2 \ .
$$
We denote by $p_f$ the $ T^2$-invariant point in $C$ corresponding to
the vertex $f\in V(C)$, and by $\ell_e$ the $ T^2$-invariant line
$\PP^1$ corresponding to the edge $e\in E(C)$.

Since the zero-cycles $Z_l$ are not supported on $\ell_\infty$, they
must be contained in the fixed point set $V(C)$ in $C$. Thus each
$Z_l$ is a union of subschemes $Z_l^f$, $f\in V(C)$ supported at the
$ T^2$-fixed points $p_f\in C$. If we choose a local coordinate system
$(x,y)\in\C^2$ in a patch $U_f$ around $p_f$, then the
$ T^2$-invariant ideal of $Z_l^f$ in the coordinate ring $\C[x,y]$ of
$U_f\cong\C^2$ is generated by the $ T^2$-eigenfunctions with
non-trivial characters, which are monomials $x^i\,y^j$, and hence 
$Z_l^f$ corresponds to a Young diagram $\lambda_l^f$ with
$\big|\,\lambda_l^f\,\big|$ boxes (with monomial $x^i\,y^j$ placed at
$(i+1,j+1)$). The ideal is spanned by monomials outside the Young
diagram. Likewise, since the $ T^2$-invariant two-cycles $D_l$ are
disjoint from the line at infinity, they are supported along the edges
$\ell_e\cong\PP^1$, $e\in E(C)$.

It follows that the fixed point set $\calm_{N,d,k}(C)^{\tilde T}$ is
parametrized by finitely many $N$-tuples
\beq
(\vec D,\vec{\mbf \lambda}) \=\big((D_1,\mbf \lambda_1)\,,\,\dots\,,\,(D_N,\mbf \lambda_N)\big)
\ ,
\label{DYNtuple}\eeq
where
\beq
D_l~\in~H_2(C,\Z) \ \cong \ \bigoplus_{e\in E(C)}\,\Z[\ell_e]
\label{Dlelle}\eeq
are $ T^2$-invariant divisors in $C$ and
$$
\mbf \lambda_l \=\big(\,\lambda_l^f\,\big)_{f\in V(C)}
$$
is a vector of Young tableaux with $\big|\,\mbf \lambda_l\,\big|:=\sum_{f\in 
  V(C)}\,\big|\,\lambda_l^f\,\big|$ boxes. The Young tableaux parametrize
the contributions from regular pointlike D0-brane instantons
(freely moving inside $C$). The divisors parametrize the contributions
from D2-brane instantons (wrapping $D_l$ with appropriate units of
magnetic flux). We will see that these two types of
contributions factorize completely. We can write the topological
invariants of a generic element (\ref{DYNtuple}) in
$\calm_{N,d,k}(C)^{\tilde T}$ in terms of this combinatorial
data. Since $c_1(\cali_l)=0$, the constraint $d=c_1(\cale)$ can be written as  
\beq
d \=\sum_{l=1}^N\,c_1\big(\Ocal_{C}(D_l)\big) \ ,
\label{c1EdDldvee}\eeq
whereas (\ref{kc2E}) becomes
\beq
k \=\sum_{l=1}^N\,\big|\,\mbf \lambda_l\,\big|+\int_{C}\,\Big(\,
\sum_{l<l'}\,c_1\big(\Ocal_{C}(D_l)\big)\wedge c_1\big(
\Ocal_{C}(D_{l'})\big)-\frac{1}2\,d\wedge d \, \Big) \ .
\label{KYlDl}\eeq
The Young diagram box sum in (\ref{KYlDl}) gives the length of the
singularity set of the sheaf (\ref{Edecomp}).

To make contact with our previous formulas, let us rewrite these constraints in a more explicit
parametrization. Define the invertible, symmetric, integer-valued
intersection matrix $\mbf C=(\mbf C_{e,e'})_{e,e'\in E(C)}$ between lines of the
toric graph $\Delta(C)$ as
$$
\mbf C_{e,e'} \ := \ \ell_e\cdot\ell_{e'} \ .
$$
Each divisor $\ell_e$ is canonically associated to a
$ T^2$-equivariant holomorphic line
bundle $\call_e:=\Ocal_{C}(\ell_e)$ whose first Chern class $c_1(\call_e)\in
H^2_{\rm cpt}(C,\Z)$ is Poincar\'e dual to $\ell_e$ with 
\beq
\int_{\ell_e}\,c_1\big(\call_{e'}\,\big)\=\mbf C_{e,e'} \qquad
\mbox{and} \qquad \int_{C}\,c_1\big(\call_e\big)\wedge c_1\big(
\call_{e'}\big)\=\mbf C_{e,e'} \ .
\label{Chern1PD}\eeq
Note that these intersection numbers differ from (\ref{tautintnum}),
i.e. the basis of line bundles $\call_e$ does \emph{not} coincide with the
tautological line bundles $\cR_i$.
Using (\ref{Dlelle}) we can write $c_1(\Ocal_{C}(D_l))=\sum_{e\in
  E(C)}\,u^e_{l}\,c_1(L_e)$ with $u^e_{l}\in\Z$. Then the
constraint (\ref{c1EdDldvee}) can be expressed as
\beq
d \=\sum_{l=1}^N~\sum_{e\in E(C)}\,u^e_{l}\,c_1\big(\call_e\big) \ .
\label{duel}\eeq
Combined with (\ref{KYlDl}) and using (\ref{Chern1PD}) this gives
\beq
k \= \sum_{l=1}^N\,\big|\,\mbf \lambda_l\,\big|+\sum_{l<l'}\,u^e_{l}\,
\mbf C_{e,e'}\,u^{e'}_{l'}-\frac12\,\sum_{l,l'=1}^N\,
u^e_{l}\,\mbf C_{e,e'}\,u^{e'}_{l'} \= \sum_{l=1}^N\,\big|\,\mbf \lambda_l\,\big|-
\frac12\,\sum_{l=1}^N\,u^e_{l}\,\mbf C_{e,e'}\,u^{e'}_{l} \ .
\label{kintsum}\eeq

The Vafa--Witten--Nekrasov partition function on the toric surface $C$
is now defined as the $\tilde T$-equivariant index
\beq
Z_{U(N)}^{C}(q,\mbf Q)^{\rm fr}\=\sum_{k\in\Q}\,q^k~
\sum_{d\in H_{\rm
    cpt}^2(C,\Z)}\,\mbf Q^d\,\chi\big(\calm_{N,d,k}(C)\big) \ , 
\label{VWNpartfndef}\eeq
where as before $q=\e^{2\pi\im\tau}$ with $\tau$ the complexified gauge coupling
constant and
$$
\mbf Q^d \ := \ \prod_{e\in E(C)}\,\big(Q_e\big)^{(\mbf C^{-1})^{e,e'}\,\int_{\ell_{e'}}\,d}
$$
for a collection of formal variables $Q_e$, $e\in E(C)$ and a given
cohomology class $d\in H_{\rm cpt}^2(C,\Z)$. From the localization formula (\ref{ABlocformula}) we see that the
$\tilde T$-equivariant Euler class cancels in (\ref{ABloc}) at the fixed points. This cancellation in the
fluctuation determinants is a consequence of (\ref{Ext01van}) which
implies that the obstruction bundle is trivial, and there is an
isomorphism between the tangent and normal bundles over the fixed
point set $\calm_{N,d,k}(C)^{\tilde T}$ in the instanton moduli
space. Because of this cancellation, we see that the contribution from
each critical point of the gauge theory is 
independent of the equivariant parameters $\epsilon_i$ and $a_l$. In the sum over critical points, we
can replace the sum over $k\in\Q$ using (\ref{kintsum}) by sums over
Young diagrams $\vec {\mbf\lambda}$ and over the integers
$\vec{\mbf u}= (\mbf u_{1},\dots,\mbf u_{N})$, $\mbf u_l=(u_l^e)_{e\in
E(C)}\in\Z^{n-1}$ representing magnetic fluxes through $\ell_e$ on D2-branes
wrapping the divisors $D_l$. The sum over $d\in H^2_{\rm cpt}(C,\Z)$
may then be saturated using~(\ref{duel}).

Putting everything together, and using (\ref{duel}) and
(\ref{Chern1PD}) to write
$$
\int_{\ell_e}\,d \= \sum_{l=1}^N\,\mbf C_{e,e'}\,u^{e'}_{l} \ , 
$$
the partition function (\ref{VWNpartfndef}) thereby becomes
\beq
Z_{U(N)}^{C}(q,\mbf Q)^{\rm fr} \=\sum_{\vec{\mbf\lambda}}~
\sum_{\vec{\mbf u}\in\Z^{N\,(n-1)}}\,
q^{|\,\vec{\mbf\lambda}\,| -\frac12\,\vec{\mbf u}\cdot \mbf C\vec{\mbf u}}~
\mbf Q^{\mbf u}
\label{VWNpartfnresum}\eeq
where we have used (\ref{uidef}). The sums over Young tableaux decouple, and for each vertex $f\in V(C)$ and
each $l=1,\dots,N$ they produce a factor of the Euler function
$\widehat\eta(q)^{-1}$. Then we can bring (\ref{VWNpartfnresum}) into the final form
\beq
Z_{U(N)}^C(q,\mbf Q)^{\rm fr} \=\frac1{\widehat\eta(q)^{N\,\chi(C)}}~
\sum_{\vec{\mbf u}\in\Z^{N\,b_2(C)}}\,q^{-\frac12\,\vec{\mbf u}\cdot
  \mbf C\vec{\mbf u}}~\mbf Q^{\mbf u}
\label{VWNpartfinal}\eeq
where we recall that $\chi(C)=|V(C)|=n$ and $b_2(C)=|E(C)|=n-1$.

This formula differs from (\ref{HJpartfn}) in that only integral
values of the first Chern class are permitted in
(\ref{VWNpartfinal}). In the case of an $A_{p-1}$ singularity, fractional first
Chern classes can be incorporated by constructing instead the moduli space of torsion free
sheaves on the orbifold compactification $\overline{C}\,^{\rm orb} = C \cup \tilde\ell_{\infty}$ of the hyperK\"ahler ALE
space $C$, where
$\tilde\ell_{\infty} = \PP^1 / \Gamma$~\cite{KN,Nakajima1,Nakajima2,Nakajima3}. In a neighbourhood
of infinity we can approximate $\overline{C}$ by $\PP^2 /
\Gamma$ with the singularity at the origin resolved. More
precisely, we obtain the divisor $\tilde\ell_{\infty}$ by gluing together
the trivial bundle $\cO_C$ on $C$ with the line bundle $\cO_{\PP^2/\Gamma}(1)$ on $\PP^2 /
\Gamma$. The latter bundle has a $\Gamma$-equivariant structure such that
the map $\cO_C \rightarrow \cO_{\PP^2/\Gamma}(1)$ is
$\Gamma$-equivariant. Let us examine some examples which are covered
by the analysis above.

\subsubsection*{Affine plane}

Let $C=\C^2$. Then $\,\overline{C}\,=\PP^2$ with 
$\ell_\infty=[0,z_1,z_2]\cong\PP^1$ and intersection number
$\ell_\infty\cdot\ell_\infty=1>0$. Since $H^2(\PP^2,\Z)\cong\Z$ in
this case, the constraint (\ref{dellinfty}) implies $d=0$. The
instanton moduli space $\calm_{N,k}(\C^2)$ in this case is a smooth
variety of complex dimension $2N\,k$. Furthermore,
$n=\chi(\PP^2)-2=3-2=1$. The partition function (\ref{VWNpartfndef})
computed with fixed $d=0$ and $n=1$ is thus given by
$$
Z_{U(N)}^{\C^2}(q)^{\rm fr} \=\widehat\eta(q)^{-N} \ ,
$$
which coincides with the instanton partition function on $\C^2$ described in Section~\ref{ADHMC2}

\subsubsection*{Local $\PP^1$}

Let $C=C_p$, $p>0$ be the total space of
the holomorphic line bundle $\Ocal_{\PP^1}(-p)$ of degree $-p$ over
$\PP^1$. Then $\,\overline{C}\,=\F_p:=\PP(\Ocal_{\PP^1}(-p)\oplus\Ocal_{\PP^1})$
is the $p$-th Hirzebruch surface. Let
$\ell_0=\PP(0\oplus\Ocal_{\PP^1})\cong\PP^1$ and
$\ell_\infty=\PP(\Ocal_{\PP^1}(-p)\oplus0)\cong\PP^1$. Then
$\ell_0\cdot\ell_0=-p<0$, $\ell_\infty\cdot\ell_\infty=p>0$, and
$\ell_0\cdot\ell_\infty=0$ so that the line $\ell_0$ does not pass
through the ``line at infinity''. From the constraint
(\ref{dellinfty}) and the second Betti number $b_2(C_p)=1$,
it follows that $d=m\,c_1(\Ocal_{C_p}(\ell_0))$ for
some $m\in\Z$. In this case, the instanton moduli space
$\calm_{N,m,k}(C_p)$ is a smooth variety of complex dimension
$2N\,k-p\,m^2$. The positivity of this dimension is a constraint on the integral
classes which ensures that the moduli space is non-empty. One has $\chi(\F_p)=1+2+1=4$, and hence
$n=\chi(\F_p)-2=2$. The partition function (\ref{VWNpartfinal}) in
this case becomes
\beq
Z_{U(N)}^{C_p}(q,Q)^{\rm fr} \=\frac1{\widehat\eta(q)^{2N}}~
\sum_{\vec u\in\Z^N}\,
q^{\frac p{2}\,\vec u\cdot\vec u}~Q^{u}
\label{ZXpqz}\eeq
with $u=u_1+\dots+u_N$.

Like the ALE spaces, one can work instead with a stack
compactification $\,\overline{C}\,^{\rm orb}$ of the total space of
the bundle $\Ocal_{\PP^1}(-p)$ obtained by adding a divisor
$\tilde\ell_\infty\cong\ell_\infty/\Gamma$. The resulting variety is a
toric Deligne--Mumford stack whose coarse space is the Hirzebruch
surface $\F_p$~\cite{BPT}. With this compactification one produces framed sheaves with
fractional first Chern classes $d=\frac
mp\,c_1(\Ocal_{C_p}(\ell_0))$, $m\in\Z$, and in this case the
partition function takes the form~\cite{BPT}
$$
Z_{U(N)}^{C_p}(q,Q)^{\rm orb} \= \left(\,\frac{\theta_3\big(\frac
  vp\,\big|\,\frac\tau p\big)}{\widehat\eta(q)^2}\,\right)^N
$$
anticipated by refs.~\cite{FMP,GSST}, where $q=\e^{2\pi\ii\tau}$,
$Q=\e^{2\pi\ii v}$ and
$\theta_3(v|\tau)=\sum_{n\in\Z}\,q^{\frac12\,n^2}~ Q^n$ is a Jacobi
elliptic function.

On these spaces one also has a version of the Hitchin--Kobayashi
correspondence which makes contact with the realization of instantons
as self-dual connections. Namely, $SU(N)$-instantons on $C_p$ are in
one-to-one correspondence with holomorphic bundles $\cale$ of rank $N$
on $C_p$ with $c_1(\cale)=0$ together with a framing at
infinity~\cite{GKM}. When $p=2$ these spaces coincide with the $A_1$
ALE space.

\subsubsection*{Compact surfaces}

We will now examine situations under which the formula (\ref{VWNpartfinal})
holds in the case of \emph{compact} surfaces, providing rigorous justification for some of the conjectural
formulas of ref.~\cite{FMP}. Let $C$ be a compact, smooth,
projective toric surface. Let $c_\infty$ be a generic point in
$C$ which is disjoint from the torus invariant lines $\ell_e$, $e\in E(C)$ 
of $C$. Let $\calm_{N,d,k}(C,c_\infty)$ be the moduli space of isomorphism classes $[\cale]$ of torsion
free coherent sheaves $\cale$ on $C$ with topological Chern invariants as
in (\ref{topinvtsE}), together with a ``framing'' at the point
$c_\infty$. When $\cale$ is locally free, this framing means a choice of basis for
the fibre space $\cale_{c_\infty}\cong\C^N$. When $\cale$ is not locally free,
the framing is defined with respect to a locally free resolution
$\cale^\bullet\to \cale\to0$.

Let $\sigma:(\,\widehat{C},\widehat{\ell}_\infty)\to (C,c_\infty)$ be
the blow-up of $C$ at $c_\infty$. Then $\widehat{C}$ is also a
compact, smooth, projective toric surface. We will suppose that the
exceptional divisor
$\widehat{\ell}_\infty=\sigma^{-1}(c_\infty)\cong\PP^1$ has positive
self-intersection
$\widehat{\ell}_\infty\cdot\widehat{\ell}_\infty>0$. Let $\calm^{\rm fr}_{N,d,k}(\,\widehat{C}\,)$ be the moduli space of isomorphism classes $\big[\,\widehat{\cale}\,\big]$ of torsion-free sheaves $\widehat{\cale}$ on $\widehat{C}$ which are framed on the line $\widehat{\ell}_\infty$ as above. The blow-up map $\sigma$ determines mutually inverse sheaf morphisms $\big[\,\widehat{\cale}\,\big]\mapsto [\cale]=\big[\sigma_*\widehat{\cale}\,\big]$ and $[\cale]\mapsto \big[\,\widehat{\cale}\,\big]=[\sigma^*\cale]$, and hence an explicit isomorphism
\beq
\calm_{N,d,k}(C,c_\infty) \ \cong \ \calm_{N,d,k}^{\rm fr}\big(\,\widehat{C}\,\big) \ .
\label{blowupiso}\eeq
It follows from above that the moduli space
$\calm_{N,d,k}(C,c_\infty)$ is thus a smooth variety of complex
dimension $2N\,k+d^2$, where $d^2:=\int_C\,d\wedge d$. Using the
isomorphism (\ref{blowupiso}), the instanton partition function on $C$
can be computed from the blow-up formula
\beq
Z_{U(N)}^C(q,\mbf Q)^{\rm fr}
\=Z_{U(N)}^{\widehat{C}\,\setminus\widehat{\ell}_\infty}(q,\mbf
Q)^{\rm fr}
\label{ZCblowup}\eeq
where the right-hand side is given by the formula (\ref{VWNpartfinal}).

As an explicit example, consider the complex projective plane
$C=\PP^2$. Let $z_\infty$ be a generic
point on $\PP^2$ disjoint from the line
$\ell_\infty=[0,z_1,z_2]$. Then the instanton moduli space
$\calm_{N,d,k}(\PP^2,z_\infty)$ is a smooth variety of complex
dimension $2N\,k+d^2$, where $d\in\Z$. Let
$\sigma:\widehat{\PP}{}^2\to\PP^2$ be the blowup of $\PP^2$ at
$z_\infty$ with exceptional divisor
$\widehat{\ell}_\infty=\sigma^{-1}(z_\infty)\cong\PP^1$. Then there
are
$n=\chi(\,\widehat{\PP}{}^2\setminus\widehat{\ell}_\infty)= \chi(\PP^2)=3$
maximal two-cones, and
$b_2(\,\widehat{\PP}{}^2\setminus\widehat{\ell}_\infty) =b_2(\PP^2)=1$
edge in the toric graph
$\Delta(\,\widehat{\PP}{}^2\setminus\widehat{\ell}_\infty) =
\Delta(\PP^2)$. Since $\ell_\infty\cdot\ell_\infty =1$ in this case, the partition function (\ref{ZCblowup}) gives the $U(N)$ formula
\beq
Z_{U(N)}^{\PP^2}(q,Q)^{\rm fr} \=\frac1{\widehat\eta(q)^{3N}}~\sum_{\vec u\in\Z^N}\,
q^{-\frac12\,\vec u\cdot\vec u}~Q^{u}
\label{ZP2qzUN}\eeq
with $u=u_1+\cdots+u_N$. For $N=1$ this agrees with the $U(1)$ gauge
theory partition function on $\PP^2$ derived in ref.~\cite[Sec.~5.1]{CK-PS}.

\subsection{Stability\label{Stability}}

Next we address Point~(b) at the end of
Section~\ref{D4wallcrossing} If one wishes to relax the requirement
of framing, then one must carefully analyse stability issues in order
to obtain well-defined instanton moduli spaces. Let $C$ be a smooth, quasi-projective toric surface. Again we fix
invariants $N\in\N$, $k\in\Q$ and $d\in H^2_{\rm cpt}(C,\Z)$. Let
$\calm_{N,d,k}(C)$ be the unframed moduli space consisting of
isomorphism classes $[\cale]$ of semistable torsion-free coherent sheaves
$\cale$ on $C$ such that $\cale$ has topological Chern invariants
\beq
N\=\ch_0(\cale)\={\rm rank}(\cale) \ , \qquad d\=\ch_1(\cale)\=c_1(\cale)  \qquad
\mbox{and} \qquad k\=-\int_{C}\,\ch_2(\cale) \ .
\label{topinvtsEgen}\eeq

The notion of ``semistability'' considers sheaves with fixed Hilbert
polynomial. In this case, a general result of Maruyama~\cite{Maruyama}
constructs the algebraic scheme $\calm_{N,d,k}(C)$ and shows that it
is \emph{projective}. In particular, it admits a
$\tilde T$-equivariant embedding into a smooth variety. In the case that $C$ is a
polarized surface, i.e. it admits a smooth distribution in its tangent
bundle $TC$ which is integrable and Lagrangian (in an appropriate sense), we can formulate the notion of semistability in
terms of the more familiar slope semistability
following ref.~\cite{Gieseker}. Let $\call$ be a fixed ample line bundle on
$C$. For a coherent torsion-free sheaf $\cale$ on $C$, define the
polynomial
$$
\rho_\cale(n) \=\frac{\chi\big(\cale\otimes \call^{\otimes n}\big)}{{\rm rank}(\cale)}
$$
for $n\in\N_0$. Then $\cale$ is said to be $\call$-semistable if
$$
\rho_\calf(n) \ \leq \ \rho_\cale(n)
$$
for all $n\gg0$ whenever $\calf$ is a coherent subsheaf of $\cale$. In this
case, Gieseker~\cite{Gieseker} constructs $\calm_{N,d,k}(C)$ as a
projective $Quot$ scheme. The slope of $\cale$ is the rational number
$$
\mu(\cale) \=\frac{{\rm deg}(\cale)}{{\rm rank}(\cale)} \ ,
$$
where the degree of $\cale$ is defined using the polarization as
\beq
{\rm deg}(\cale) \=\int_C\,c_1(\cale)\wedge c_1(\call) \ .
\label{degreeE}\eeq
An application of the Riemann--Roch theorem shows~\cite{Gieseker}
$$
\rho_\cale(n) \=\frac n2\,\int_C\,\Big(c_1(\call)\wedge\big(n\,c_1(\call)+c_1(C)
\big)+\frac{2\,\ch_2(\cale)+c_1(\cale)\wedge c_1(C)}{{\rm rank}(\cale)}\Big)+
\mu(\cale)+\chi(\Ocal_C) \ .
$$
It follows that $\rho_\cale(n)\geq\rho_\calf(n)$ for $n\gg0$ if and only if
$\mu(\cale)\geq\mu(\calf)$, and hence $\call$-semistability is equivalent to the
usual quasi-BPS instanton equations in this case. In physics
applications one is interested in instances where $C$ is a K\"ahler
surface. In this case one can use the K\"ahler polarization and take
$c_1(\call)$ to be the K\"ahler two-form $k_0$ in (\ref{degreeE}).

The following result computes the expected dimension of the instanton
moduli space in this case.
\begin{lemma}
The virtual dimension of $\calm_{N,d,k}(C)$ equals
$2N\,k+d^2-(N^2-1)\,\chi(\Ocal_C)$, where $d^2:=\int_C\,d\wedge d$ and
$\chi(\Ocal_C)$ is the holomorphic Euler characteristic of $C$. 
\end{lemma}
\begin{proof}
As discussed in Section~\ref{Eqint}, the dimension of the virtual tangent space $T_{[\cale]}^{\rm vir}\calm_{N,d,k}(C)=\Ext_{\Ocal_C}^1(\cale,\cale)\ominus\Ext_{\Ocal_C}^2(\cale,\cale)$ to the instanton moduli space at a point $[\cale]\in\calm_{N,d,k}(C)$ in obstruction theory is the difference of Euler characteristics
$\chi(\Ocal_C\otimes\Ocal_C^\vee)-\chi(\cale\otimes \cale^\vee\,)$. The latter
quantity may be computed for $\cale$ locally free by using the Hirzebruch--Riemann--Roch
theorem to write 
\beq
\chi\big(\cale\otimes \cale^\vee\,\big) \=
\int_C\,\ch\big(\cale\otimes \cale^\vee\,\big)\wedge{\rm Td}(C) \ .
\label{GRRthm}\eeq
Let $\nu$ be the generator of $H^4(C,\Z)\cong\Z$ which is Poincar\'e
dual to a point in $C$. Then using (\ref{topinvtsEgen}) one computes
the Chern character 
$$
\ch\big(\cale\otimes \cale^\vee\,\big)\=\ch(\cale)\wedge\ch\big(\cale^\vee\,\big)\=
(N+d-k\,\nu)\wedge(N-d-k\,\nu)\=
N^2-d\wedge d-2N\,k\,\nu \ .
$$
The Todd characteristic class of the tangent bundle of $C$ is given in
terms of Chern classes of $TC$ as 
\beq
{\rm Td}(C) \=1+\mbox{$\frac12$}\,c_1(C)+\mbox{$\frac1{12}$}\,
\big(c_1(C)\wedge c_1(C)+c_2(C)\big) \ .
\label{ToddC}\eeq
We may thus write (\ref{GRRthm}) as
\bea
\chi\big(\cale\otimes \cale^\vee\,\big)&=&\int_C\,\Big(
\frac{N^2}{12}\,\big(c_1(C)\wedge c_1(C)+c_2(C)\big)-d\wedge
d-2N\,k\,\nu 
\Big)\nonumber\\[4pt] &=&-2N\,k-d^2+N^2\,\int_C\,{\rm Td}(C) \ .
\nonumber\eea
From the Hirzebruch--Riemann--Roch formula and $\ch(\Ocal_C)=1$, one
has 
\beq
\chi(\Ocal_C)\=\int_C\,{\rm Td}(C)\=\chi\big(\Ocal_C\otimes
\Ocal_C^\vee\big)
\label{chiOcalC}\eeq
and the result follows when $\cale$ is a bundle. When $\cale$ is not locally
free, we use a locally free resolution $\cale^\bullet\to\cale \to0$ along with
additivity of the Chern character.
\end{proof}

This result shows that the expected dimension of the instanton moduli
space for $N=1$ coincides with the dimension of the framed moduli
space $\calm_{1,d,k}(C)$ of Section~\ref{modspframed} Indeed, torsion
free sheaves of rank one are always stable in the sense explained above. Since any torsion free sheaf decomposes as the product $\cale=\cali\otimes\call$ of an ideal sheaf of a zero-dimensional subscheme and a line bundle, the moduli space factorizes into a product
$$
\calm_{1,d,k}(C)\= \calm_{1,0,k+d^2/2}(C)\times\calm_{1,d,-d^2/2}(C) \ ,
$$
where $\calm_{1,0,K}(C)\cong C^{[K]}$ is the Hilbert scheme of $K$
points on $C$ which is a smooth variety of dimension $2K$, and the
(zero-dimensional) Picard group $\calm_{1,d,-d^2/2}(C)\cong {\rm
  Pic}^d(C)$ parametrizes fractional instantons. The factorization of
the $U(1)$ gauge theory partition function into contributions from
regular and fractional instantons then follows from the
multiplicativity of the Euler class under tensor product of (tangent)
bundles. For compact toric surfaces $C$, it is given by the formula
(\ref{VWNpartfinal}) with $N=1$. In the case of Hirzebruch--Jung surfaces, the
contributions of fractional charges in this case is shown in
ref.~\cite{CK-PS} to arise from the non-compact prime divisors of $C$
via linear equivalences. It amounts to identifying the generators
$c_1(\cR_i)$ of the K\"ahler cone in $CH^1(C)\otimes\Q$ with the
duals to the exceptional divisors $D_i$ which generate the Mori cone
in $CH_1^{\rm cpt}(C)\otimes\Q$. For the ALE spaces, the moduli space
$\calm_{1,0,k+d^2/2}(C)$ coincides with the quiver variety of the ADHM
construction (with suitable stability conditions)~\cite{Kuznetsov}.

The situation for higher rank $N>1$ is much more complicated. In this
case, slope-stability does not seem to properly account for walls of
marginal stability extending to infinity which describe wall-crossing
behaviour of the partition functions counting D4--D2--D0 brane bound
states on Calabi--Yau manifolds $X$ with $h^{1,1}(X)>1$~\cite{DM}. As
discussed in Section~\ref{Wallcrossing}, the physical theory is
described by moduli spaces of stable objects in the derived category ${\sf
  D}^b(\coh(X))$, as the observed D-brane decays are impossible in the
abelian category $\coh(X)$ of coherent sheaves on $X$.

\subsection{Perpendicular partition functions and universal sheaves}

Finally, we come to the last Point~(c) at the end of
Section~\ref{D4wallcrossing} For $U(1)$ gauge theory this relationship
is analysed in detail in ref.~\cite{CK-PS}. In this case a
four-dimensional analog of the topological vertex formalism can be
developed. On the gauge theory side, the vertex contributions should
be computed by a version of the perpendicular partition function of
Section~\ref{DTtheory}, i.e. the generating function
$P_{U(N)}^{m_1,m_2}(q)$ for instantons on $C=\C^2$ with fixed
asymptotics such that $Z_{U(N)}^{\C^2}(q)=P_{U(N)}^{0,0}(q)$, while
the edge contributions should be read off from the character of the
corresponding universal sheaf. In the
remainder of this section we will describe in detail the instanton
moduli space with boundary conditions specified by integers $m_1,m_2$,
and show that the fixed point loci of the induced $\tilde T$-action
are enumerated by two-dimensional Young diagrams with asymptotics
$m_1,m_2$. This gives the four-dimensional version of the gauge theory 
gluing rules of Section~\ref{DTtheory} and also a first principles
derivation of the empirical vertex rules of ref.~\cite[Sec.~4.3]{CK-PS}.

\subsubsection*{Instanton moduli spaces on $\F_0$}

We will first describe how to generate instantons with non-trivial
first Chern class using our previous formalism. For this, rather than
working with the projective plane $\overline{C}=\PP^2$, it is more convenient
to work with the ``two-point compactification'' of $C=\C^2$ given by a
product of two projective lines $\F_0=\PP_z^1\times\PP^1_w$, where the
labels will be used to keep track of each of the two factors. This
variety can be identified as the zeroth Hirzebruch surface
$\F_0=\PP(\Ocal_{\PP^1}\oplus\Ocal_{\PP^1})$, i.e. the projective
compactification of the trivial line bundle
$\Ocal_{\PP^1}=\C\times\PP^1$, which is again a toric surface. For
$(t_z,t_w)\in T^2$, the toric action is described by the automorphism
$F_{t_z,t_w}$ of $\F_0$ defined by
\beq
F_{t_z,t_w}\big([z_0,z_1]\,;\,[w_0,w_1]\big) \=
\big([z_0,t_z\,z_1]\,;\,[w_0,t_w\,w_1]\big) \ .
\label{Ft1t2def}\eeq
Let $p_z:\F_0\to\PP^1_z$ and $p_w:\F_0\to\PP^1_w$ be the canonical
projections, and $i_z:\PP^1_z\cong\PP_z^1\times{\rm
  pt}\hookrightarrow\F_0$ and $i_w:\PP^1_w\cong{\rm
  pt}\times\PP^1_w\hookrightarrow\F_0$ the natural inclusions. There
are now two independent distinguished lines at infinity fixed 
by this $ T^2$ action, given by the pushforwards
$\ell_z=(i_z)_*\PP_z^1$ and $\ell_w=(i_w)_*\PP_w^1$, for which the
respective factors of the action of $ T^2=\C^*\times\C^*$ reduce to
the standard $\C^*$-action on $\ell_z,\ell_w\cong\PP^1$. These two
divisors have intersection products
$\ell_z\cdot\ell_z=\ell_w\cdot\ell_w=0$ and $\ell_z\cdot\ell_w=1$, and
the canonical divisor is $K_{\F_0}=-2\ell_z-2\ell_w$. They generate
the Picard group of $\F_0$, and hence induce a bigrading on line
bundles over $\F_0$. For $m_z,m_w\in\Z$, we write
\beq
\Ocal_{\F_0}(m_z,m_w)~:=~\Ocal_{\F_0}(m_z\,\ell_z+m_w\,\ell_w)\=
p_z^*\Ocal_{\PP^1_z}(m_z)\otimes p_w^*\Ocal_{\PP^1_w}(m_w) \ .
\label{bigrading}\eeq
By K\"unneth's theorem, the non-trivial cohomology groups of $\F_0$
are given by
$$
H^0(\F_0,\Z)\=\Z \ , \qquad
H^2(\F_0,\Z)\=\Z[\xi_z]\oplus\Z[\xi_w] \qquad \mbox{and} \qquad
H^4(\F_0,\Z)\=\Z[\xi_z\wedge\xi_w] \ ,
$$
where $\xi_z=c_1(\Ocal_{\F_0}(1,0))=p_z^*c_1(\Ocal_{\PP^1_z}(1))$ and
$\xi_w=c_1(\Ocal_{\F_0}(0,1))=p_w^*c_1(\Ocal_{\PP^1_w}(1))$.

Let $\calm_{N;(m_z,m_w);k}(\F_0)$ be the moduli space of isomorphism
classes $[\cale]$ of torsion-free sheaves $\cale$ on $\F_0$ with
topological Chern invariants as in (\ref{topinvtsE}), where
$d=m_z\,\xi_z+m_w\,\xi_w$, and non-trivial framings along the two
independent ``directions at infinity'' prescribed by two fixed
isomorphisms $\cale\big|_{\ell_z}\cong W\otimes\Ocal_{\PP^1_z}(m_z)$ and
$\cale\big|_{\ell_w}\cong W\otimes\Ocal_{\PP^1_w}(m_w)$, where $W$ is a
fixed $N$-dimensional complex vector space. We will say that 
such a sheaf is ``trivialized at infinity'' if it is equipped with two
isomorphisms $\cale\big|_{\ell_z}\cong W\otimes\Ocal_{\PP^1_z}$ and
$\cale\big|_{\ell_w}\cong W\otimes\Ocal_{\PP^1_w}$. The moduli space of
trivialized sheaves is denoted
$\calm_{N,k}(\F_0):=\calm_{N;(0,0);k}(\F_0)$. These moduli spaces
carry an obvious induced action of the torus
$\tilde T= T^2\times T^N$, and the following formal
arguments show that the instanton counting in these moduli spaces is
the same as before.

\begin{proposition}
There is a natural $\tilde T$-equivariant birational equivalence
between the moduli spaces
$$
\calm_{N,k}(\F_0)~\cong~\calm_{N,k}\big(\C^2\big) \ .
$$
\label{MNkF0C2prop}\end{proposition}
\begin{proof}
Represent $\F_0$ as a divisor in $\PP^2\times\PP^1$ via the embedding
$\PP_z^1\hookrightarrow\PP^2$, $[z_0,z_1]\mapsto
[z_0,z_1,z_1]$. Let $\widehat{\widehat\PP}{}^2$ be the surface
obtained as the blowup $p$ of a pair of points on the line at infinity
$\ell_\infty\subset\PP^2$ (see Section~\ref{modspframed}). Note that
$\ell_\infty$ is disjoint from the image of the line $\ell_z$. Then
there is a correspondence diagram
$$
\xymatrix{
 & \widehat{\widehat\PP}{}^2 \ar[ld]_p \ar[rd]^q & \\
\PP^2 & & \F_0
}
$$
where $q$ is the blowup of the intersection point
$\ell_z\cdot\ell_w$ on $\F_0=\PP^1_z\times\PP^1_w$. The corresponding
Fourier--Mukai functors $q_*\,p^*$ and $p_*\,q^*$ determine
equivalences on the categories of coherent sheaves over $\PP^2$ and
$\F_0$, and hence induce mutually inverse rational maps between the
corresponding sets of isomorphism classes of trivially framed torsion
free sheaves.
\end{proof}

\begin{proposition}
For each fixed $(m_z,m_w)\in\Z^2$, there is a natural
$\tilde T$-equivariant bijection between the moduli spaces 
$$
\calm_{N,k}(\F_0)~\cong~\calm_{N\,;\,
(N\,m_z,N\,m_w)\,;\,k(m_z,m_w)}(\F_0)
$$
where
\beq
k(m_z,m_w)\=k-N\,m_z\,m_w \ .
\label{km1m2def}\eeq
\label{Youngshiftprop}\end{proposition}
\begin{proof}
Given $[\cale]\in\calm_{N,k}(\F_0)$, define
\beq
\cale(m_z,m_w) \=\cale\otimes \Ocal_{\F_0}(m_z,m_w) \ .
\label{Em1m2def}\eeq
By (\ref{topinvtsE}) (with $d=0$), one has
\beq
\ch(\cale) \=N-k\,\xi_z\wedge\xi_w \ ,
\label{chExi}\eeq
and using multiplicativity of the Chern character we compute
\bea
\ch\big(\cale(m_z,m_w)\big) &=& (N-k\,\xi_z\wedge\xi_w)\wedge\big(1+m_z\,
\xi_z\big)\wedge\big(1+m_w\,\xi_w\big) \nonumber\\[4pt]
&=& N+N\,\big(m_z\,\xi_z+m_w\,\xi_w\big)+
(N\,m_z\,m_w-k)\,\xi_z\wedge\xi_w \ .
\label{chEmcompute}\eea
This therefore gives a map $[\cale]\mapsto[\cale(m_z,m_w)]$ on
$\calm_{N,k}(\F_0)\to \calm_{N;(N\,m_z,N\,m_w);k(m_z,m_w)}(\F_0)$. An
identical calculation shows that the map $\cale(m_z,m_w)\mapsto
\cale(m_z,m_w)\otimes\Ocal_{\F_0}(-m_z,-m_w)$ is its inverse. These
maps clearly induce bijections between $(m_z,m_w)$-framings and
trivializations at infinity.
\end{proof}

Proposition~\ref{MNkF0C2prop} shows that an instanton gauge bundle $\cale$
with trivial asymptotics on $\F_0$ can again be regarded as an element
$[\cale]\in\calm_{N,k}(\C^2)$. For fixed $m_z,m_w\in\Z$,
Proposition~\ref{Youngshiftprop} shows that the counting of
torsion-free sheaves in the moduli spaces $\calm_{N,k}(\F_0)$ and
$\calm_{N;(N\,m_z,N\,m_w);k(m_z,m_w)}(\F_0)$ coincide. In other words,
the number of instantons with fixed non-trivial asymptotics matches
the Young tableau count, as the number of sheaves (\ref{Em1m2def}) is
independent of the integers
$m_z,m_w$. Proposition~\ref{Youngshiftprop} also reproduces the
formula~\cite[Sec.~4.2.2]{CK-PS} for the renormalized volume of Young tableaux (reproduced
above for $N=1$) as the induced shift in instanton charge in
(\ref{km1m2def}) of the sheaves~(\ref{Em1m2def}).

The bijection of Proposition~\ref{MNkF0C2prop} is not necessarily a
diffeomorphism. In fact, we have not yet shown that
$\calm_{N,k}(\F_0)$ is a smooth variety. Naively, one may try to argue
this by using the formalism of ref.~\cite{GL}, as our sheaves
are trivialized on the divisor $\ell_\infty=\ell_z+\ell_w$ which has
positive self-intersection $\ell_\infty\cdot\ell_\infty=2$, and
$\F_0\setminus\ell_\infty\cong\C^2$. However, this divisor is not
irreducible, in particular it is not a $\PP^1$. Note that formally
applying the arguments which led to (\ref{ZXpqz}) in the limit $p=0$
would produce the partition function
\beq
Z_{U(N)}^{\F_0}(q) \=\big(Z_{U(N)}^{\C^2}(q)\big)^2 \ ,
\label{ZF0q}\eeq
as the Hirzebruch surface has $n=\chi(\F_0)-2=2$. Although the generating
function (\ref{ZF0q}) is ``doubled'', the instanton counting problems
on $\F_0$ and on $\C^2$ are nevertheless identical. A rigorous
derivation of the formula (\ref{ZF0q}) using the techniques of
Section~\ref{modspframed} follows once we establish that the equivalence of
Proposition~\ref{MNkF0C2prop} is an isomorphism of underlying smooth
varieties. Recall from Section~\ref{modspframed} that the complex
dimension of the instanton moduli space $\calm_{N,k}(\C^2)$ is
$2N\,k$.

\begin{proposition}
The moduli space $\calm_{N,k}(\F_0)$ is a smooth quasi-projective
variety of complex dimension $2N\,k$.
\label{MNkF0smoothprop}\end{proposition}
\begin{proof}
As usual, the trivialization condition at infinity guarantees Gieseker
semistability and hence quasi-projectivity, as discussed in Section~\ref{Stability} Using the
divisor $\ell_\infty=\ell_z+\ell_w$ of $\F_0$, one constructs the
framed moduli space $\calm_{N,k}(\F_0)$ as in ref.~\cite{HL} with tangent
spaces $\Ext_{\Ocal_{\F_0}}^1(\cale,\cale(-\ell_\infty))$ and obstruction spaces given by
$\Ext_{\Ocal_{\F_0}}^2(\cale,\cale(-\ell_\infty))$. By ref.~\cite[Sec.~5]{BGK}, one has
$$
\Ext_{\Ocal_{\F_0}}^0\big(\cale\,,\,\cale(-1,-1)\big)\= \Ext_{\Ocal_{\F_0}}^2\big(\cale\,,\,\cale(-1,-1)\big)\=0
$$
and hence $\calm_{N,k}(\F_0)$ is smooth. The dimension of the tangent
spaces is thus equal to minus the Euler characteristic
$-\chi(\cale,\cale(-1,-1))$, which for $\cale$ locally free may be calculated by
using the Hirzebruch--Riemann--Roch theorem to write
\beq
\chi\big(\cale\,,\,\cale(-1,-1)\big) \=
\int_{\F_0}\,\ch\big(\cale^\vee\otimes \cale(-1,-1)\big)\wedge{\rm Td}(\F_0)
\ .
\label{chiEE1}\eeq
Using (\ref{chExi}) and (\ref{chEmcompute}) one computes the Chern
character
\bea
\ch\big(\cale^\vee\otimes \cale(-1,-1)\big)&=&\ch\big(\cale^\vee\,\big)\wedge
\ch\big(\cale(-1,-1)\big) \nonumber\\[4pt] &=&
\big(N-k\,\xi_z\wedge\xi_w\big)\wedge\big(N-N\,(\xi_z+\xi_w)+(N-k)\,
\xi_z\wedge\xi_w\big) \nonumber\\[4pt] &=&
N^2-N^2\,(\xi_z+\xi_w)+\big(N^2-2N\,k\big)\,\xi_z\wedge\xi_w \ . 
\nonumber\eea
The Todd characteristic class and holomorphic Euler characteristic are
given by (\ref{ToddC}) and (\ref{chiOcalC}) with $C=\F_0$, and we may
thus write (\ref{chiEE1}) as
\bea
\chi\big(\cale\,,\,\cale(-1,-1)\big)&=& \int_{\F_0}\,\Big(\,\frac{N^2}{12}\,
\big(c_1(\F_0)\wedge c_1(\F_0)+c_2(\F_0)\big)-\frac{N^2}2\,
(\xi_z+\xi_w)\wedge c_1(\F_0)\nonumber \\ && \qquad\qquad +\,
\big(N^2-2N\,k\big)\,\xi_z\wedge\xi_w\,\Big) \nonumber\\[4pt] &=&
N^2\,\chi(\Ocal_{\F_0})-\frac{N^2}2\,\int_{\F_0}\,c_1(\F_0)\wedge
(\xi_z+\xi_w)+N^2-2N\,k \ .
\label{chiEE1calc}\eea
One has $\chi(\Ocal_{\F_0})=1$ since $b_2(\F_0)=2=h^{1,1}(\F_0)$. Since
$\xi_z,\xi_w\in H^2(\F_0,\Z)$ are the Poincar\'e duals of
$[\ell_z],[\ell_w]\in H_2(\F_0,\Z)$, with $\ell_z,\ell_w\cong\PP^1$
and $\ell_z\cdot\ell_z=\ell_w\cdot\ell_w=0$, the remaining integral in
(\ref{chiEE1calc}) may be computed as
\beq
\int_{\F_0}\,c_1(\F_0)\wedge(\xi_z+\xi_w)\=2\,\int_{\PP^1}\,
c_1\big(\PP^1\big)\=4 \ ,
\label{c1P1int}\eeq
where in the last step we used the fact that the holomorphic
tangent bundle of $\PP^1$ can be identified with
$\Ocal_{\PP^1}(2)$. Putting everything together we find
\beq
\chi\big(\cale\,,\,\cale(-1,-1)\big) \=-2N\,k
\nonumber\eeq
and the result follows when $\cale$ is a bundle. If $\cale$ is not locally
free, then we simply consider a locally free resolution $\cale^\bullet\to
\cale\to0$ and use additivity of the Chern characteristic class.
\end{proof}

In a similar vein, the equivalence established in
Proposition~\ref{Youngshiftprop} as it stands is only a bijective
correspondence. The moduli space for non-trivial
asymptotics is an example of the moduli spaces of ``decorated
sheaves'' described in ref.~\cite[Sec.~4.B]{HLbook}, with the
$(m_z,m_w)$-framed sheaves on $\F_0$ providing examples of semistable
framed modules. The moduli space
$\calm_{N;(N\,m_z,N\,m_w);k(m_z,m_w)}(\F_0)$ is thus a projective
scheme with a universal family, i.e. it is fine, and so possesses a
universal sheaf. We will now argue that the moduli spaces
$\calm_{N;(N\,m_z,N\,m_w);k(m_z,m_w)}(\F_0)$ are smooth and
diffeomorphic to one another for all $m_z,m_w\in\Z$. Then the
computations of Section~\ref{modspframed} can be repeated by
integrating over $\calm_{N;(N\,m_z,N\,m_w);k(m_z,m_w)}(\F_0)$
using (\ref{km1m2def}) and G\"ottsche's formula to get the desired perpendicular partition
function
$$
P_{U(N)}^{m_z,m_w}(q)\= q^{-N\,m_z\,m_w}~\widehat\eta(q)^{-2N} \ .
$$
For $N=1$, this yields a rigorous derivation of the vertex rules of
ref.~\cite{CK-PS}.

Let $\tilde\calm(\F_0)$ be the moduli space of isomorphism classes of pairs $\tilde
\cale=(\cale,\alpha)$, where $\cale$ is an $(m_z,m_w)$-framed torsion free sheaf
on $\F_0$ and $\alpha:\Ocal_{\F_0}(m_z,m_w)^{\oplus N}\to\cale$ is the
  surjective morphism induced by the framing. Consider the family of
  pairs $\tilde \calf=(\calf,\beta)$ consisting of a coherent sheaf $\calf$ on
  $\F_0$ together with a homomorphism
  $\beta:\Ocal_{\F_0}(m_z,m_w)^{\oplus r}\to\calf$ for some
  $r\in\N$. These pairs define objects of an abelian category
  $\coh_f(\F_0)$ with the obvious morphisms between pairs induced by
  morphisms of coherent sheaves. Any $(m_z,m_w)$-framed torsion free
  sheaf $\cale$ on $\F_0$ clearly defines an object of $\coh_f(\F_0)$,
  with $r=N$ and $\ch(\calf)=\ch(\cale)$.

Given the abelian category $\coh_f(\F_0)$, one can derive functors,
and hence compute sheaf cohomology in $\coh_f(\F_0)$. We denote the
corresponding $\Ext$-functors by $\Ext^i_f$. Then $\Ext_f^2(\tilde
\cale,\tilde \cale)$ is the obstruction space for $\tilde\calm(\F_0)$ at
$\tilde \cale$, while $\Ext_f^1(\tilde \cale,\tilde \cale)$ is the tangent space
to $\tilde\calm(\F_0)$ at $\tilde \cale$. Generally, for any two objects
$\tilde \calf,\tilde \cG$ of $\coh_f(\F_0)$, the group $\Ext_f^1(\tilde
\calf,\tilde \cG)$ can be computed in a purely algebraic way in terms of
equivalence classes of extensions
\beq
\xymatrix{
\Ocal_{\F_0}(m_z,m_w)^{\oplus s}~\ar[d] \ar[r]&~
\Ocal_{\F_0}(m_z,m_w)^{\oplus(s+r)}~\ar[d] \ar[r]&~
\Ocal_{\F_0}(m_z,m_w)^{\oplus r}\ar[d] \\
\tilde \cG~\ar[r]&~\tilde \cK~\ar[r]&~\tilde \calf \ .
}
\label{Extf1comp}\eeq
This definition of $\Ext_f^1(\tilde \calf,\tilde \cG)$ is equivalent to the
usual definition in terms of a twisted Dolbeault complex, and in
particular it gives the deformation theory of the moduli space
$\tilde\calm(\F_0)$. Furthermore, there is a spectral sequence
connecting it to the ordinary $\Ext$-groups $\Ext_{\Ocal_{\F_0}}^i(\calf,\cG)$.

The line bundle $\Ocal_{\F_0}(m_z,m_w)$ is a projective
$\Ocal_{\F_0}$-module and is therefore flat, whence the
tensor product map defined by (\ref{Em1m2def}) preserves exact
sequences such as (\ref{Extf1comp}). It follows that the tangent
spaces $\Ext_f^1(\tilde \cale,\tilde \cale)$ are the same for all
$m_z,m_w\in\Z$. Thus by Proposition~\ref{MNkF0smoothprop}, the
moduli spaces $\tilde\calm(\F_0)$ are all smooth and diffeomorphic to
one another. This is consistent with the fact that stability (either
slope stability of torsion free sheaves~\cite{Gieseker} or stability of
framed modules~\cite[Sec.~4.B]{HLbook}) is preserved by twisting with the
line bundles~$\Ocal_{\F_0}(m_z,m_w)$.

\subsubsection*{Universal sheaves}

We now construct the universal sheaf for the standard framing on $\F_0$, and compute its $\tilde T$-equivariant
character. Since the moduli space $\calm_{N,k}(\F_0)$ is
fine~\cite{HL}, there exists a universal sheaf
$\Ecal\to\F_0\times\calm_{N,k}(\F_0)$, i.e. a torsion-free sheaf 
$\Ecal$ such that $\Ecal\big|_{\F_0\times[\cale]}\cong \cale$ for every point
$[\cale]\in\calm_{N,k}(\F_0)$, unique up to tensor product with a (unique)
line bundle. There are two natural vector bundles over
$\calm_{N,k}(\F_0)$ associated to a universal sheaf $\Ecal$. Firstly,
there is the framing bundle $W=H^0(\F_0,\Ecal|_{\ell_\infty})$ of rank
$N$ given by the fibre at infinity. Secondly, there is the bundle $V$
of ``Dirac zero modes'', whose fibre at a given locally free sheaf
$[\cale]\in\calm_{N,k}(\F_0)$ restricted to
$\R^4\cong\C^2\cong\F_0\setminus\ell_\infty$ is the space of
$L^2$-solutions to the Dirac equation on $\R^4$ in the
background of the fundamental representation of the instanton gauge
field corresponding to $[\cale]\big|_{\C^2}$. It
is constructed explicitly as follows. Let $\pi_1$ and $\pi_2$ be the
canonical projections of $\F_0\times\calm_{N,k}(\F_0)$ onto the first
and second factors, respectively. Then the Dirac bundle is
$$
V \= R^1\pi_{2*}\big(\Ecal\otimes\pi_1^*\Ocal_{\F_0}(-1,-1)\big) \ ,
$$
where $R^1\pi_{2*}$ is the first right derived functor of the
pushforward functor $\pi_{2*}$. Its fibre over a point
$[\cale]\in\calm_{N,k}(\F_0)$ is the vector space $H^1(\F_0,\cale(-1,-1))$.

\begin{proposition}
The Dirac bundle $V$ is a vector bundle of rank $k$ over
$\calm_{N,k}(\F_0)$.
\label{Vtrivrank}\end{proposition}
\begin{proof}
By ref.~\cite[Sec.~5]{BGK}, one has
$$
H^0\big(\F_0\,,\,\cale(-1,-1)\big)\=H^2\big(\F_0\,,\,\cale(-1,-1)\big)\=0
$$
and hence $V$ is a vector bundle of rank equal to
$-\chi(\cale(-1,-1))$. For $\cale$ locally free, the Hirzebruch--Riemann--Roch
theorem, together with (\ref{chEmcompute}), (\ref{ToddC}) with
$C=\F_0$, (\ref{c1P1int}) and $\chi(\Ocal_{\F_0})=1$ give
\bea
{\rm rank}(V)&=&-\int_{\F_0}\,\ch\big(\cale(-1,-1)\big)\wedge\Td(\F_0)
\nonumber\\[4pt] &=&-\int_{\F_0}\,\big(N-N\,(\xi_z+\xi_w)+(N-k)\,
\xi_z\wedge\xi_w\big) \nonumber\\ && \qquad\qquad 
\wedge~\Big(1+\mbox{$\frac12$}\,c_1(\F_0)+\mbox{$\frac1{12}$}\,
\big(c_1(\F_0)\wedge c_1(\F_0)+c_2(\F_0)\big)\Big) \nonumber\\[4pt]
&=&-N\,\chi(\Ocal_{\F_0})+\frac N2\,\int_{\F_0}\,c_1(\F_0)\wedge
(\xi_z+\xi_w)+k-N \= k \ .
\nonumber\eea
Again, the statement for generic torsion free sheaves
$[\cale]\in\calm_{N,k}(\F_0)$ follows by considering a locally free
resolution $\cale^\bullet\to \cale\to0$. 
\end{proof}

By definition, (\ref{Ft1t2def}) and the construction of
Section~\ref{modspframed}, the $\tilde T$-equivariant characters of the framing
and Dirac bundles regarded as $\tilde T$-modules, at a fixed point in
$\calm_{N,k}(\F_0)^{\tilde T}$ labelled by an $N$-coloured Young tableau
$\vec \lambda=(\lambda_1,\dots,\lambda_N)$, are given by
\beq
\ch_{\tilde T}(W)(\vec \lambda)\=\sum_{l=1}^N\,e_l \qquad \mbox{and} \qquad
\ch_{\tilde T}(V)(\vec \lambda)\=\sum_{l=1}^N\,e_l~\sum_{(i,j)\in \lambda_l}\,
t_z^{i-1}\,t_w^{j-1}
\label{chTWVY}\eeq
where $e_l=\e^{a_l}$.
Let us now describe how the ADHM construction is modified in this
case. This has been worked out in ref.~\cite[Sec.~5]{BGK}. The linear
algebraic data is formally the same as for the analysis of framed instantons
on $\PP^2$~\cite{Nakajima}, defined by linear operators
$$
B_z~\in~\Hom(V,V) \ , \qquad
B_w~\in~\Hom(V,V) \ , \qquad
I~\in~\Hom(W,V) \qquad \mbox{and} \qquad
J~\in~\Hom(V,W) \ .
$$
Then any framed torsion free sheaf $[\cale]\in\calm_{N,k}(\F_0)$ can be
represented as the middle cohomology group of the complex
\beq
\xymatrix{
0~\ar[r] & ~V\otimes\Ocal_{\F_0}(-1,-1)~\ar[r]^{~\sigma} & ~
{\begin{matrix}
V\otimes\Ocal_{\F_0}(-1,0) \\ \oplus \\
V\otimes\Ocal_{\F_0}(0,-1) \\ \oplus \\
W\otimes\Ocal_{\F_0}
\end{matrix}}
~\ar[r]^{~~~~\tau} & ~ V
\otimes\Ocal_{\F_0}~\ar[r] & ~0
}
\label{F0ADHMcomplex}\eeq
which is exact at the first and last terms, where
$$
\sigma \=\begin{pmatrix}
z_0\,B_z-z_1 \\ w_0\,B_w-w_1 \\ z_0\,w_0\,J
\end{pmatrix}
\qquad \mbox{and} \qquad \tau \=\begin{pmatrix}
-(w_0\,B_w-w_1) & z_0\,B_z-z_1 & I
\end{pmatrix} \ .
$$

The virtual $\tilde T$-equivariant bundle defined by the cohomology
of the complex (\ref{F0ADHMcomplex}) gives a
representative of the isomorphism class of the universal sheaf $\Ecal$
in the $\tilde T$-equivariant K-theory group
$K_{\tilde T}^0(\F_0\times\calm_{N,k}(\F_0))$ as
\beq
\Ecal \=\big(\Ocal_{\F_0}\boxtimes W\big)~\oplus~
\big(S^-\ominus S^+\big)\boxtimes V \ ,
\label{Ecalvirtual}\eeq
where
$$
S^+\=\Ocal_{\F_0}(-1,-1)\oplus\Ocal_{\F_0} \qquad \mbox{and} \qquad
S^-\=\Ocal_{\F_0}(-1,0)\oplus\Ocal_{\F_0}(0,-1)
$$
are $\tilde T$-equivariant bundles over $\F_0$ which, after tensoring
with a line bundle of degree one, restrict to the
usual positive/negative chirality spinor bundles over
$\R^4\cong\C^2\cong\F_0\setminus\ell_\infty$. In the topologically twisted gauge theory, fermion fields become differential forms and these bundles are identified with the bundles of even/odd holomorphic forms over $\F_0$~\cite{BlauTh}. By (\ref{Ft1t2def}), the
holomorphic line bundles (\ref{bigrading}) have $\tilde T$-equivariant characters
$$
\ch_{\tilde T}\big(\Ocal_{\F_0}(m_z,m_w)\big) \=t_z^{-m_z}\,
t_w^{-m_w} \ ,
$$
and consequently
\beq
\ch_{\tilde T}\big(S^+\big)\=1+t_z\,t_w \qquad \mbox{and} \qquad
\ch_{\tilde T}\big(S^-\big)\=t_z+t_w \ .
\label{spinorbunschars}\eeq
Using (\ref{chTWVY}) and (\ref{spinorbunschars}), we may thus compute
the character of the universal sheaf (\ref{Ecalvirtual}) at a fixed
point in the instanton moduli space as
\bea
\ch_{\tilde T}(\Ecal)(\vec \lambda)&=&
\ch_{\tilde T}(W)(\vec \lambda)+\big(\ch_{\tilde T}(S^-)-
\ch_{\tilde T}(S^+)\big)~\ch_{\tilde T}(V)(\vec \lambda) \nonumber\\[4pt]
&=& \sum_{l=1}^N\,e_l\,\Big(1-(1-t_z)\,(1-t_w)\,
\sum_{(i,j)\in \lambda_l}\,t_z^{i-1}\,t_w^{j-1} \Big) \ ,
\label{chEcalF0}\eea
which coincides with the standard expression for instantons on
$\R^4$~\cite{NY}.

For completeness, let us record here the $\tilde T$-equivariant Chern character of the tangent bundle $T\calm_{N,k}(\F_0)$ over the instanton moduli space at the torus fixed points. As the only non-vanishing cohomology group of (\ref{adhmdefcomplex}) is
a model of the tangent space, it can be computed as the equivariant
index of this complex. For this, let
us recall some combinatorial definitions. Let $\lambda$ be a Young
diagram. Define the arm and leg lengths of a box $(i,j)\in\lambda$
respectively by
$$
A_\lambda(i,j)\= \varrho_i-j \qquad \mbox{and} \qquad L_\lambda(i,j)\= \varrho_j^t-i \
, 
$$
where $\varrho_i$ is the length of the $i$-th column of $\lambda$ and
$\varrho_j^t$ is the length of the $j$-th row of $\lambda$. At a fixed point $\vec\lambda$, the character can then be expressed 
after some algebra in terms of the characters of the representation as
$$
\ch_{\tilde T} \big(T \mathfrak{M}_{N,k}(\F_0) \big)(\vec\lambda) \= \sum_{l,m=1}^N \,e_l \,
e_m^{-1} ~ \Big(~ \sum_{(i,j) \in\lambda_l} \, t_z^{-L_{\lambda_l}(i,j)}\,
t_w^{A_{\lambda_m}(i,j) + 1} + \sum_{(i,j) \in \lambda_m} \,
t_z^{L_{\lambda_l}(i,j) + 1} \, t_w^{-A_{\lambda_m}(i,j)} ~ \Big) \ .
$$
Again this coincides with the standard result~\cite{BFMT,NY,Shadchin}.
As in ref.~\cite{LMN}, after toric localization one has
$$
\ch_{\tilde T}\big(T \mathfrak{M}_{N,k}(\F_0)\big) \= - \oint_{\F_0\times \calm_{N,k}(\F_0)}\, \ch_{\tilde T}(\mathfrak{E})\wedge
\ch_{\tilde T}\big(\mathfrak{E}^\vee\,\big)\wedge \mathrm{Td}_{\tilde T}(\F_0)
$$
at the $\tilde T$-fixed points $\vec\lambda$, up to a universal
perturbative contribution (the character of $W\otimes W^*$). This
expression formally generalizes to generic toric
surfaces~\cite{NekrasovICMP}, and is used to construct vertex gluing
rules below. From the
top Chern class one can also straightforwardly extract the equivariant Euler classes
$$
e_{\tilde T}\big(T\calm_{N,k}(\F_0)\big)(\vec\lambda) \= 
\prod_{l,m=1}^N\, n_{l,m}^{\vec\lambda} (\epsilon_z,\epsilon_w ,
\vec{a}) \ ,
$$
where
\bea
n_{l,m}^{\vec\lambda} (\epsilon_z,\epsilon_w , \vec{a}) &=&\prod_{(i,j)\in \lambda_l} \, \Big(a_m-a_l-L_{\lambda_m}(i,j)\,\epsilon_z+\big(
A_{\lambda_l}(i,j)+1\big)\,\epsilon_w \Big) \nonumber\\ && \times~
\prod_{(i',j'\,)\in \lambda_m}\,\Big(a_m-a_l +
\big(L_{\lambda_l(i',j'\,)}+1\big)\,\epsilon_z- A_{\lambda_m}(i',j'\,)\,\epsilon_w
\Big) \ . \nonumber
\end{eqnarray}

\subsubsection*{Framed modules on $\F_0$ from Dirac modules on $\PP^1$}

For the case of non-trivial $(m_z,m_w)$-framings, let
$$
\iota(m_z,m_w)\,:\,\calm_{N,k}(\F_0)~\longrightarrow~
\calm_{N\,;\,(N\,m_z,N\,m_w)\,;\,k(m_z,m_w)}(\F_0)
$$
be the isomorphism induced by the map $[\cale]\mapsto[\cale(m_z,m_w)]
=[\cale\otimes\Ocal_{\F_0}(m_z,m_w)]$, and let $\Ecal$ be a universal
sheaf on $\F_0\times\calm_{N,k}(\F_0)$. Then the torsion free sheaf
$$
\Ecal(m_z,m_w) \=\big(\id_{\F_0}\times\iota(m_z,m_w)^{-1}\big)^*(\Ecal)
$$
is a universal sheaf on
$\F_0\times\calm_{N;(N\,m_z,N\,m_w);k(m_z,m_w)}(\F_0)$. The equivariant character of $\Ecal(m_z,m_w)$ at the fixed points of the toric action modifies the decompositions (\ref{chTWVY}) to allow for ``propagators'' which appear along the edges $\PP^1$ of the toric diagram $\Delta(C)$ of a generic toric surface $C$. These modifications of the $\tilde T$-equivariant character of the
universal sheaf corresponding to the instanton sheaves
(\ref{Em1m2def}) are determined via shifts by the $\C^*$-equivariant character of
modules of solutions to the Dirac equation on $\PP^1$. We will now
describe these modules explicitly. As we have seen, the bundle
$\Ocal_{\PP^1}(m)$ of degree $m\in\Z$ over the Riemann sphere $\PP^1$
is the crucial ingredient in generating instantons with non-trivial
framings. After choosing a hermitean metric on the fibres of
$\Ocal_{\PP^1}(m)$, it is the holomorphic line bundle underlying the
standard Dirac monopole line bundle of topological charge $m$. The
form of the Dirac operator $\Dirac_m$ in the background of the
corresponding monopole gauge potential is well-known and can be
conveniently described following ref.~\cite{PS}.

Given homogeneous coordinates $[z_0,z_1]$ on $\PP^1$, let $y=z_1/z_0$
denote stereographic coordinates on the northern hemisphere. Then the
twisted Dirac operator in the monopole background of magnetic charge
$m\in\Z$ is given by
$$
\Dirac_{m} \=\begin{pmatrix}0&\Dirac^-_{m}\\\Dirac^+_{m}&0
\end{pmatrix} \ ,
$$
where
\bea
\Dirac^+_{m}&=&\frac1{2}\,\Big[\left(1+y\,\yb\right)\,
\frac\partial{\partial\yb}-
\frac12\,(m+1)\,y\Big] \ , \nonumber
\\[4pt] \Dirac^-_{m}&=&-\frac1{2}\,\Big[\left(1+y\,\yb\right)\,
\frac\partial{\partial y}+
\frac12\,(m-1)\,\yb\Big] \ .
\nonumber\eea
These operators act on sections of the chiral/antichiral spinor line
bundles associated to the twisted holomorphic spinor bundles
$\Ocal_{\PP^1}(\pm\,1)\otimes\Ocal_{\PP^1}(m)$. We are interested in
the subspaces of zero modes
\beq
S_m^\pm \=\ker\Dirac_m^\pm \ .
\label{Smpmdef}\eeq

The action of $t\in\C^*$ is implemented by the automorphism $F_t$ of
$\PP^1$ defined by
\beq
F_t(y,\yb) \=\big(t\,y\,,\,t^{-1}\,\yb\big) \ .
\label{FtP1}\eeq
The irreducible representations $\,\underline{T}\,^i\cong\C$ of $\C^*$ are labelled by
their weights $i\in\Z$ and are defined by $z\mapsto t\cdot z=t^i\,z$
for $t\in\C^*$ and $z\in\C$. The corresponding $\C^*$-eigenspace
decomposition of the modules (\ref{Smpmdef}) can be described for all
$m\in\Z$ as follows.
\begin{theorem}
The isotopical decompositions of the spinor modules $S_m^\pm$ over
$\PP^1$, as $\C^*$-modules, are given by
\bea
S_m^+&=&\bigoplus_{i=1}^{|m|}\,
\,\underline{T}\,^{i-1} \qquad \mbox{and} \qquad S_m^-\=\{0\} \quad
\mbox{for} \quad m<0 \ , \nonumber\\[4pt]
S_m^-&=&\bigoplus_{i=1}^m\,\,\underline{T}\,^{-(i-1)} 
\qquad \mbox{and} \qquad S_m^+\=\{0\}
\quad \mbox{for} \quad m>0 \ .
\label{SmpmP1decomps}\eea
\label{P1spinmoddecomp}\end{theorem}
\begin{proof}
The solutions of the Dirac equation are given by $L^2$-solutions of
the differential equations
\beq
\Dirac_m^\pm\psi^\pm_m \=0
\label{Diracpmeqs}\eeq
for the spinors $\psi_m^\pm\in\ker\Dirac_m^\pm$. The line bundle
$\Ocal_{\PP^1}(m)$ has holomorphic transition function $y^m$
transforming sections from the northern hemisphere to the southern
hemisphere of $\PP^1$, which after unitary reduction to a hermitean
line bundle becomes $(y/\yb)^{m/2}$. By using these
transition functions, it is easy to see that the only solutions of the
equations (\ref{Diracpmeqs}) which are regular on both the northern
and southern hemispheres are of the form
\bea
\psi_m^+&=&\frac1{(1+y\,\yb)^{(|m|-1)/2}}~\sum_{i=1}^{|m|}\,
\xi_i\,y^{i-1} \qquad \mbox{and} \qquad \psi_m^-\=0 \quad \mbox{for}
\quad m<0 \ , \nonumber\\[4pt]
\psi_m^-&=&\frac1{(1+y\,\yb)^{(m-1)/2}}~\sum_{i=1}^{m}\,
\tilde\xi_i\,\yb\,^{i-1} \qquad \mbox{and} \qquad \psi_m^+\=0 \quad
\mbox{for} \quad m>0 \ ,
\nonumber\eea
with constant coefficients $\xi_i,\tilde\xi_i\in\C$. The result now
follows from (\ref{FtP1}).
\end{proof}

We can use Theorem~\ref{P1spinmoddecomp} to compute the
$\C^*$-equivariant characters of the spinor modules
(\ref{Smpmdef}). One finds
\beq
\ch_{\C^*}\big(S_m^-\big)\=\sum_{i=1}^m\,t^{-(i-1)}\=
\frac{1-t^{-m}}{1-t^{-1}}
\label{spinchar+}\eeq
for $m>0$, while
\beq
\ch_{\C^*}\big(S_m^+\big)\=\sum_{i=1}^{|m|}\,t^{i-1}\=
\frac{1-t^{|m|}}{1-t}
\label{spinchar-}\eeq
for $m<0$. In the non-equivariant limit $t\to1$, the
characters (\ref{spinchar+}) and (\ref{spinchar-}) reproduce the known
index of the Dirac operator, ${\rm index}(\Dirac_m)=-m$, in the
monopole background. These characters shift the character (\ref{chEcalF0}). For example, with $m_z,m_w>0$ and $N=1$ one has
$$
\ch_{\tilde T}\big(\Ecal(m_z,m_w)\big)(\lambda)\= \ch_{\tilde T}(\Ecal)(\lambda) + (1-t_z)\,\ch_{\C_z^*}\big(S^-_{m_z}\big) + (1-t_w)\, \ch_{\C_w^*}\big(S_{m_w}^-\big) \ .
$$
This completes the derivation of the gluing rules of
ref.~\cite[Sec.~4.3]{CK-PS}.

\bigskip

\section{D2-brane gauge theory and Gromov--Witten invariants}

\noindent
In this final section we study the reduction of the four-dimensional
gauge theories of Section~\ref{D4GT} on local curves, and examine in
detail the example of local $\PP^1$ which has been extensively
described from the point of view of Vafa--Witten theory. We begin with
a somewhat heuristic description of how these two-dimensional
supersymmetric gauge theories are induced. Then we proceed to a more
formal topological field theory formalism which systematically
computes topological string amplitudes and Gromov--Witten
invariants. Finally, we briefly address wall-crossing issues once
again from a physical standpoint.

\subsection{$q$-deformed two-dimensional Yang--Mills theory\label{qYM2D}}

One of the simplest classes of examples are the non-compact local
Calabi--Yau threefolds which are fibred over curves. If $\Sigma_g\to X$ is
a holomorphically embedded curve of genus $g$ in a Calabi--Yau
threefold, then the holomorphic tangent bundle restricts to $\Sigma_g$
as $T\Sigma_g\oplus\cN_{\Sigma_g}$, where the normal bundle
$\cN_{\Sigma_g}$ is a holomorphic bundle of rank two over
$\Sigma_g$. By the Calabi--Yau condition $c_1(X)=0$, one has
$c_1(\cN_{\Sigma_g})=-\chi(\Sigma_g)=2g-2$. Thus in a neighbourhood of
$\Sigma_g$, the manifold $X$ looks like the total space of a
holomorphic bundle $\cN\to\Sigma_g$ of rank two with
$c_1(\cN)=2g-2$. Hence we consider the total space of the bundle
$$
X\= X_p\= \Ocal_{\Sigma_g}(p+2g-2)\oplus\Ocal_{\Sigma_g}(-p)
$$
over $\Sigma_g$. Since $\Ocal_{\Sigma_g}(p)$ is a holomorphic line
bundle over $\Sigma_g$ of degree $p$, one has
$$
c_1(X_p)\=(p+2g-2)+(-p)\=2g-2
$$ 
as required. For genus $g=0$, this is
just the example of local $\PP^1$ which was introduced in
Section~\ref{HJspaces}

In our applications to black hole microstate counting, we count bound
states of D4--D2--D0 branes in $X_p$ with $N$ D4-branes wrapping the
four-cycle $C_p$ which is the total space of the line bundle
$\Ocal_{\Sigma_g}(-p)\to\Sigma_g$, together with D2-branes wrapping
the base Riemann surface $\Sigma_g$ (embedded in $C$ and $X_p$ as the
zero section). One can then localize the path integral of the $\cN=4$
topological gauge theory on $C_p$ to the $\C^*$-invariant modes along
the fibre of $\Ocal_{\Sigma_g}(-p)$. The result is an effective $U(N)$
gauge theory on $\Sigma_g$ whose action is given by~\cite{AOSV}
\beq
S\= \frac1{g_s}\,\int_{\Sigma_g}\,\Tr(\Phi\,F)+\frac\vartheta{g_s}\,
\int_{\Sigma_g}\,\Tr(\Phi\,k)-\frac p{2g_s}\,\int_{\Sigma_g}\,
\Tr\big(\Phi^2\,k\big) \ ,
\label{SBF2D}\eeq
where $\vartheta=g_s\,\phi^2/2\pi$ and the last term is a mass
deformation which originates in four dimensions due to the non-triviality
of the bundle $\Ocal_{\Sigma_g}(-p)$~\cite{Vafa}. The scalar field $\Phi$ is
given by the holonomy of the four-dimensional gauge connection at
infinity
$$
\Phi(z)\= \int_{\Ocal_{\Sigma_g}(-p)_z}\, F(z) \=
\oint_{S^1_{z,|u|=\infty}}\, A \ ,
$$
where $z\in\Sigma_g$ and $u\in\Ocal_{\Sigma_g}(-p)$.

The action (\ref{SBF2D}) is just the BF-theory representation of
two-dimensional Yang--Mills theory, with the usual $F^2$-term coming
from performing the gaussian integral over $\Phi$ in the functional
integral. If $\Phi$ were arbitrary, then we could use diagonalization
techniques to reduce the partition function to the usual heat kernel
expansion of two-dimensional Yang--Mills theory. However, the
identification of $\Phi$ as a holonomy means that it should be treated
as a periodic field, and the change from a Lie algebra valued variable
to a group-like variable affects the Jacobian of the path integral
measure that arises from diagonalization. This modifies the usual
Migdal expansion of two-dimensional gauge theory to give the partition
function
\beq
Z_{U(N)}^{C_p}(q,Q)\= \sum_{R} \, \dim_q(R)^{2-2g} \,
q^{\frac p2\,C_2(R)}\,Q^{C_1(R)} \ ,
\label{qYMpartfn}\eeq
where $q=\e^{-g_s}$ as before and $Q=\e^{\ii\vartheta}$. The sum runs
over unitary irreducible representations $R=(R_1,\dots,R_N)$ of the
gauge group $U(N)$ with weights $R_i$, first and second Casimir
invariants $C_1(R)=\sum_i\,R_i$ and $C_2(R)=\sum_i\,R_i^2$, and quantum dimension
$$
\dim_q(R)\= \prod_{1\leq i<j\leq N}\,\frac{[R_i-R_j+j-i]_q}
{[j-i]_q} \ ,
$$
where the $q$-number is defined by
$$
[n]_q \ := \ \frac{q^{n/2}-q^{-n/2}}{q^{1/2}-q^{-1/2}}
$$
with $[n]_q=n+O(q-1)$ for $q\to1$. This reduced two-dimensional gauge
theory is a $q$-deformation of ordinary Yang--Mills theory.

\subsection{Topological field theory}

We will now give a more precise definition of this gauge theory using
techniques of two-dimensional topological field theory. This formalism can then be used
to compute topological string amplitudes. The idea is that the cutting
and pasting of base Riemann surfaces is equivalent to the cutting and
pasting of the corresponding local Calabi--Yau threefolds, by either
adding or cancelling off D-branes corresponding to the
boundaries in the topological vertex formalism~\cite{BP,AOSV}. The operations of gluing manifolds satisfy all axioms of a two-dimensional topological field
theory. By computing the open string topological A-model amplitudes on
a few Calabi--Yau manifolds, we then get \emph{all} others by
gluing. In the following we focus on the
genus zero case $\Sigma_0=\PP^1$ for simplicity and illustrative
purposes. 

When the base of the fibration is the sphere, we will only need the
basic \emph{Calabi--Yau cap amplitude} corresponding to a disk $D$,
represented symbolically as
$$
\mbox{\includegraphics[width=3cm]{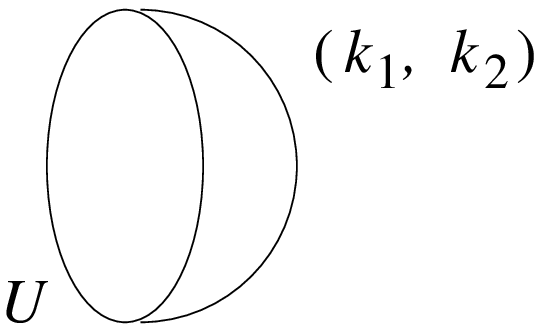}}
$$
where the \emph{levels} $k_1=e(\call_1)$ and $k_2=e(\call_2)$ label the degrees in $H^2(D,\partial
D)$ of a pair of line bundles $\call_1\oplus\call_2\to D$, and the unitary
matrix $U$ labels the holonomy of a gauge connection on $\partial
D\cong S^1$. Under concatenation the levels add. The basic cap
amplitude is defined by
$$
Z\bigg(~
\mbox{\includegraphics[width=2cm]{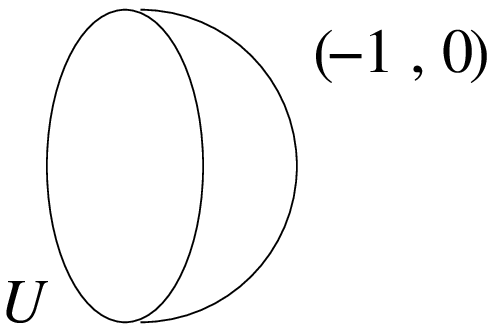}}
~\bigg) \= \sum_R\,d_q(R)\,q^{-k_R/4}\,\ch_R(U) \ ,
$$
where $q=\e^{-g_s}$ as before, the sum over $R$ runs through the space of
symmetric functions (Young tableaux), i.e. irreducible representations
of $SU(\infty)$, and $\ch_R(U)$ is the character of the holonomy $U$
in the representation $R$. The quantity
$$
d_q(R) \ := \ \prod_{(i,j)\in R}\,\frac1{\big[h(i,j)\big]_q}
$$
is the quantum dimension of the symmetric group representation
corresponding to the Young diagram of $R$, and
$$
k_R \= 2\,\sum_{(i,j)\in R}\,(i-j) \ .
$$
The quantity
$$
d_q(R)\,q^{-k_R/4} \ =: \ C_{R,\emptyset,\emptyset}(q) \=
W_{R,\emptyset}(q)
$$
is the topological vertex amplitude for the A-model topological string
theory with a single stack of D-branes on the non-compact Calabi--Yau
threefold $X=\C^3$. Similarly, define the cap amplitude
$$
Z\bigg(~
\mbox{\includegraphics[width=2cm]{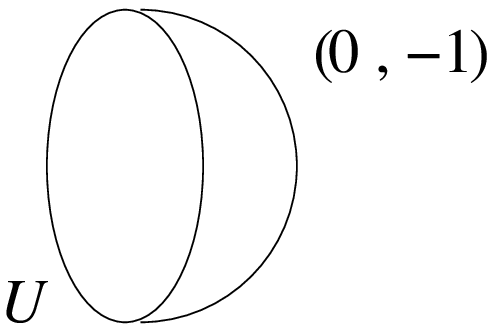}}
~\bigg) \= \sum_R\,d_q(R)\,q^{k_R/4}\,\ch_R(U) \ .
$$

For the gluing rules, we sew two open Riemann surfaces $\Sigma_L$ and
$\Sigma_R$ together along their common boundary to get the Riemann
surface $\Sigma_{L\cup
  R}$. For this, the orientations of the corresponding boundary
circles must be opposite. Reversing the orientation acts on the
boundary holonomy as $U\to U^{-1}$, and hence the gluing rules read
$$
Z(\Sigma_{L\cup R}) \= \int_{SU(\infty)}\, \dd U~Z(\Sigma_L)(U)\,
Z(\Sigma_R)(U^{-1})
$$
where $\dd U$ is the bi-invariant Haar measure and the integrals
can be evaluated explicitly by using the orthogonality relations for
characters
$$
\int_{SU(\infty)}\,\dd U~ \ch_{R_1}(U)\, \ch_{R_2}(U^{-1}) \=
\delta_{R_1,R_2} \ .
$$
We will also need the the Calabi--Yau trinion (or ``pants'')
amplitudes
$$
Z\bigg(~
\mbox{\includegraphics[width=4cm]{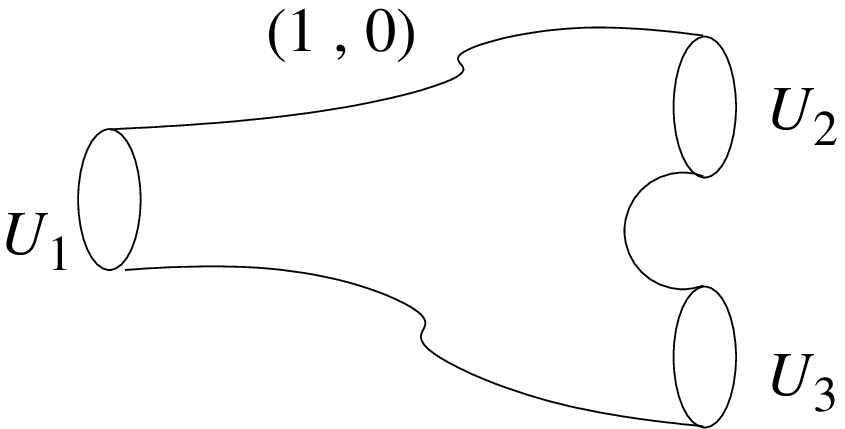}}
~\bigg) \= \sum_R\,\frac1{d_q(R)}\,q^{k_R/4}~ \prod_{i=1}^3\, \ch_R(U_i)
$$
and
$$
Z\bigg(~
\mbox{\includegraphics[width=4cm]{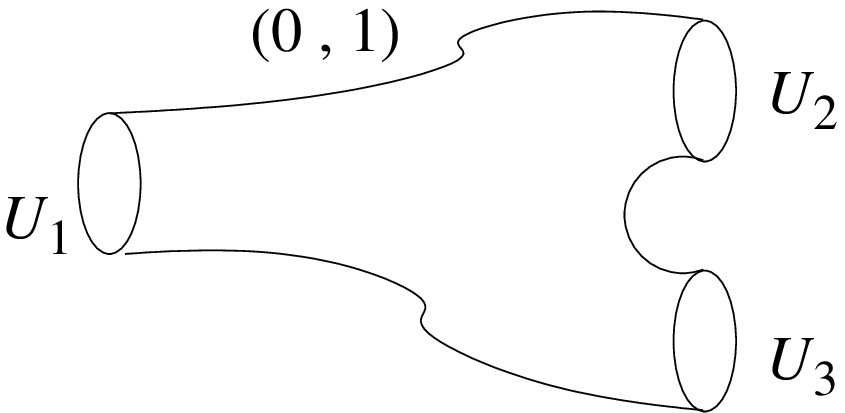}}
~\bigg) \= \sum_R\,\frac1{d_q(R)}\,q^{-k_R/4}~ \prod_{i=1}^3\,
\ch_R(U_i) \ .
$$

Let us point out two particularly noteworthy aspects of this construction
thus far. Firstly, the construction defines a functor of tensor
categories
$$
Z\,:\, {\sf S}_{\call_1,\call_2}~\longrightarrow~ {\sf Rep} \ ,
$$
where ${\sf Rep}$ is the representation category of $SU(\infty)$ and
${\sf S}_{\call_1,\call_2}$ is the geometric tensor category defined
as follows. The objects of ${\sf S}_{\call_1,\call_2}$ are compact
oriented one-manifolds $Y$, i.e. disjoint unions of oriented
circles. A morphism $Y_1\to Y_2$ between two objects of ${\sf
  S}_{\call_1,\call_2}$ is a triple $(W,Y_1,Y_2)$, where $W$ is an
oriented cobordism between $Y_1$ and $Y_2$, i.e. a smooth oriented two-manifold with
boundary $\partial W=Y_1\amalg (-Y_2)$, and the complex line bundles
$\call_1,\call_2$ over $W$ are trivialized on $\partial W$. This is
analogous to the Baum--Douglas description of D-branes in
K-homology~\cite{RS}. There is a
natural notion of equivalence provided by boundary preserving,
oriented diffeomorphisms $f:W\to W'$ with $\call_i\cong f^*\call_i'$,
$i=1,2$. Composition of morphisms is given by concatenation of
cobordism and gluing of bundles along the concatenation using the
trivializations. For a connected cobordism $W$, we can label the
isomorphism classes of a pair of line bundles $(\call_1,\call_2)$ by
the \emph{levels} $(k_1,k_2)$, where $k_i=e(\call_i)\in H^2(W,\partial
W)$. Under concatenation, these levels add.

Secondly, the $\C^*$-action on $\PP^1$ lifts to an action of the torus
$T=(\C^*)^2$ on $X_p$, via the natural scaling action on the fibres. The Gromov--Witten invariants in this case are
\emph{defined} via the virtual localization formula as a residue
integral over the $T$-fixed point locus, along the lines explained in
Section~\ref{Eqint} A stable map to $X_p$ which is $T$-invariant
factors through the zero section. It follows that there is a natural
isomorphism of moduli spaces
$\calm_g(X_p,\beta)^T=\calm_g(\PP^1,\beta)$, with $\beta=d\in\Z$, and consequently
\beq
\big[\calm_g(X_p,d)^T\,\big]^{\rm vir} \=
\big[\calm_g(\PP^1,d) \big]^{\rm vir} \ .
\label{XpP1rel}\eeq
This equality of virtual fundamental cycles implies that the
invariants constructed by integrating over each moduli space
coincide. It means that topological string theory on $X_p$ is
equivalent to a field theory on $\PP^1$, which is just the reduction
we argued in Section~\ref{qYM2D} In this case the Gromov--Witten
invariants of $X_p$ correspond to degree $d$ Hurwitz numbers of
$\PP^1$.

Returning to our computations, we get the annulus (or tube) amplitude
by contracting the cap and trinion amplitudes to get
\bea
Z\bigg(~
\mbox{\includegraphics[width=4cm]{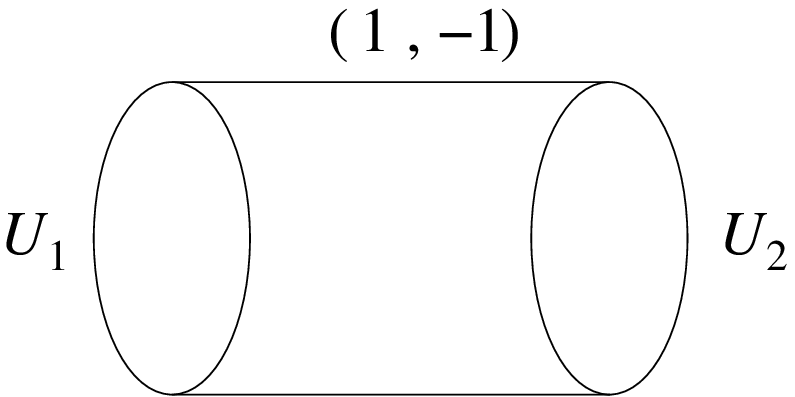}}
~\bigg) &=& Z\bigg(~
\mbox{\includegraphics[width=5cm]{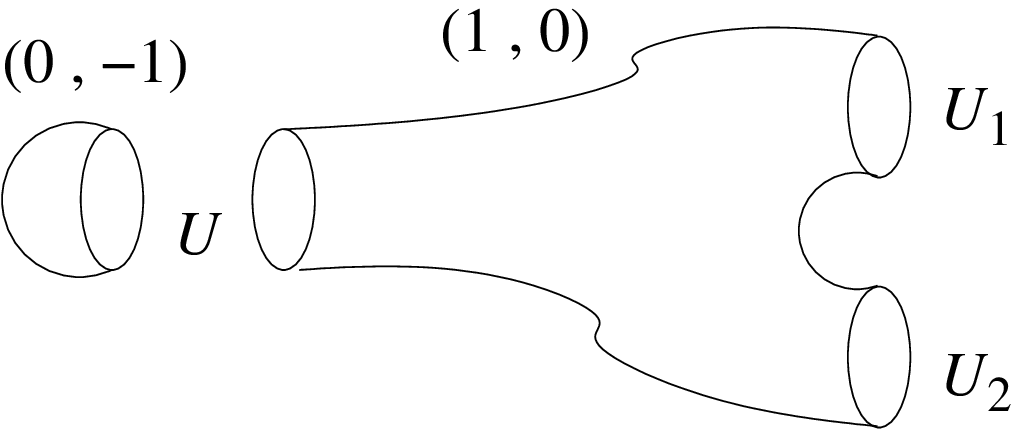}}
~\bigg) \nonumber\\[4pt]
&=& \sum_R\,q^{k_R/2}\,\ch_R(U_1)\,\ch_R(U_2) \ .
\nonumber\eea
Finally, we compute the amplitude of the fibration $X_p\to \PP^1$ by
using (\ref{XpP1rel}) and the gluing rules. In order to get the
appropriate bundle degrees $(p-2,-p)$, we glue $p$ tubes between two
caps to get
\bea
Z(X_p) &=& Z\bigg(~
\mbox{\includegraphics[width=2cm]{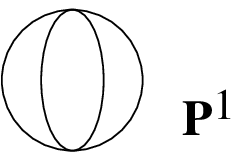}}
~\bigg) \nonumber\\[4pt] &=& Z\bigg(~
\mbox{\includegraphics[width=6cm]{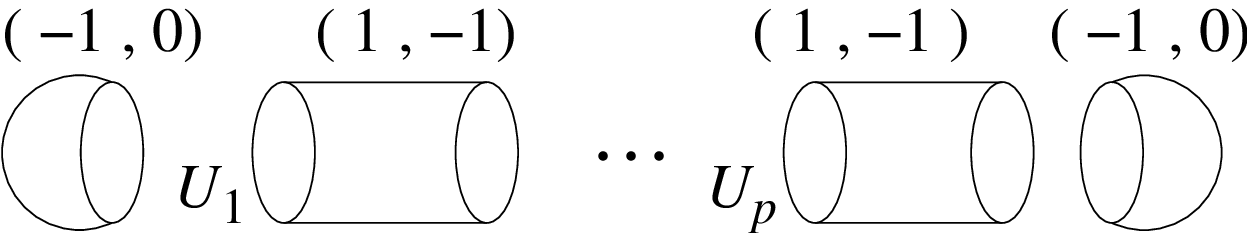}}
~\bigg) \nonumber\\[4pt] &=& \sum_R\, d_q(R)^2\, q^{(p-1)\,k_R/2} \ .
\label{ZXp}\eea
For $p=1$, the threefold $X_1\cong K_{\PP^1}^{1/2}\oplus
K_{\PP^1}^{1/2}$ is the resolved conifold
and this formula is a $q$-deformation of the classical Hurwitz formula
counting unramified covers of the Riemann sphere~$\PP^1$.

In fact, the quantity $d_q(R)$ is analogously the $N\to\infty$ limit
of the quantum dimension $\dim_q(R)$ of the $U(N)$ representation with
the same Young tableau. This follows by Schur--Weyl reciprocity which
writes $SU(N)$ representations in terms of representations of the
symmetric group $S_N$. Let $R$ be a representation corresponding to a
Young diagram $\lambda$ with row lengths $\lambda_i$, $i=1,\dots,d(\lambda)$,
where $\lambda_1\geq\lambda_2\geq\dots$ and $d(\lambda)$ is the number
of rows in $\lambda$. Set $\mu=q^N$, and define another $q$-number by
$$
[n]_\mu\=
\frac{\mu^{1/2}\,q^{n/2}-\mu^{-1/2}\,q^{-n/2}}{q^{1/2}-q^{-1/2}} \ .
$$
Then the quantum dimension of $R$ can be written as
$$
\dim_q(R) \= \prod_{1\leq i<j\leq d(\lambda)}\,
\frac{[\lambda_i-\lambda_j+j-i]_q}{[j-i]_q}~
\prod_{i=1}^{d(\lambda)}\,
\frac{\prod\limits_{v=1-i}^{\lambda_i-i}\,
  [v]_\mu}{\prod\limits_{v=1}^{\lambda_i}\, 
  \big[v-i+d(\lambda)\big]_q} \ .
$$
This is a Laurent polynomial in $\mu^{\pm\,1/2}$ whose coefficients
are rational functions of $q^{\pm\,1/2}$. The leading power of $\mu$
is $\frac{|\lambda|}2$, and the coefficient of this power is the
rational function of $q^{\pm\,1/2}$ given by
$$
q^{-k_R/4}\,\prod_{1\leq i<j\leq d(\lambda)}\,
\frac{[\lambda_i-\lambda_j+j-i]_q}{[j-i]_q}~\prod_{i=1}^{d(\lambda)}~ \prod\limits_{v=1}^{\lambda_i}\, 
\frac1{\big[v-i+d(\lambda)\big]_q} \= d_q(R)\,q^{-k_R/4} \ .
$$
As before, this is just the topological vertex amplitude
$C_{R,\emptyset,\emptyset}(q)$. The limit $q\to1$ gives the ordinary
dimension of the representation $R$.

Generalizing the generating function (\ref{ZXp}) thereby gives
$$
Z_{q{\rm YM}}\big(\PP^1\big) \= \sum_R\,\dim_q(R)^2\,q^{\frac p2\,C_2(R)} \ ,
$$
where now the sum runs over irreducible representations $R$ of $U(N)$ and
$C_2(R)=k_R+N(R)$ is the second Casimir invariant of $R$. By the
attractor mechanism, the K\"ahler modulus of $\PP^1$ is given by
$t=\frac N2\,(p-2)\,g_s$. In the limit $q\to1$, this two-dimensional
field theory coincides with ordinary Yang--Mills theory on the sphere
$\PP^1$. Note that this is a weak-coupling limit $g_s\to0$, with
$p\,g_s=g_{\rm YM}^2\,A$ and $A$ the area of $\PP^1$. Using similar
constructions as above, it is possible to formulate this
gauge theory directly as a two-dimensional topological field theory,
without any reference to the extrinsic bundle structure of the local
threefold $X_p$. Explicit computations of the corresponding
Gromov--Witten invariants of $X_p$ can be found in
refs.~\cite{CCGPSS1,CCGPSS2}.

\subsection{Wall-crossing formulas}

Finally, let us briefly comment on the relationship to black hole
entropy counting. Reinstating the $\vartheta$-angle as in
(\ref{qYMpartfn}), the attractor mechanism fixes the K\"ahler modulus
of $\Sigma_0=\PP^1$ as
$$
t\= 2\pi\ii \frac{X^1}{X^0} \=(p-2)\,\frac{N\, g_s}2-\ii\vartheta \ .
$$
The large $N$ limit of the partition function (\ref{qYMpartfn}) should
possess this modulus. The modulus squared structure anticipated by the
OSV relation (\ref{OSV}) is given in this limit by~\cite{AOSV}
$$
Z_{q{\rm YM}}\big(\PP^1\big) \= \sum_{r\in\IZ}~\sum_{R_1,R_2}\,
Z_{R_1,R_2}^{q{\rm YM}+}(t+p\,g_s\,r)\, Z_{R_1,R_2}^{q{\rm
    YM}-}(\,\overline{t}+p\,g_s\,r) \ ,
$$
where the sum over $r$ runs through $U(1)$ charges in the local
decomposition $U(N)\sim U(1)\times SU(N)$ of the gauge group, the
second sum runs through pairs of irreducible $SU(N)$ representations
$R_1,R_2$, and $Z_{R_1,R_2}^{q{\rm YM}+}(t)$ is the perturbative
topological string amplitude on $X_p$ with two stacks of D-branes in
the fibre. The conjugate amplitude is $Z_{R_1,R_2}^{q{\rm YM}-}(t)=
(-1)^{C_1(R_1)+C_1(R_2)}\,Z_{R_1^t,R_2^t}^{q{\rm YM}+}(t)$. Thus in
this case the OSV relation (\ref{OSV}) is modified to the symbolic
form
$$
Z_{\rm BH} \= \sum_{r\in\IZ}~\sum_b\,\big|Z_{\rm top}^{(b,r)}\big|^2 \
,
$$
where the first sum runs over Ramond--Ramond fluxes $r$ through the base
Riemann surface $\Sigma_g$, while the second sum is over fibre
D-branes which carry additional moduli $\widehat{t}$ measuring their
distances from the base $\Sigma_g$ along the non-compact fibre
directions.

Further analysis of the relationships with $\mathcal{N}=4$ Yang--Mills
theory on $C_p$ are studied in refs.~\cite{CCGPSS1,CCGPSS3,GSST,BT},
where it is observed that the instanton expansions of the two-dimensional and four-dimensional
gauge theory partition functions differ by perturbative
contributions. These extra factors arise from the non-compactness of
the surface $C_p$. When $C_p$ is the resolution of an $A_{p,n}$
singularity, its boundary is the three-dimensional Lens space
$L(p,n)=S^3/\Gamma_{(p,n)}$. The perturbative factors can then be
identified as the partition function of Chern--Simons gauge theory on
the boundary, and arise as a consequence of the fact that the
two-dimensional gauge theory implicitly integrates over all boundary
conditions. Once these boundary contributions are stripped, the
enumeration of instantons in the two-dimensional and four-dimensional
gauge theories coincide~\cite{GSST}.

\bigskip

\section*{Acknowledgments}

\noindent
It is a pleasure to thank M.~Cirafici, E.~Gasparim, A.-K.~Kashani-Poor, A.~Maciocia and
A.~Sinkovics for enjoyable discussions and collaborations on some of
the results presented here. This work was supported in part by grant ST/G000514/1 ``String Theory
Scotland'' from the UK Science and Technology Facilities Council.

\bigskip

%\newpage


\begin{thebibliography}{99}
%\addtolength{\itemsep}{-6pt}

\bibitem{AOSV}
  M.~Aganagic, H.~Ooguri, N.~Saulina and C.~Vafa,
  ``Black holes, $q$-deformed $2D$ Yang--Mills, and non-perturbative topological
  strings,''
  Nucl.\ Phys.\  B {\bf 715} (2005) 304--348
  [arXiv:hep-th/0411280].
  %%CITATION = NUPHA,B715,304;%%

\bibitem{AOVY}
  M.~Aganagic, H.~Ooguri, C.~Vafa and M.~Yamazaki,
  ``Wall-crossing and M-theory,''
  arXiv:0908.1194~[hep-th].
  %%CITATION = ARXIV:0908.1194;%%

\bibitem{AGT}
L.F.~Alday, D.~Gaiotto and Y.~Tachikawa,
``Liouville correlation functions from four-dimensional gauge
theories,''
Lett. Math. Phys. {\bf 91} (2010) 167--197
[arXiv:0906.3219~[hep-th]].

\bibitem{AM}
  E.~Andriyash and G.W.~Moore,
  ``Ample D4--D2--D0 decay,''
  arXiv:0806.4960~[hep-th].
  %%CITATION = ARXIV:0806.4960;%%

\bibitem{Aspinwall1}
P.S.~Aspinwall,
``D-branes on Calabi--Yau manifolds'',
in: {\sl Theoretical Advanced Study Institute in Elementary Particle
  Physics (TASI 2003): Recent Trends in String Theory},
ed. J.M.~Maldacena, World Scientific (2005) 1--152 
[arXiv:hep-th/0403166].
%%CITATION = HEP-TH 0403166;%%

%\cite{BGK}
\bibitem{BGK}
V.~Baranovsky, V.~Ginzburg and A.~Kuznetsov,
``Wilson's grassmannian and a noncommutative quadric,''
Int. J. Math. Res. Not. {\bf 21} (2003) 1155--1197
[arXiv:math/0203116~[math.AG]].

\bibitem{BKS}
  L.~Baulieu, H.~Kanno and I.M.~Singer,
  ``Special quantum field theories in eight and other dimensions,''
  Commun.\ Math.\ Phys.\  {\bf 194} (1998) 149--175
  [arXiv:hep-th/9704167].
  %%CITATION = HEP-TH 9704167;%%

\bibitem{Behrend}
K.~Behrend,
``Donaldson--Thomas invariants via microlocal geometry,''
Ann. Math. {\bf 170} (2009) 1307--1338
[arXiv:math.AG/0507523].

\bibitem{BlauTh}
  M.~Blau and G.~Thompson,
  ``Euclidean SYM theories by time reduction and special holonomy  manifolds,''
  Phys.\ Lett.\ B {\bf 415} (1997) 242--252
  [arXiv:hep-th/9706225].
  %%CITATION = HEP-TH 9706225;%%

\bibitem{BT}
  G.~Bonelli and A.~Tanzini,
  ``Topological gauge theories on local spaces and black hole entropy
  countings,''
  Adv.\ Theor.\ Math.\ Phys.\  {\bf 12} (2008) 1429--1446
  [arXiv:0706.2633~[hep-th]].
  %%CITATION = 00203,12,6;%%

\bibitem{BN}
  H.W.~Braden and N.A.~Nekrasov,
  ``Spacetime foam from noncommutative instantons,''
  Commun.\ Math.\ Phys.\  {\bf 249} (2004) 431--448
  [arXiv:hep-th/9912019].
  %%CITATION = CMPHA,249,431;%%

\bibitem{BPT}
U.~Bruzzo, R.~Poghossian and A.~Tanzini,
``Poincar\'e polynomial of moduli spaces of framed sheaves on (stacky)
Hirzebruch surfaces,''
arXiv:0909.1458~[math.AG].

\bibitem{BFMT}
  U.~Bruzzo, F.~Fucito, J.F.~Morales and A.~Tanzini,
  ``Multi-instanton calculus and equivariant cohomology,''
  J. High Energy Phys. {\bf 0305} (2003) 054
  [arXiv:hep-th/0211108].
  %%CITATION = JHEPA,0305,054;%%

\bibitem{BP}
J.~Bryan and R.~Pandharipande,
``The local Gromov--Witten theory of curves,''
J. Amer. Math. Soc. {\bf 21} (2008) 101--136
[arXiv:math/0411037~[math.AG]].

\bibitem{CCGPSS1}
  N.~Caporaso, M.~Cirafici, L.~Griguolo, S.~Pasquetti, D.~Seminara and R.J.~Szabo,
  ``Topological strings and large $N$ phase transitions. I: Nonchiral expansion
  of $q$-deformed Yang--Mills theory,''
  J. High Energy Phys. {\bf 0601} (2006) 035
  [arXiv:hep-th/0509041]; ``Black holes, topological strings and large $N$ phase transitions,''
  J.\ Phys.\ Conf.\ Ser.\  {\bf 33} (2006) 13--25
  [arXiv:hep-th/0512213].
  %%CITATION = 00462,33,13;%% .
  %%CITATION = JHEPA,0601,035;%%

\bibitem{CCGPSS2}
  N.~Caporaso, M.~Cirafici, L.~Griguolo, S.~Pasquetti, D.~Seminara and R.J.~Szabo,
  ``Topological strings and large $N$ phase transitions. II: Chiral expansion  of
  $q$-deformed Yang--Mills theory,''
  J. High Energy Phys. {\bf 0601} (2006) 036
  [arXiv:hep-th/0511043].
  %%CITATION = JHEPA,0601,036;%%

\bibitem{CCGPSS3}
  N.~Caporaso, M.~Cirafici, L.~Griguolo, S.~Pasquetti, D.~Seminara and R.J.~Szabo,
  ``Topological strings, two-dimensional Yang--Mills theory and Chern--Simons
  theory on torus bundles,''
  Adv.\ Theor.\ Math.\ Phys.\  {\bf 12} (2008) 981--1058
  [arXiv:hep-th/0609129].
  %%CITATION = 00203,12,981;%%

\bibitem{CJ}
  W.-y.~Chuang and D.L.~Jafferis,
  ``Wall-crossing of BPS states on the conifold from Seiberg duality and
  pyramid partitions,''
  Commun.\ Math.\ Phys.\  {\bf 292} (2009) 285--301
  [arXiv:0810.5072~[hep-th]].
  %%CITATION = CMPHA,292,285;%%

\bibitem{CK-PS}
M.~Cirafici, A.-K.~Kashani-Poor and R.J.~Szabo,
``Crystal melting on toric surfaces,''
arXiv:0912.0737~[hep-th].
  %%CITATION = ARXIV:0912.0737;%%

\bibitem{CSS}
 M.~Cirafici, A.~Sinkovics and R.J.~Szabo,
  ``Cohomological gauge theory, quiver matrix models and Donaldson--Thomas
  theory,''
  Nucl.\ Phys.\  B {\bf 809} (2009) 452--518
  [arXiv:0803.4188~[hep-th]]; ``Instantons and Donaldson--Thomas invariants,''
  Fortsch.\ Phys.\  {\bf 56} (2008) 849--855
  [arXiv:0804.1087~[hep-th]].
  %%CITATION = FPYKA,56,849;%%
  %%CITATION = NUPHA,B809,452;%%

\bibitem{CLS}
L.~Cirio, G.~Landi and R.J.~Szabo,
``Algebraic deformations of toric varieties I. General
constructions,'' 
arXiv:1001.1242~[math.QA].

\bibitem{coxrev}
D.A.~Cox,
``Toric varieties and toric resolutions,''
Progr. Math. {\bf 181} (2000) 259--284;
``What is a toric variety?,''
Contemp. Math. {\bf 334} (2003) 203--223.

\bibitem{dBDE-SMVdeB}
  J.~de~Boer, F.~Denef, S.~El-Showk, I.~Messamah and D.~Van~den~Bleeken,
  ``Black hole bound states in $AdS_3 \times S^2$,''
  J. High Energy Phys. {\bf 0811} (2008) 050
  [arXiv:0802.2257~[hep-th]].
  %%CITATION = JHEPA,0811,050;%%

\bibitem{DenefMoore}
  F.~Denef and G.W.~Moore,
  ``Split states, entropy enigmas, holes and halos,''
  arXiv:hep-th/0702146.
  %%CITATION = HEP-TH/0702146;%%

\bibitem{DM}
  D.-E.~Diaconescu and G.W.~Moore,
  ``Crossing the wall: Branes vs. bundles,''
  arXiv:0706.3193~[hep-th].
  %%CITATION = ARXIV:0706.3193;%%

%\cite{Dijkgraaf:2007fe}
\bibitem{DS}
  R.~Dijkgraaf and P.~Sulkowski,
  ``Instantons on ALE spaces and orbifold partitions,''
  J. High Energy Phys. {\bf 0803} (2008) 013
  [arXiv:0712.1427~[hep-th]].
  %%CITATION = JHEPA,0803,013;%%

%\cite{Dijkgraaf:2007sw}
\bibitem{DHSV}
  R.~Dijkgraaf, L.~Hollands, P.~Sulkowski and C.~Vafa,
  ``Supersymmetric gauge theories, intersecting branes and free fermions,''
  J. High Energy Phys. {\bf 0802} (2008) 106
  [arXiv:0709.4446~[hep-th]].
  %%CITATION = JHEPA,0802,106;%%

\bibitem{DN}
M.R.~Douglas and N.A.~Nekrasov,
``Noncommutative field theory,''
Rev. Mod. Phys. {\bf 73} (2002) 977--1029 [arXiv:hep-th/0106048].
%%CITATION = HEP-TH 0106048;%%

\bibitem{ES}
G.~Ellingsrud and S.A.~Str\o mme,
``On the homology of the Hilbert scheme of points in the plane,''
Invent. Math. {\bf 87} (1987) 343--352.

\bibitem{FMP}
  F.~Fucito, J.F.~Morales and R.~Poghossian,
  ``Instanton on toric singularities and black hole countings,''
  J. High Energy Phys. {\bf 0612} (2006) 073
  [arXiv:hep-th/0610154].
  %%CITATION = JHEPA,0612,073;%%

\bibitem{fulton}
W.~Fulton,
``Introduction to toric varieties,''
Ann. Math. {\bf 131}, Princeton Univ. Press (1997).

\bibitem{GMN}
  D.~Gaiotto, G.W.~Moore and A.~Neitzke,
  ``Wall-crossing, Hitchin systems, and the WKB approximation,''
  arXiv:0907.3987~[hep-th].
  %%CITATION = ARXIV:0907.3987;%%

\bibitem{GL}
 E.~Gasparim and C.-C.M.~Liu,
  ``The Nekrasov conjecture for toric surfaces,''
  arXiv:0808.0884~[math.AG].
  %%CITATION = ARXIV:0808.0884;%%

\bibitem{GKM}
E.~Gasparim, T.~K\"oppe and P.~Majumdar,
``Local holomorphic Euler characteristic and instanton decay,''
Pure Appl. Math. Quart. {\bf 4} (2008) 161--179
[arXiv:0709.2577~[math.AG]].

%\cite{Gieseker}
\bibitem{Gieseker}
D.~Gieseker,
``On the moduli of vector bundles on an algebraic surface,''
Ann. Math. {\bf 106} (1977) 45--60.

%\cite{G-SV}
\bibitem{G-SV}
    G.~Gonzalez-Sprinberg and J.L.~Verdier,
    ``Construction g\'eometrique de la correspondence de McKay,''
     Ann. Sci. \'Ecole Norm. Sup. {\bf 16} (1983) 409--449.

\bibitem{GopVafa}
  R.~Gopakumar and C.~Vafa,
  ``M-theory and topological strings I and II,''
  arXiv:hep-th/9809187 and arXiv:hep-th/9812127.
  %%CITATION = HEP-TH/9812127;%%
  %%CITATION = HEP-TH/9809187;%%

%\cite{GraPand}
\bibitem{GraPand}
T.~Graber and R.~Pandharipande,
``Localization of virtual classes,''
Invent. Math. {\bf 135} (1999) 487--518
[arXiv:alg-geom/9708001].

\bibitem{GSST}
  L.~Griguolo, D.~Seminara, R.J.~Szabo and A.~Tanzini,
  ``Black holes, instanton counting on toric singularities and $q$-deformed
  two-dimensional Yang--Mills theory,''
  Nucl.\ Phys.\  B {\bf 772} (2007) 1--24
  [arXiv:hep-th/0610155].
  %%CITATION = NUPHA,B772,1;%%

\bibitem{HP}
  C.~Hofman and J.-S.~Park,
  ``Cohomological Yang--Mills theories on K\"ahler 3-folds,''
  Nucl.\ Phys.\ B {\bf 600} (2001) 133--162
  [arXiv:hep-th/0010103].
  %%CITATION = HEP-TH 0010103;%%

%\cite{HL}
\bibitem{HL}
D.~Huybrechts and M.~Lehn,
``Stable pairs on curves and surfaces,''
J. Algebraic Geom. {\bf 4} (1995) 67--104
[arXiv:alg-geom/9211001].

%\cite{HLbook}
\bibitem{HLbook}
D.~Huybrechts and M.~Lehn,
``The geometry of moduli spaces of sheaves,''
Aspects Math. E {\bf 31}, Friedr. Vieweg Sohn (1997).

\bibitem{INOV}
  A.~Iqbal, N.A.~Nekrasov, A.~Okounkov and C.~Vafa,
  ``Quantum foam and topological strings,''
  J. High Energy Phys. {\bf 0804} (2008) 011
  [arXiv:hep-th/0312022].
  %%CITATION = JHEPA,0804,011;%%

\bibitem{JM}
  D.L.~Jafferis and G.W.~Moore,
  ``Wall-crossing in local Calabi--Yau manifolds,''
  arXiv:0810.4909~[hep-th].
  %%CITATION = ARXIV:0810.4909;%%

\bibitem{Joyce}
D.D.~Joyce,
``Configurations in abelian categories IV. Invariants and changing stability conditions,''
Adv. Math. {\bf 217} (2008) 125--204
[arXiv:math.AG/0410268].

\bibitem{KS}
A.~Konechny and A.S.~Schwarz,
``Introduction to matrix theory and noncommutative geometry,''
Phys. Rept. {\bf 360} (2002) 353--465 [arXiv:hep-th/0012145]; [arXiv:hep-th/0107251].
%%CITATION = HEP-TH 0012145;%%
%%CITATION = HEP-TH 0107251;%%

\bibitem{KontSoib}
M.~Kontsevich and Y.~Soibelman,
``Stability structures, motivic Donaldson--Thomas invariants and cluster transformations,''
  arXiv:0811.2435~[math.AG].

\bibitem{KrausSh}
P.~Kraus and M.~Shigemori, ``Noncommutative instantons and the
Seiberg--Witten map,'' J. High Energy Phys. {\bf 0206} (2002) 034
[arXiv:hep-th/0110035].

\bibitem{KN}
    P.B.~Kronheimer and H.~Nakajima,
    ``Yang--Mills instantons on ALE gravitational instantons,''
    Math. Ann. \textbf{288}
    (1990) 263--307.

\bibitem{Kuznetsov}
    A.~Kuznetsov,
    ``Quiver varieties and Hilbert schemes,''
    Moscow Math. J. {\bf 7} (2007) 673--697
    [arXiv:math/0111092~[math.AG]].

\bibitem{LL}
  J.M.F.~Labastida and C.~Lozano,
  ``Mathai--Quillen formulation of twisted $\mathcal{N} = 4$ supersymmetric gauge  theories
  in four dimensions,''
  Nucl.\ Phys.\  B {\bf 502} (1997) 741--790
  [arXiv:hep-th/9702106].
  %%CITATION = NUPHA,B502,741;%%

\bibitem{LMN}
A.S.~Losev, A.~Marshakov and N.A.~Nekrasov,
``Small instantons, little strings and free fermions,''
in: {\sl From Fields to Strings: Circumnavigating Theoretical Physics}, Vol.~1, eds. M.~Shifman, A.~Vainshtein and J.F.~Wheater, World Scientific (2004) 581--621 [arXiv:hep-th/0302191].

\bibitem{MSW}
  J.M.~Maldacena, A.~Strominger and E.~Witten,
  ``Black hole entropy in M-theory,''
  J. High Energy Phys. {\bf 9712} (1997) 002
  [arXiv:hep-th/9711053].
  %%CITATION = JHEPA,9712,002;%%

\bibitem{Marino}
  M.~Mari\~no,
``Chern--Simons theory and topological strings,'' Rev.\
Mod.\ Phys.\  {\bf 77} (2005) 675--720 [arXiv:hep-th/0406005];
  ``Chern--Simons theory, matrix models, and topological strings,''
  Int. Ser. Mon. Phys. {\bf 131}, Oxford Univ. Press (2005).

%\cite{Maruyama}
\bibitem{Maruyama}
M.~Maruyama,
``Moduli of stable sheaves II,''
J. Math. Kyoto Univ. {\bf 18} (1978) 557--614;
``On boundedness of families of torsion free sheaves,''
J. Math. Kyoto Univ. {\bf 21} (1981) 673--701.

\bibitem{MNOP}
    D.~Maulik, N.A.~Nekrasov, A.~Okounkov and R.~Pandharipande,
    ``Gromov--Witten theory and Donaldson--Thomas theory I,''
    Compos. Math. {\bf 142} (2006) 1263--1285
    [arXiv:math.AG/0312059].

\bibitem{MNOPII}
    D.~Maulik, N.A.~Nekrasov, A.~Okounkov and R.~Pandharipande,
    ``Gromov--Witten theory and Donaldson--Thomas theory II,''
    Compos. Math. {\bf 142} (2006) 1286--1304
    [arXiv:math.AG/0406092].

\bibitem{MOOP}
D.~Maulik, A.~Oblomkov, A.~Okounkov and R.~Pandharipande,
``Gromov--Witten/Donaldson--Thomas correspondence for toric 3-folds,''
arXiv:0809.3976~[math.AG].

\bibitem{Mihailescu}
  M.~Mihailescu, I.Y.~Park and T.A.~Tran,
  ``D-branes as solitons of an $\mathcal{N} = 1$, $D = 10$ noncommutative gauge theory,''
  Phys.\ Rev.\  D {\bf 64} (2001) 046006
  [arXiv:hep-th/0011079].
  %%CITATION = PHRVA,D64,046006;%%

\bibitem{Nakajima1}
    H.~Nakajima,
    ``Moduli spaces of anti-selfdual connections on ALE gravitational instantons,''
    Invent. Math. \textbf{102} (1990) 267--303.

\bibitem{Nakajima2}
    H.~Nakajima,
    ``Instantons on ALE spaces, quiver varieties, and Kac--Moody algebras,''
    Duke Math. J. \textbf{76} (1994) 365--416.

\bibitem{Nakajima}
  H.~Nakajima, ``Lectures on Hilbert schemes of points on surfaces,''
  Univ. Lect. Ser. {\bf 18}, Amer. Math. Soc. (1999).

\bibitem{Nakajima3}
    H.~Nakajima,
    ``Sheaves on ALE spaces and quiver varieties,''
    Moscow Math. J. \textbf{7} (2007) 699--722.

\bibitem{NY}
H.~Nakajima and K.~Yoshioka, 
``Lectures on instanton counting,''
CRM Proc. Lect. Notes {\bf 38} (2004) 31--101
[arXiv:math/0311058~[math.AG]].

\bibitem{NekrasovLectures}
N.A.~Nekrasov, ``Lectures on open strings and noncommutative
gauge fields,'' in: {\sl Gravity, Gauge Theory and Strings},
eds. C.~Bachas, A.~Bilal, M.R.~Douglas, N.A.~Nekrasov and F.~David,
Springer (2002) 477--495 [arXiv:hep-th/0203109].

\bibitem{Nekrasov}
  N.A.~Nekrasov,
  ``Seiberg--Witten prepotential from instanton counting,''
  Adv.\ Theor.\ Math.\ Phys.\  {\bf 7} (2004) 831--864
  [arXiv:hep-th/0206161].
  %%CITATION = 00203,7,831;%%

\bibitem{NekrasovICMP}
  N.A.~Nekrasov,
  ``Localizing gauge theories,''
%\href{http://www.slac.stanford.edu/spires/find/hep/www?irn=6882099}{SPIRES entry}
in: {\sl 14th International Congress on Mathematical
Physics}, ed. J.-C.~Zambrini, World Scientific (2005) 644--652.

\bibitem{NS}
  N.A.~Nekrasov and A.S.~Schwarz,
  ``Instantons on noncommutative $\R^4$ and $(2,0)$ superconformal six-dimensional
  theory,''
  Commun.\ Math.\ Phys.\  {\bf 198} (1998) 689--703
  [arXiv:hep-th/9802068].
  %%CITATION = CMPHA,198,689;%%

\bibitem{ORV}
  A.~Okounkov, N.~Reshetikhin and C.~Vafa,
  ``Quantum Calabi--Yau and classical crystals,''
  Progr. Math. {\bf 244} (2006) 597--618
  [arXiv:hep-th/0309208].
  %%CITATION = HEP-TH/0309208;%%

\bibitem{OSV}
 H.~Ooguri, A.~Strominger and C.~Vafa,
  ``Black hole attractors and the topological string,''
  Phys.\ Rev.\  D {\bf 70} (2004) 106007
  [arXiv:hep-th/0405146].
  %%CITATION = PHRVA,D70,106007;%%

\bibitem{PT}
R.~Pandharipande and R.P.~Thomas,
``Curve counting via stable pairs in the derived category,''
Invent. Math. {\bf 178} (2009) 407--447
[arXiv:0707.2348~[math.AG]]; ``Stable pairs and BPS invariants,''
J. Amer. Math. Soc. {\bf 23} (2010) 267--297
[arXiv:0711.3899~[math.AG]].

\bibitem{PS}
 A.D.~Popov and R.J.~Szabo,
  ``Quiver gauge theory of nonabelian vortices and noncommutative instantons
  in higher dimensions,''
  J.\ Math.\ Phys.\  {\bf 47} (2006) 012306
  [arXiv:hep-th/0504025].
  %%CITATION = JMAPA,47,012306;%%

\bibitem{RS}
  R.M.G.~Reis and R.J.~Szabo,
  ``Geometric K-homology of flat D-branes,''
  Commun.\ Math.\ Phys.\  {\bf 266} (2006) 71--122
  [arXiv:hep-th/0507043].
  %%CITATION = CMPHA,266,71;%%

\bibitem{Shadchin}
  S.~Shadchin,
  ``Saddle point equations in Seiberg--Witten theory,''
  J. High Energy Phys. {\bf 0410} (2004) 033
  [arXiv:hep-th/0408066].
  %%CITATION = JHEPA,0410,033;%%

\bibitem{Szabo1}
R.J.~Szabo,
``Quantum field theory on noncommutative spaces,''
Phys.\ Rept.\ {\bf 378} (2003) 207--299 [arXiv:hep-th/0109162].
%%CITATION = HEP-TH 0109162;%%

\bibitem{Vafa}
  C.~Vafa,
  ``Two-dimensional Yang--Mills, black holes and topological strings,''
  arXiv:hep-th/0406058.
  %%CITATION = HEP-TH/0406058;%%

\bibitem{VW}
  C.~Vafa and E.~Witten,
  ``A strong coupling test of S-duality,''
  Nucl.\ Phys.\  B {\bf 431} (1994) 3--77
  [arXiv:hep-th/9408074].
  %%CITATION = NUPHA,B431,3;%%

\bibitem{Vonk}
  M.~Vonk,
  ``A mini-course on topological strings,''
  arXiv:hep-th/0504147.
  %%CITATION = HEP-TH/0504147;%%

\bibitem{WittenB}
  E.~Witten,
  ``BPS bound states of D0--D6 and D0--D8 systems in a $B$-field,''
  J. High Energy Phys. {\bf 0204} (2002) 012
  [arXiv:hep-th/0012054].
  %%CITATION = JHEPA,0204,012;%%

\bibitem{Wyllard}
N.~Wyllard,
``$A_{N-1}$ conformal Toda field theory correlation functions from conformal $\mathcal{N}=2$ $SU(N)$ quiver gauge theories,''
J. High Energy Phys. {\bf 0911} (2009) 002
[arXiv:0907.2189~[hep-th]].

\end{thebibliography}
\end{document}